\documentclass[11pt]{article}

\usepackage[latin1]{inputenc}

\usepackage{amstext,amsmath,amssymb,amsfonts,bbm}
\usepackage{hyperref}
\usepackage{authblk}
\usepackage{caption} 
\usepackage{epsfig} 
\usepackage{color} 
\usepackage{graphicx} 
\usepackage{braket} 
\usepackage{slashed} 
\usepackage{dsfont} 
\usepackage{amsthm}  
\usepackage{young}
\usepackage{soul}

\allowdisplaybreaks[4]

\usepackage{psfrag}
\usepackage{tikz} 
\usetikzlibrary{arrows}
\usetikzlibrary{plotmarks}
\usetikzlibrary{external}
\usetikzlibrary{patterns}
\usepackage{pgfplots}
\pgfplotsset{compat=1.9}
\usetikzlibrary{patterns}
\usetikzlibrary{calc}
\usepackage[compat=1.1.0]{tikz-feynman}
\usetikzlibrary{automata,positioning}
\usepackage{stmaryrd}

\usepackage{float}  
\usepackage{subfig}
\usepackage[lmargin=50pt,rmargin=60pt,tmargin=60pt,bmargin=65pt]{geometry}
\usepackage{cite}

\newcommand{\be}{\begin{equation}}
\newcommand{\ee}{\end{equation}}

\newcommand{\sgn}{{\rm sgn}}

\newcommand{\cG}{\mathcal{G}}

\allowdisplaybreaks[4]

\newtheorem{remark}{Remark}
\newtheorem{proposition}{Proposition}
\newtheorem{definition}{Definition}
\newtheorem{lemma}{Lemma}
\newtheorem{theorem}{Theorem}

\theoremstyle{remark}

\usepackage[nocfg]{nomencl}
\makenomenclature

\begin{document}

\title{\bf Melonic large $N$ limit of $5$-index irreducible random tensors}

\author[1]{Sylvain Carrozza}
\author[2,3]{Sabine Harribey}

\affil[1]{\normalsize\it 
Institute for Mathematics, Astrophysics and Particle Physics, Radboud University
Heyendaalseweg 135, 6525 AJ Nijmegen, The Netherlands.
 \authorcr  email: s.carrozza@hef.ru.nl
 \authorcr \hfill}
 
 \affil[2]{\normalsize \it 
 CPHT, CNRS, Ecole Polytechnique, Institut Polytechnique de Paris, Route de Saclay, \authorcr 91128 PALAISEAU, 
 France
 \authorcr email:  sabine.harribey@polytechnique.edu
 \authorcr \hfill }

\affil[3]{\normalsize\it 
Heidelberg University, Institut f\"ur Theoretische Physik, Philosophenweg 19, 69120 Heidelberg, Germany
 \authorcr \hfill}

\date{}

\maketitle

\hrule\bigskip

\begin{abstract}
\noindent We demonstrate that random tensors transforming under rank-$5$ irreducible representations of $\mathrm{O}(N)$ can support melonic large $N$ expansions. Our construction is based on models with sextic ($5$-simplex) interaction, which generalize previously studied rank-$3$ models with quartic (tetrahedral) interaction \cite{Benedetti:2017qxl, Carrozza:2018ewt}. Beyond the irreducible character of the representations, our proof relies on recursive bounds derived from a detailed combinatorial analysis of the Feynman graphs. Our results provide further evidence that the melonic limit is a universal feature of irreducible tensor models in arbitrary rank. 
\end{abstract}

\hrule\bigskip

\tableofcontents

\section{Introduction}

In recent years, tensor models have been shown to admit a specific kind of large $N$ limit, known as the \textit{melonic limit} \cite{Bonzom:2011zz,RTM,Klebanov:2018fzb}. Its main appeal is that it is both richer than that of vector models \cite{Guida:1998bx,Moshe:2003xn} and simpler than the planar limit of matrix models \cite{'tHooft:1973jz,Brezin:1977sv,DiFrancesco:1993nw}. As a result, it has proven uniquely valuable as an analytic tool to explore strong coupling effects in many-body quantum physics. Tensor models have for instance found applications to strongly-coupled quantum mechanics \cite{Witten:2016iux,Gurau:2016lzk,Klebanov:2016xxf,Peng:2016mxj,Krishnan:2016bvg,Krishnan:2017lra,Bulycheva:2017ilt,Choudhury:2017tax,Halmagyi:2017leq,Klebanov:2018nfp,Carrozza:2018psc} (see also \cite{Delporte:2018iyf,Klebanov:2018fzb} for reviews), by providing alternatives to the SYK model which do not rely on a disorder average \cite{Sachdev:1992fk, Kitaev2015, Maldacena:2016hyu, Polchinski:2016xgd,Gross:2016kjj}. In higher dimensions, they can be investigated as proper quantum field theories, and have given rise to a new family of conformal field theories known as \textit{melonic CFTs} \cite{Giombi:2017dtl,Prakash:2017hwq,Benedetti:2017fmp,Giombi:2018qgp,Benedetti:2018ghn,Benedetti:2019eyl,Benedetti:2019ikb,Benedetti:2019rja,Lettera:2020uay} (see also \cite{Benedetti:2020seh,Gurau:2019qag} for reviews). 

For such applications, the key feature is the existence of a large $N$ expansion dominated by melon diagrams. When tensor models were first introduced, in zero dimension and in the context of random geometry and quantum gravity \cite{Ambjorn:1990ge,Sasakura:1990fs}, such a limit  was initially lacking. A proper generalization of the genus expansion of matrix models was only discovered later, in the context of so-called \emph{colored tensor models} \cite{Gurau:2009tw,Gurau:2010ba,Gurau:2011xq}. This program, motivated by random geometry in dimension $d \geq 3$, led to different realizations of the melonic limit, all dominated by the same type of tree-like Feynman graph structure (see e.g. \cite{Lionni:2017yvi, Bonzom:2018btd}). So-called \emph{uncolored tensor models} \cite{Bonzom:2012hw} have been particularly well-studied (see \cite{Rivasseau:2013uca, Rivasseau:2014ima,  Eichhorn:2018phj} for reviews), and triggered rigorous developments in the broader context of tensorial group field theory \cite{BenGeloun:2011rc, Samary:2012bw, Carrozza:2012uv, Carrozza:2013wda,Krajewski:2016svb,Rivasseau:2017xbk}. In such models, each index of the tensor transforms independently under its own symmetry group. Indices with different positions in the tensor are therefore not allowed to mix; they carry different \emph{colors}, which significantly constrains the combinatorial structure of the theory. The subclass of uncolored tensor models of interest for large $N$ quantum field theory applications are the ones that generate \emph{bilocal} melonic radiative corrections \cite{Dartois:2013he, Carrozza:2015adg, Ferrari:2017ryl, Ferrari:2017jgw, Benedetti:2020iyz,Prakash:2019zia}, rank-$3$ tensors transforming under the tri-fundamental representation of $\mathrm{O}(N)$ being a simple and popular choice \cite{Carrozza:2015adg}. 

More recently, and in rank $3$, it was understood how to generalize the bilocal melonic limit to ordinary tensor representations of $\mathrm{O}(N)$ and $\mathrm{Sp}(N)$ \cite{Klebanov:2017nlk, Gurau:2017qya, Benedetti:2017qxl, Carrozza:2018ewt, Carrozza:2018psc}, thereby going beyond colored and uncolored models. It might at first appear that completely symmetric rank-$3$ tensor models (such as the ones initially introduced in the nineties \cite{Ambjorn:1990ge,Sasakura:1990fs}) cannot support a bilocal melonic limit. Indeed, they instead support a vector-like (and ultralocal) large $N$ limit \cite{Benedetti:2017qxl}. However, removing the vector modes contained in the traces of the tensor is sufficient to reach a melonic regime, as was initially proposed in \cite{Klebanov:2017nlk} and proven in \cite{Benedetti:2017qxl}. In a similar way, one can conjecture that any \emph{irreducible} tensor representation (of $\mathrm{O}(N)$ or $\mathrm{Sp}(N)$) can support a melonic large $N$ limit, as was proven rigorously in rank-$3$ \cite{Benedetti:2017qxl, Carrozza:2018ewt, Carrozza:2018psc}. Steps have also been taken to extend those results to Hermitian multi-matrix models \cite{Carrozza:2020eaz}, in the spirit of \cite{Ferrari:2017ryl}. 

In the present work, we make an additional contribution to this program, by confirming that $\mathrm{O}(N)$ irreducible tensor representations of rank $5$ also support (bilocal) melonic limits. To this effect, we follow the combinatorial methods developed in \cite{Benedetti:2017qxl}, with suitable adaptations. As in rank $3$, technical difficulties arise from the existence of Feynman graphs which violate the maximum scaling naively allowed by the large $N$ limit. The irreduciblity condition implies that such contributions necessarily cancel out upon resummation, and are ultimately harmless. However, their mere existence greatly complicates the (standard) recursive strategy we use to bound the other Feynman amplitudes. Some aspects of the combinatorial construction we end up with are more elementary than in rank-$3$ (e.g. triangle subgraphs play no role here), but others are much more involved (e.g. in addition to two- and four-point functions, the present analysis involves bounds on eight-point functions).

\subsection{Outline of the results}

We start off with a real $\mathrm{O}(N)$ tensor transforming in one of the seven inequivalent irreducible representations of rank $5$. These consist of the symmetric traceless representation, the antisymmetric representation, and five additional representations with mixed index symmetries. The orthogonal projector on one of these irreducible subspaces defines a (degenerate) covariance $\pmb P$. The resulting Gaussian distribution is then perturbed by a sextic interaction with the combinatorics of a $5$-simplex. 

Our main results are stated in Theorem~\ref{theorem_princ}, establishing the existence of a large $N$ expansion, and Theorem~\ref{theorem_LO}, stating that the leading order graphs are melons. As a guide to the reader, we now outline the main steps of the proof. 

\paragraph{Perturbative expansion.} We first expand the free energy into Feynman amplitudes which, here, are labeled by rooted connected combinatorial maps. We then go to a more detailed representation in terms of stranded graphs. Indeed, each half-edge in a map carries five indices and can therefore be represented with five strands. In turn, each term in the propagator has a specific tensorial structure inducing a particular pairing of the strands propagating along the edge. We can then re-write the perturbative expansion in terms of those stranded graphs. Given that each propagator edge can take up to $945$ configurations, this representation is highly uneconomical. However, the merit of stranded graphs is that their large $N$ asymptotics is transparently encoded into their combinatorial structure. This leads to the following estimate of the amplitude of a stranded graph $G$:
\begin{equation}
\mathcal{A}(G) = K(G) N^{-\omega(G)} \left( 1 + \mathcal{O} (1/N) \right)\,, 
\end{equation}
where $K(G)$ is a non-vanishing rational number independent from $N$, and $\omega$ is the \emph{degree} of the graph (see Eq.~\eqref{eq:degree}).

The degree $\omega$ is an integer quantity. If we could prove it to be bounded from below, the existence of a large $N$ expansion would immediately follow. However, this conjecture is not true in general: stranded graphs with arbitrarily negative degrees do exist. We will therefore need to rely on a subtler strategy: for any map, we will prove that none of its stranded configurations with negative degree (if they exist) actually contribute to the full amplitude.

\paragraph{Subtracting and deleting.} The proof of our main results then proceeds in two steps.
\begin{itemize}
\item We first identify a family of stranded graphs supporting arbitrarily negative degrees. Those problematic graphs happen to be generated by melon and double-tadpole maps. Thanks to the irreducibility of the representation, we can straightforwardly prove that the amplitudes of those maps are in fact well-behaved at large $N$. For convenience, we will subtract them through a partial resummation of the perturbative series, governed by a closed and algebraic Schwinger-Dyson equation. 
\item We then prove that, once the problematic configurations have been subtracted, all the remaining stranded graphs have non-negative degree. This is done by induction on the number of vertices of the graphs, through suitable combinatorial deletions of subgraphs. The important condition that no melon or double-tadpole should be generated by such moves makes the construction rather delicate and technical. 
\end{itemize}

\paragraph{Leading-order.} After proving the existence of the large $N$ expansion, the last step consists in showing that it is dominated by melon diagrams. At this stage, one might be tempted to prove the following improved statement: any stranded graph with no melon and no double-tadpole has in fact strictly positive degree. Again, this is not so simple, as stranded graphs with no melon or tadpole can have vanishing degrees. However, there are again cancellations, and it turns out that none of those configurations can actually contribute to the full amplitudes of their parent maps. We will implicitly account for such cancellations by mean of Cauchy-Schwarz inequalities which, once the large $N$ expansion has been established, can be used to directly bound the full amplitudes of non-melonic Feynman maps (without having to resort to the stranded representation). We will conclude that a Feynman map is leading order if and only if it is melonic.

\paragraph{Plan of the paper} In section \ref{sec:model}, we introduce the models and our main results: Theorem~\ref{theorem_princ} and Theorem~\ref{theorem_LO}. In section \ref{sec:expansion}, we perform the perturbative expansion and define two different types of diagrams: Feynman maps and stranded graphs. We also introduce in more detail the problematic subgraphs that could potentially lead to violations of the maximum scaling in $N$. In section \ref{sec:tmd}, we introduce the important notion of \textit{boundary graph}, as well as various particular subgraphs that will play important roles in the rest of the proof. In section \ref{sec:subtraction}, we perform the explicit subtraction of melons and double-tadpoles. We then arrive to an equivalent theory with renormalized covariance in which melons and double-tadpoles have been subtracted from the Feynman expansion. In section \ref{sec:deletions}, we establish a number of lemmas and propositions enabling the recursive deletion of certain subgraphs, that will be instrumental to ultimately prove Theorem~\ref{theorem_princ}. Finally, in section \ref{sec:LO}, we prove Theorem~\ref{theorem_LO} and show that melons dominate the large $N$ expansion. 
In Appendix~\ref{ap:bounds}, we prove some useful bounds on the number of faces of stranded graphs, while Appendix~\ref{app:particular} provides a proof of Lemma~\ref{lemma:particular_cases} (which handles a number of particular cases). Finally, for the reader's convenience, some of our main definitions and notations are summarized in the following nomenclature.

\printnomenclature
\addcontentsline{toc}{section}{Nomenclature}

\section{The models and the main results}
\label{sec:model}

We consider a real tensor $T_{abcde}$, transforming in the tensor product of five fundamental representations of the orthogonal group $O(N)$ (hence, $a,b,c,d,e=1,\dots,N$).
The action of the symmetric group $\mathcal{S}_5$ and the trace operation allows to decompose $T_{abcde}$ into irreducible components, which are themselves tensors of rank $5$, $3$ or $1$.  
In this work, we will focus on the seven inequivalent representations of rank $5$. They are the traceless representations with index symmetry given by the following Young tableaux:
$$
\begin{Young}  \cr  \cr  \cr \cr  \cr  \end{Young}\quad \begin{Young}  & & & & \cr \end{Young}\quad 
\begin{Young}  & & &  \cr \cr \end{Young}\quad 
\begin{Young}  & &  \cr \cr \cr \end{Young} \quad \begin{Young}  &   \cr \cr \cr \cr\end{Young}\quad \begin{Young}  & &  \cr & \cr \end{Young}\quad 
\begin{Young}  &   \cr & \cr \cr \end{Young}
$$
The first two correspond to the antisymmetric and symmetric traceless representations, respectively, while the other five have mixed index permutation symmetries (that is, they carry representations of $\mathcal{S}_5$ of dimension higher than $2$). Given such a representation, a central object in our construction will be the orthogonal projector on the associated linear subspace of tensors, with respect to the canonical scalar product: $\langle T \vert T' \rangle := T_{abcde} T'_{abcde}$.\footnote{We assume Einstein's summation convention throughout this work.} The kernel of this projector will serve as a degenerate covariance, it is therefore crucial for it to be symmetric. As an illustration, let us find the orthogonal projectors on completely symmetric traceless tensors and completely antisymmetric tensors.

\begin{description}
 \item[\bf Symmetric traceless representation.]
 
Let us start from a completely symmetric tensor $T_{abcde}$. We can decompose $T$ into a traceless part $T^0_{abcde}$, a symmetric traceless tensor of rank tree $T^1_{abc}$, and a vector $T^2_a$:\footnote{Here we use the notation:
\begin{equation*}
V_{(a_1\dots a_n)}=\frac{1}{n!}\sum_{p \in \mathcal{S}_n}V_{a_{p(1)}\dots a_{p(n)}}
\end{equation*}
for $V$ a tensor with $n$ indices. 

In particular, we have 
\begin{equation*}
T^{1}_{(cde}\delta_{ab)}=\frac{1}{10}\left(T^1_{cde}\delta_{ab}+T^1_{bde}\delta_{ac}+T^1_{bce}\delta_{ad}+T^1_{bcd}\delta_{ae}+T^1_{ade}\delta_{bc}
 +T^1_{ace}\delta_{bd}+T^1_{acd}\delta_{be}+T^1_{abe}\delta_{cd}+T^1_{abd}\delta_{ce}+T^{1}_{abc}\delta_{de}\\\right).
\end{equation*}
}
\begin{align*}
T_{abcde}=&T^0_{abcde}+10T^1_{(cde}\delta_{ab)}
+ 30T^2_{(a}\delta_{bc}\delta_{de)}\,.
\end{align*} 
By taking successive traces over pairs of indices, we obtain 
\begin{equation*}
T_{abcdd}=(N+6)T^1_{abc}+2(N+4)(T^2_{a}\delta_{bc}+T^2_{b}\delta_{ac}+T^2_{c}\delta_{ab})
\end{equation*}
and
\begin{equation*}
T_{abbdd}=2(N+2)(N+4)T^2_{a}\,,
\end{equation*}
which, combined, leads to the following expressions for $T^1$ and $T^2$:
\begin{align*}
T^1_{abc}&=\frac{1}{N+6}\left(T_{abcdd}-\frac{1}{(N+2)}\left(T_{addee}\delta_{bc}+T_{bddee}\delta_{ac}+T_{cddee}\delta_{ab}\right)\right) \,, \\ 
T^2_a &= \frac{1}{2(N+2)(N+4)} T_{abbdd}\,.
\end{align*}
This allows us to define a projector on symmetric traceless tensors as the projector on symmetric tensors minus the projector on traces. We find:
\begin{align}
\pmb S_{\pmb a,\pmb b}=\frac{1}{5!}&\left[\sum_{\sigma \in \mathcal{S}_5}\prod_{i=1}^5\delta_{a_ib_{\sigma(j)}}-\frac{2}{N+6}\sum_{\substack{\{i_1,i_2,i_3\} \cup\{i_4,i_5\} \\ = \llbracket 1,5 \rrbracket}}\sum_{\substack{\{j_1,j_2,j_3\}\cup\{j_4,j_5\}\\ = \llbracket 1,5 \rrbracket}}\delta_{a_{i_4}a_{i_5}}\delta_{b_{j_4}b_{j_5}}\sum_{\sigma \in \mathcal{S}_3}\prod_{k=1}^3\delta_{a_{i_k}b_{j_{\sigma(k)}}}\right. \crcr
&\left.+\frac{2}{(N+4)(N+6)}\sum_{\substack{\{i_1\}\cup\{i_2,i_3\}\cup\{i_4,i_5\} \\ =\llbracket 1,5 \rrbracket}}\sum_{\substack{\{j_1\}\cup\{j_2,j_3\}\cup\{j_4,j_5\} \\= \llbracket 1,5 \rrbracket}}\delta_{a_{i_1}b_{j_1}}\delta_{a_{i_2}a_{i_3}}\delta_{a_{i_4}a_{i_5}}\delta_{b_{j_2}b_{j_3}}\delta_{b_{j_4}b_{j_5}}\right] \,,
\label{sym}
\end{align}
where we use the short-hand notation $\pmb a=(a_1,a_2,a_3,a_4,a_5)$, $\pmb b=(b_1,b_2,b_3,b_4,b_5)$ (and so on).

Moreover, one can readily check that $\pmb S_{\pmb a,\pmb b} = \pmb S_{\pmb a,\pmb b}$, so that $\pmb S$ is the looked-for orthogonal projector.

\item[\bf Antisymmetric representation.]

The orthogonal projector on completely antisymmetric tensors takes the form:
\begin{equation}\label{antisym}
\pmb A_{\pmb a,\pmb b}=\frac{1}{5!}\sum_{\sigma \in \mathcal{S}_5}\epsilon(\sigma)\prod_{i=1}^5\delta_{a_ib_{\sigma(j)}} 
\end{equation}
\end{description}

A covariance for the other five inequivalent irreducible representations can be obtained in a similar fashion. As it turns out, the explicit form of this projector is not necessary for our proofs to go through, so we only provide a brief sketch of the general construction. As a first step, one can construct the canonical projector associated to the target Young tableau, by first symmetrizing over indices appearing in a same row, then antisymmetrizing over indices appearing in a same column. After projecting out the trace components, one obtains a projector with the desired image. However, in contrast to what happened with the completely symmetric and symmetric traceless representations, this first projector will not in general be orthogonal. If so, one needs to orthogonalize it as a last step in the construction.  

\

The generic tensor model with $5$-simplex (or complete graph) interaction is defined by the action:
\begin{align}
S(T)=&\frac{1}{2}\sum_{a_1, \, \ldots, \, a_5}T_{a_1a_2a_3a_4a_5}T_{a_1a_2a_3a_4a_5}\crcr
& -\frac{\lambda}{6N^5}\sum_{a_1, \, \ldots, \, a_{15}}T_{a_1a_2a_3a_4a_5}T_{a_5a_6a_7a_8a_9}T_{a_{9}a_{4}a_{10}a_{11}a_{12}}T_{a_{12}a_{8}a_{3}a_{13}a_{14}}T_{a_{14}a_{11}a_{7}a_{2}a_{15}}T_{a_{15}a_{13}a_{10}a_{6}a_{1}}
\end{align}
We will denote the $5$-simplex pattern of contraction by
\begin{equation}\label{eq:5-simplex}
\delta^h_{\pmb a\pmb b\pmb c\pmb d\pmb e\pmb f}=\delta_{a_1f_5}\delta_{a_2e_4}\delta_{a_3d_3}\delta_{a_4c_2}\delta_{a_5b_1}\delta_{b_2f_4}\delta_{b_3e_3}\delta_{b_4d_2}\delta_{b_5c_1}\delta_{c_3f_3}\delta_{c_4e_2}\delta_{c_5d_1}\delta_{d_4f_2}\delta_{d_5e_1}\delta_{e_5f_1}\,,
\end{equation}
such that the action can be simplified to:
\begin{align}
S(T)=&\frac{1}{2}T_{\pmb a}T_{\pmb a} -\frac{\lambda}{6N^5}\delta^h_{\pmb a\pmb b\pmb c\pmb d\pmb e\pmb f}T_{\pmb a}T_{\pmb b}T_{\pmb c}T_{\pmb d}T_{\pmb e}T_{\pmb f}\,.
\end{align}

We also denote by $\pmb 1$ the identity operator $\pmb 1_{\pmb a, \pmb b}=\pmb 1_{a_1a_2a_3a_4a_5,b_1b_2b_3b_4b_5}=\prod_{i=1}^5\delta_{a_ib_i}=\delta_{\pmb a\pmb b}$, and by $\partial_T$ the tensor of derivative operators $(\partial_T)_{\pmb a}\equiv \frac{\partial}{\partial T_{\pmb a}}$.
With these notations at hand, we can write the partition function, the free energy and its first derivative as:
\begin{align}
Z_{\pmb 1}(\lambda)&=\int [dT]e^{-S(T)}=\left[e^{\frac{1}{2}\partial_T\pmb 1 \partial_T}e^{\frac{\lambda}{6N^5}\delta^h_{\pmb a\pmb b\pmb c\pmb d\pmb e\pmb f}T_{\pmb a}T_{\pmb b}T_{\pmb c}T_{\pmb d}T_{\pmb e}T_{\pmb f}}\right]_{T=0} ~,~\crcr
\ln Z_{\pmb 1}(\lambda)&=\ln \lbrace \int [dT]e^{-S(T)} \rbrace , ~~F_{\pmb 1}(\lambda)=\frac{6}{N^{5}}\lambda\partial_{\lambda}\ln Z_{\pmb 1}(\lambda) \,.
\end{align}
In the following, $\pmb P$ will denote any one of the seven orthogonal projectors on irreducible rank-$5$ tensor representations. We will sometimes illustrate our calculations with $\pmb P =\pmb A$ or $\pmb S$, but our main results hold for any irreducible representation. The irreducible tensor model of interest can be obtained from the generic model by disallowing the propagation of modes which are in the kernel of $\pmb P$. This in turn amounts to replacing the non-degenerate covariance $\pmb 1$ by the degenerate covariance $\pmb P$:
\begin{equation}
F_{\pmb P}(\lambda)=\frac{6}{N^{5}}\lambda\partial_{\lambda}\ln \left\lbrace \left[e^{\frac{1}{2}\partial_T \pmb P \partial_T}e^{\frac{\lambda}{6N^5}\delta^h_{\pmb a\pmb b\pmb c\pmb d\pmb e\pmb f}T_{\pmb a}T_{\pmb b}T_{\pmb c}T_{\pmb d}T_{\pmb e}T_{\pmb f}}\right]_{T=0}\right\rbrace \; . 
\label{eq:model1}
\end{equation}
Note that, in this equation, the tensor $T$ has no symmetry property under permutation of its indices. However, as only the projected modes $\pmb P T$ propagate, we can equivalently change variables to $P=\pmb P T$ as done in \cite{Benedetti:2017qxl}. We can then write: 
\begin{align}
F_{\pmb P}(\lambda)&=\frac{6}{N^{5}}\lambda\partial_{\lambda}\ln \left\lbrace \left[e^{\frac{1}{2}\partial_P \pmb P \partial_P}e^{\frac{\lambda}{6N^5}\delta^h_{\pmb a\pmb b\pmb c\pmb d\pmb e\pmb f}P_{\pmb a}P_{\pmb b}P_{\pmb c}P_{\pmb d}P_{\pmb e}P_{\pmb f}}\right]_{P=0}\right\rbrace \; , \crcr
\frac{\partial}{\partial P_{\pmb a}}P_{\pmb b} & \equiv \pmb P_{\pmb a, \pmb b} \; ,
\label{eq:model}
\end{align}
where the tensor $P$ is in the image of $\pmb P$ and thereby irreducible, and the second line is a definition. The factor $6/N^5$ is for later convenience; it will in particular ensure that $F_{{\pmb P}}$ is an order $1$ quantity in the large $N$ limit.  

\subsection{Main theorems}

The main result of this paper is the existence of a $1/N$ expansion for all seven irreducible rank-$5$ tensor models with complete graph interaction. It is given by the following theorem.

\begin{theorem}
We have (in the sense of perturbation series):
\begin{equation}\label{eq:largeN}
F_{\pmb P}(\lambda)=\sum_{\omega \in \mathbb{N}} N^{-\omega}F_{\pmb P}^{(\omega)}(\lambda)\,.
\end{equation}
\label{theorem_princ}
\end{theorem}
\begin{proof}
This follows from Eq.~\eqref{eq:model_subtracted}, Remark~\ref{rem:broken_unbroken} and Proposition~\ref{prop:positive_degree}. 
\end{proof}

In section \ref{sec:LO}, we further prove that these models are dominated by melon diagrams (which we introduce in section \ref{sec:tmd}). This is given by the following theorem.

\begin{theorem}
In equation \eqref{eq:largeN}, the leading order contribution $F_{\pmb P}^{(0)}(\lambda)$ is a sum over melonic stranded graphs. For small enough $\lambda$, it is the unique continuous solution of the polynomial equation
\begin{equation}
1- X + m_{\pmb P} \lambda^2 X^6  = 0
\end{equation} 
such that $F_{\pmb P}^{(0)}(0)=1$, and where $m_{\pmb P}$ is a model-specific real constant. In particular, $m_{\pmb S} = m_{\pmb A} = \left( \frac{1}{5!} \right)^4$.
\label{theorem_LO}
\end{theorem}

\begin{proof}
This follows from Proposition~\ref{propo:LO}, as well as Eqs.~\eqref{eq:m} and \eqref{eq:sde} in section~\ref{sec:subtraction}.
\end{proof}

\medskip

For completeness, we note that we could consider other $5$-simplex interactions, that differ from our choice in Eq.~\eqref{eq:5-simplex} by a permutation of the strands on each half-edge. Namely, in general, we could introduce the modified kernel:
\begin{equation}\label{eq:modified_5-simplex}
\tilde{\delta}^{h}_{\pmb a\pmb b\pmb c\pmb d\pmb e\pmb f} = \delta^{h}_{(\sigma_1 \cdot \pmb a ) (\sigma_2 \cdot \pmb b) (\sigma_3 \cdot \pmb c) (\sigma_4 \cdot \pmb d) (\sigma_5 \cdot \pmb e)   (\sigma_6 \cdot\pmb f) }
\end{equation}
where $\{ \sigma_k \}$ are permutations in $\mathcal{S}_5$, and $\cdot$ denotes the natural action of $\mathcal{S}_5$ on a $5$-tuple. If $\pmb P = \pmb S$ or $\pmb A$, any two such choices differ at most by a sign, and are therefore equivalent. A priori, this is not necessarily so for other irreducible representations, since permuting two indices which are neither in a same column nor in a same row of the Young tableau involves non-trivial linear combinations of tensors. We leave the evaluation of the dimension of the space of $5$-simplex invariants for each irreducible representation to future work. However, we note that our main theorems remain valid for any such interaction, and in fact any linear combination thereof. The only reason we decide to focus exclusively on the kernel of Eq.~\eqref{eq:5-simplex} is to keep the combinatorial structure of the Feynman diagrams (see the next section) as elementary as possible. Indeed, the modified $5$-simplex of Eq.~\eqref{eq:modified_5-simplex} is in general not symmetric under cyclic permutation of its half-edges, and would therefore require the introduction of vertices with marked half-edges in the Feynman rules. A particularly interesting example of such a non-cyclic kernel is 
\begin{align}\label{eq:colored_5-simplex}
\tilde{\delta}^{h}_{\pmb a\pmb b\pmb c\pmb d\pmb e\pmb f} &= \delta^{h}_{\pmb a  (\sigma \cdot \pmb b) (\sigma^2 \cdot \pmb c) (\sigma^3 \cdot \pmb d) (\sigma^4 \cdot \pmb e)   (\sigma^5 \cdot\pmb f) } \\
&= \delta_{a_1f_1}\delta_{a_2e_2}\delta_{a_3 d_3}\delta_{a_4c_4}\delta_{a_5b_5}\delta_{b_3f_3}\delta_{b_4e_4}\delta_{b_2d_2}\delta_{b_1c_1}\delta_{c_2f_2}\delta_{c_3e_3}\delta_{c_5d_5}\delta_{d_4f_4}\delta_{d_1e_1}\delta_{e_5f_5} \,, \nonumber
\end{align}
where $\sigma = (15)(234)$. A specificity of this pattern of contractions is that every tensor index in position $k$ is contracted with another tensor index in position $k$, and is known as a \emph{colorable} interaction in the random tensor literature. Up to a permutation of the half-edges and of a global permutation of the tensor indices, the kernel of Eq.~\eqref{eq:colored_5-simplex} is in fact the unique colorable $5$-simplex interaction \cite{Ferrari:2017jgw}. With this choice of interaction, it is actually possible to prove a slightly improved version of Theorem~\ref{theorem_princ}, guaranteeing that $m_{\pmb P} > 0$ for \emph{any} irreducible representation $\pmb P$.\footnote{This stems from the fact that we can prove a slightly improved version of Lemma~\ref{lemma:melon_tadpole}, see Remark~\ref{rem:improved_lemma}.} In particular, this observation implies that the colorable interaction \eqref{eq:colored_5-simplex} is non-vanishing for \emph{any} irreducible representation, and therefore, that our results have non-trivial implications for any choice of irreducible propogator.\footnote{From our current understanding, we cannot guarantee that $m_{\pmb P}$ is necessarily non-vanishing with the cyclic vertex \eqref{eq:5-simplex} (unless $\pmb P = \pmb S$ or $\pmb A$). In particular, we cannot exclude the possibility that this vertex might identically vanish for some specific choice of representation $\pmb P$.} While straightforward, we leave the detailed treatment of non-cyclic vertices, as well as the general proof that $m_{\pmb P} > 0$ in the case of a colorable interaction to the interested reader.

\section{Perturbative expansion}
\label{sec:expansion}

\subsection{Feynman maps}\label{sec:feynman_maps}

Given the structure of the propagator $\pmb P$, which is in general not invariant under index permutations, it will be convenient to view the Feynman expansion as a weighted sum of \emph{combinatorial maps} (or \emph{embedded graphs}) rather than ordinary graphs. Even though combinatorial maps always provide a natural way of representing Wick contractions, they are often dispensed with in field theory because, in many instances, the Feynman amplitudes themselves only depend on the graph structure. In our context, this will remain true in representations such as the symmetric traceless or antisymmetric ones, but not in general
\cite{Carrozza:2018ewt}.\footnote{\label{ft:vertex_inv} Under permutation of the first and second half-edges of the vertex (and similarly for any other pair), we find that
\begin{equation*}
\delta^{h}_{\pmb b  \pmb a \pmb c \pmb d \pmb e \pmb f} = \delta^{h}_{(\gamma \cdot \pmb a ) (\gamma^{-1} \cdot \pmb b) ((12) \cdot \pmb c) ((23) \cdot \pmb d) ((34) \cdot \pmb e)   ((45) \cdot\pmb f) } \, , 
\end{equation*} 
where $\gamma = (12345)$. The product of the signatures of the permutations appearing on the right-hand side being even, the invariance of the vertex under permutation of its half-edges follows for $\pmb P \in \{ \pmb A ,\pmb S\}$.} We therefore resort to the language of combinatorial maps.

There are three steps to obtain the perturbative expansion of $F_{\pmb P}$. First, we Taylor expand in $\lambda$ and compute the Gaussian integrals. This leads to a sum over six-valent \emph{combinatorial maps}. We then take the logarithm, which results in a sum over only \emph{connected} combinatorial maps. Finally, we apply the operator $6\lambda \partial_{\lambda}$, which leads to \emph{rooted} connected combinatorial maps. We call a rooted map a map with a half-edge on a vertex marked with an incoming arrow. 
\nomenclature[1]{$\cG$}{\emph{Feynman map}: connected $6$-regular combinatorial map; see Sec.~\ref{sec:feynman_maps}.}

At first order in $\lambda$, $F_{\pmb P}$ corresponds to $\frac{6\times5}{2} = 15$ rooted, connected, combinatorial maps. Contrary to non-rooted maps, unlabeled rooted maps $\mathcal{M}$ come with a combinatorial weight $1$. This is why we chose to study $F_{\pmb P}$ instead of $\ln Z_{\pmb P}(\lambda)$.

\begin{figure}[htbp]
\centering
\captionsetup[subfigure]{labelformat=empty}\subfloat[]{\includegraphics[scale=1]{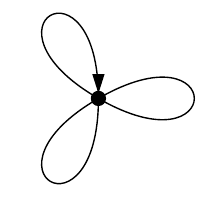}}
\hspace{1cm}
\subfloat[]{\includegraphics[scale=1]{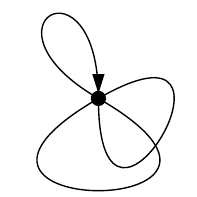}}
\hspace{1cm}
\subfloat[]{\includegraphics[scale=1]{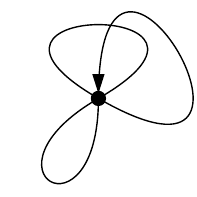}}
\caption{Three of the fifteen first order contributions to $F_{\pmb P }(\lambda)$.}
\end{figure}

We can then write $F_{\pmb P}(\lambda)$ as:
\begin{equation}
F_{\pmb P}(\lambda)=\sum_{\mathcal{M} ~ \mathrm{connected},~ \mathrm{rooted}} \lambda^{V(\mathcal{M})}\mathcal{A}(\mathcal{M})
\end{equation}
with $V(\mathcal{M})$ the number of vertices of $\mathcal{M}$.

\subsection{Stranded graphs}\label{sec:stranded_graphs}

We will now go from this representation to a more detailed one in terms of \emph{stranded graphs} $G$. Indeed, each half-edge in a map $\mathcal{M}$ carries five indices and can therefore be represented by five strands. In turn, each term in the propagator has a specific tensor structure, which induces a particular pairing of the strands being propagated along an edge. Since there are $q=10$ half-strands to be paired along a propagator, there are $(2 q -1)!! = 945$ such tensor structures, all of which appear in the symmetric traceless propagator \eqref{sym}. A \emph{stranded graph} is a combinatorial map, together with a choice of one such tensor structure per edge. As a result, a combinatorial map $\cG$ with $E$ edges gives rise to $945^{E}$ stranded graphs, which we will sometimes call \emph{stranded configurations of $\cG$}. Note that, depending on the model, only a subset of those stranded configurations may be relevant. This is clear from the expression of $\pmb A$ in \eqref{antisym}, which only features $5!=120$ of the $945$ possible tensor structures of the propagator.
\nomenclature[2]{$G$}{\emph{Stranded graph}: Feynman map, together with a choice of one tensor structure per edge; see Sec.~\ref{sec:stranded_graphs}.}

We will distinguish three types of edge configurations:
\begin{itemize}
\item In an \emph{unbroken edge}, all the strands are traversing and connecting half-strands at the two ends 
\begin{figure}[htbp]
\centering
\includegraphics[scale=1]{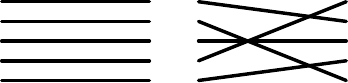}
\caption{Two examples of unbroken edges (out of $5! = 120$).}
\label{fig:unbroken}
\end{figure}
\item In a \emph{broken edge} a pair of half-strands is connected at each end of the edge, and the three other strands are traversing
\begin{figure}[htbp]
\centering
\includegraphics[scale=1]{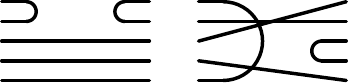}
\caption{Two examples of simply-broken edges (out of ${5 \choose 2}^2 \times 3 ! =600$).}
\label{fig:broken}
\end{figure}
\item In a \emph{doubly-broken edge}, two pairs of half-strands are connected at each end of the edge, and the fifth strand is traversing
\begin{figure}[htbp]
\centering
\includegraphics[scale=1]{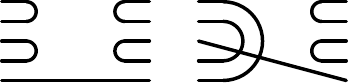}
\caption{Two examples of doubly-broken edges (out of $\left[\frac{1}{2}{5 \choose 2}{3 \choose 2}\right]^2 = 225$).}
\label{fig:doubly_broken}
\end{figure}
\end{itemize}
\nomenclature[4]{\textbf{Unbroken edge.}}{Edge of a stranded graph with all strands traversing; see Fig.~\ref{fig:unbroken}.}
\nomenclature[5]{\textbf{Broken edge.}}{Edge of a stranded graph with exactly three traversing strands; see Fig.~\ref{fig:broken}.}
\nomenclature[6]{\textbf{Doubly-broken edge.}}{Edge of a stranded graph with exactly one traversing strand; see Fig.~\ref{fig:doubly_broken}.}
In particular, the $5!$ tensor structures common to $\pmb A$ and $\pmb S$ lead to unbroken edges. Furthermore, in $\pmb S$, the $600$ terms proportional to $\frac{1}{N+6}$ are associated to  broken edges, while the $225$ terms proportional to $\frac{1}{(N+4)(N+6)}$ lead to doubly-broken edges. Moreover, the large $N$ scaling of each type of edge is universal:  for any choice of propagator $\pmb P$, unbroken tensor structures appear with a coefficient of order one, while broken (resp. doubly-broken) contributions are rescaled by factors of order $1/N$ (resp. $1/N^2$).

\begin{figure}[htbp]
\centering
\includegraphics[scale=0.5]{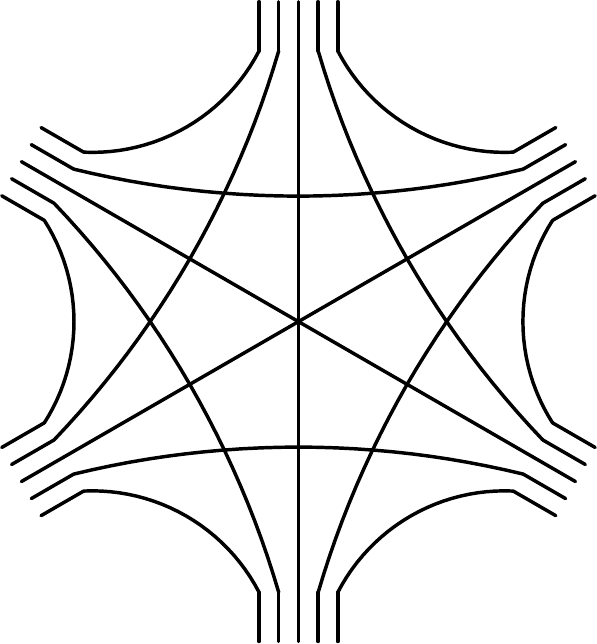}
\caption{Stranded graph representation of the interaction vertex.}\label{fig:vertex}
\end{figure}

We now turn to the stranded representation of the interaction vertex. We call each pair of indices contracted in the $5$-simplex interaction a \emph{corner}. The whole pattern of contractions is represented as a six-valent vertex with fifteen corners, as shown in Fig.~\ref{fig:vertex}. The vertices are then combined with the stranded edges to form a complete stranded diagram. A closed cycle of strands in such a diagram is called a \emph{face}. Finally, we will respectively denote by $F(G)$, $U(G)$, $B_1(G)$ and $B_2(G)$ the number of faces, unbroken edges, simply-broken edges and doubly-broken edges of $G$.

With these definitions in place, we can write the amplitude of a Feynman map as a sum of amplitudes of its standed configurations and thus recast the perturbative expansion in terms of stranded graphs: 
\begin{align*}
F_{\pmb P}(\lambda)&=\sum_{G~\text{connected, rooted}} \lambda^{V(G)}\mathcal{A}(G) \,,
\end{align*}
where $\mathcal{A}(G)$ is the amplitude of the stranded graph $G$. A key advantage is that the large $N$ behaviour of $\mathcal{A}(G)$ is explicitly encoded in the stranded structure of $G$. By inspection of the Feynman rules, each vertex contributes a scaling factor $N^{-5}$, while each broken (resp. doubly-broken) propagator is weighted by a factor $N^{-1}$ (resp. $N^{-2}$) relative to unbroken propagators. Moreover, after contracting the Kronecker delta functions entering the definition of the propagator and vertex kernels, one is left with one free sum and therefore one factor of $N$ per face. This leads to the following large $N$ asymptotics of the amplitudes:
\begin{equation}\label{eq:ampli}
\mathcal{A}(G) = K(G) N^{-\omega(G)} \left( 1 + \mathcal{O} (1/N) \right)\,, 
\end{equation}
where $K(G)$ is a non-vanishing rational number independent from $N$, and the \emph{degree} $\omega$ of the stranded graph $G$ is:\footnote{Note that the first term in \eqref{eq:degree} reflects the (conventional) $N^{-5}$ scaling introduced in the definition of $F_{\bf P}$; see Eq.~\eqref{eq:model}.}
\begin{equation}\label{eq:degree}
\omega(G)=5+5V(G)+B_1(G)+2B_2(G)-F(G)\,.
\end{equation}
For a given choice of $\pmb P$, we can work out an exact formula for $\mathcal{A}(G)$. For instance, when $\pmb P = \pmb A$ or $\pmb S$, we find: 
\begin{equation}
\mathcal{A}(G)=\left(\frac{\varepsilon(G)2^{B_1(G)}2^{B_2(G)}}{5!^{U(G)+B_1(G)+B_2(G)}\left(1+\frac{6}{N}\right)^{B_1(G)+B_2(G)}\left(1+\frac{4}{N}\right)^{B_2(G)}}\right)N^{-\omega(G)}
\label{eq:ampliS}
\end{equation}
where $\varepsilon(G)=(-1)^{B_1(G)}\prod_{e\in G\text{ unbroken}}\epsilon(\sigma^e)$, $\sigma^e$ is the permutation associated to the unbroken edge $e$, and $\epsilon = 1$ (when $\pmb P = \pmb S$) or $\sgn$ (when $\pmb P = \pmb A$).   

The degree $\omega$ is an integer quantity, which can \emph{a priori} take arbitrarily negative values. If one were able to prove it to be bounded from below, the existence of a large $N$ expansion would immediately follow. We will see that this is not true in general: stranded graphs with arbitrarily negative degrees do exist. However, for any map $\cG$, we will prove that none of its stranded configurations $G$ with $\omega(G) < 0$ (if they exist) actually contribute to the full amplitude $\mathcal{A}(\cG)$.

The stranded graphs and combinatorial maps appearing in the rest of the paper will always be assumed to be connected, unless specified otherwise.  

\subsection{Problematic cases}

Consider a stranded graph $G$, and let us simplify the expression of its degree. We denote by $F_p$ the number of \emph{faces of length $p$}, that is, the number of faces that have exactly $p$ corners. Each vertex contributing exactly $15$ corners to the graph, we have the relation: 
\begin{equation}
15V=\sum_{p\geq 1} pF_p\,.
\end{equation}
Together with $F=\sum_{p\geq 1} F_p$, this leads to the following expression for the degree:
\begin{equation}
\omega =5+B_1 +2B_2 +\sum_{p \geq 1} F_p \left(\frac{p}{3}-1\right)\,.
\label{eq:degree_simplified}
\end{equation}

We thus obtain the elementary but important proposition:
\begin{proposition}
Let $G$ be a stranded graph. If $F_1 (G) = F_2 (G) = 0$, then 
\begin{equation}
\omega(G) \geq 0 \,.
\end{equation}
\end{proposition} 
\begin{proof}
This immediately follows from \eqref{eq:degree_simplified}: the only terms that are not explicitly non-negative are proportional to $F_1$ and $F_2$.
\end{proof}

In Section \ref{sec:deletions}, we will prove that an even larger class of stranded graphs have non-negative degrees. As we will be working by induction on the number of vertices, it is convenient to introduce the notion of \emph{ring graph}, defined as a stranded graph with a single edge closed onto itself, and therefore, no vertex (see Fig.~\ref{fig:ring_graph}). The degree of a ring graph is then defined by equation \eqref{eq:degree_simplified}, and is clearly non-negative.  

\begin{figure}[htbp]
\centering
\includegraphics[scale=.8]{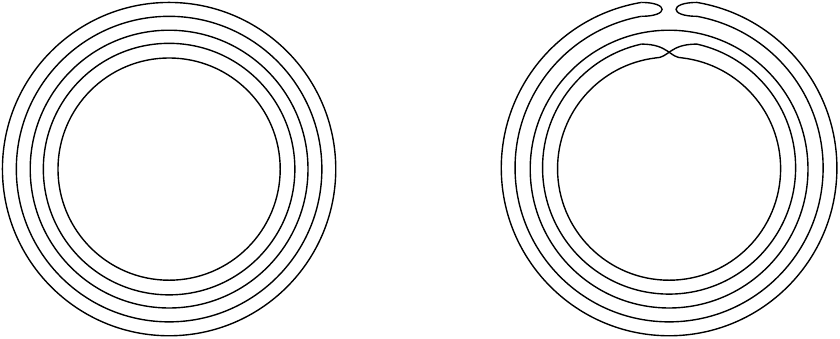}
\caption{Two stranded graphs associated to a ring map, in unbroken (left) and simply-broken (right) configurations.}
\label{fig:ring_graph}
\end{figure}

We finally introduce some nomenclature.
\begin{definition}
A \emph{short face} is a face of length one or two.
An \emph{end graph} is either a graph with no short face, or a ring graph. End graphs have non-negative degrees.
\end{definition}
We have reduced our problem significantly. The rest of the paper is dedicated to the analysis of graphs containing short faces and vertices. 

\section{Combinatorial structure of subgraphs with short faces}
\label{sec:tmd}

The purpose of this section is to introduce specific submap and subgraph structures which may support short faces, and will therefore require special attention. Before that, we also introduce the general notion of boundary graph, which conveniently captures the relation between external legs and external faces of a stranded subgraph.  

\subsection{Boundary graph}\label{sec:bdy_graph}

To any $n$-point stranded graph $G$, we associate a canonically constructed $5$-regular graph with $n$-vertices $G_\partial$, which we call the \emph{boundary graph of $G$} \cite{Gurau:2009tz}. It is constructed in such a way as to faithfully represent the tensorial structure of the correlator $G$ contributes to, up to permutations of the indices appearing in a same tensor. 

More precisely, we define $G_\partial$ through the following procedure. First, each external leg of $G$ is represented in $G_\partial$ by a $5$-valent vertex. Then, for every external strand connecting two external legs of $G$, we draw an edge between the corresponding vertices in $G_\partial$. 
For example, the stranded six-point graph consisting in a single interaction vertex has for boundary graph the complete graph on six vertices $K_6$, as represented in Fig.~\ref{fig:vertex_bndy}. Graphically, one can obtain $G_\partial$ from $G$ by deleting all its internal faces, and pinching its external legs to form vertices as represented in Fig.~\ref{fig:ex_stranded_bndy}. Finally, insofar as the external legs of $G$ are labeled, we will consider $G_\partial$ as a labeled graph. 
\nomenclature[3]{$G_\partial$}{\emph{Boundary graph} of $G$; see Sec.~\ref{sec:bdy_graph}.}

\begin{figure}[htbp]
\centering
\includegraphics[scale=0.3]{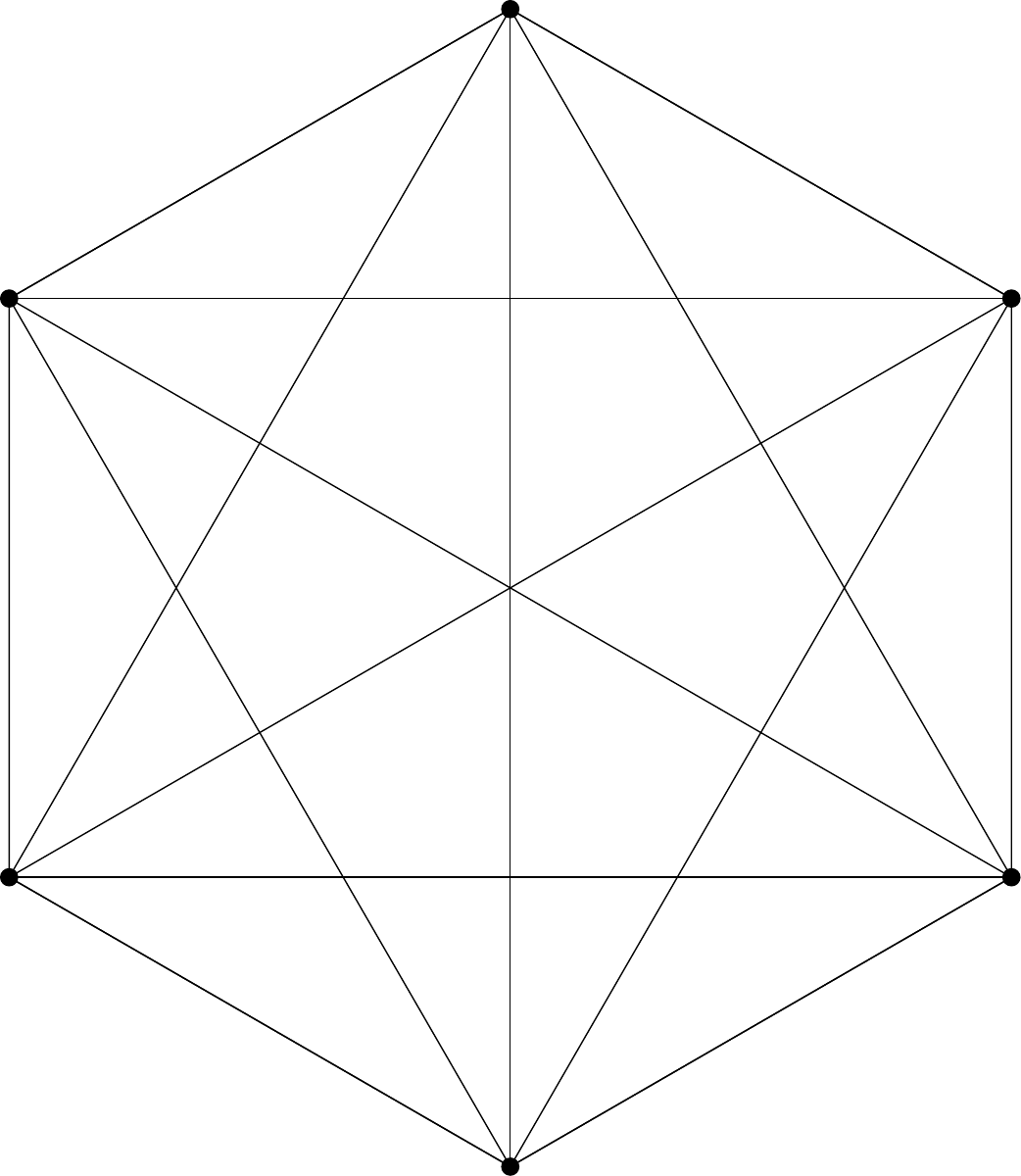}
\caption{The boundary graph of the stranded interaction vertex is the complete graph $K_6$. }
\label{fig:vertex_bndy}
\end{figure}  

\begin{figure}[htbp]
\centering
\includegraphics[scale=1.5]{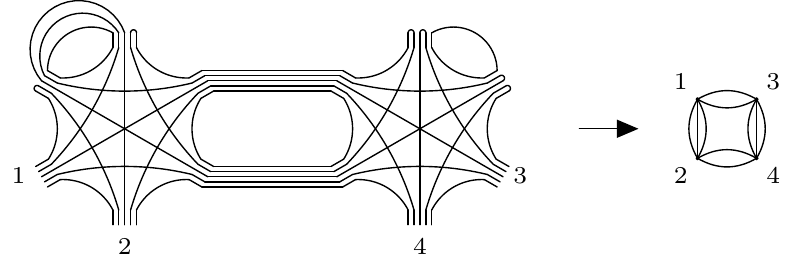}
\caption{A four-point stranded graph and its corresponding boundary graph.}
\label{fig:ex_stranded_bndy}
\end{figure}

\subsection{Faces of length one: tadpoles}

A stranded graph can only have faces of length one if its parent map contains tadpole lines. For convenience, we will distinguish two types of elementary tadpole submaps or subgraphs.  
\begin{definition}
A \emph{single-tadpole} (or, equivalently, a \emph{tadpole}) is a four-point Feynman map or stranded graph with one vertex and one self-loop. A \emph{double-tadpole} is a two-point Feynman map or stranded graph with one vertex and two self-loops. See Fig.~\ref{fig:simp_db_tadpole}.
\end{definition}
\nomenclature[7]{\textbf{Single-tadpole.}}{Four-point Feynman map or stranded graph with one vertex and one self-loop; see Fig.~\ref{fig:simp_db_tadpole}.}
\nomenclature[8]{\textbf{Double-tadpole.}}{Two-point Feynman map or stranded graph with one vertex and two self-loops; see Fig.~\ref{fig:simp_db_tadpole}.}

\begin{figure}[htbp]
\centering
\captionsetup[subfigure]{labelformat=empty}
\subfloat[]{\includegraphics[scale=1]{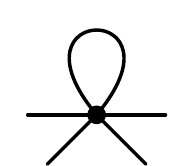}}
\hspace{1cm}
\subfloat[]{\includegraphics[scale=1]{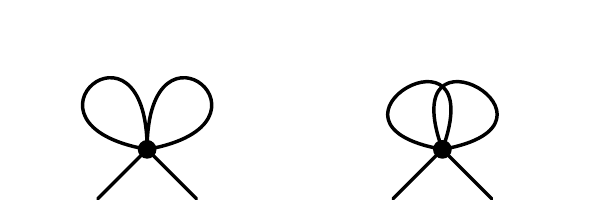}}
\caption{The unique single-tadpole map (left); and two examples of double-tadpole maps that differ through their embedding (right).}
\label{fig:simp_db_tadpole}
\end{figure}

By extension, a graph or map obtained from a single-tadpole (resp.  double-tadpole) by substituting internal edges with non-trivial two-point submaps or subgraphs will be called a \emph{generalized single-tadpole} (resp. \emph{generalized double-tadpole}).

\paragraph{Bad double-tadpoles.}

A double-tadpole graph has at most four internal faces: two of length one, and two of length two. Taking the factor $N^{-5}$ from the vertex into account, we conclude that a double-tadpole graph scales at most like $N^{-1}$. At first sight, it therefore seems that double-tadpoles cannot lead to graphs with negative degrees. But this conclusion is not warranted, which we can illustrate by considering the stranded graph represented in Fig.~\ref{fig:double_tadpole}. Such a configuration supports four faces and hence saturates our scaling bound. Furthermore, its boundary graph is that of a doubly-broken edge. We call any such configuration a \emph{bad double-tadpole}. 

\begin{figure}[htbp]
\centering
\captionsetup[subfigure]{labelformat=empty}
\subfloat[]{\includegraphics[scale=1]{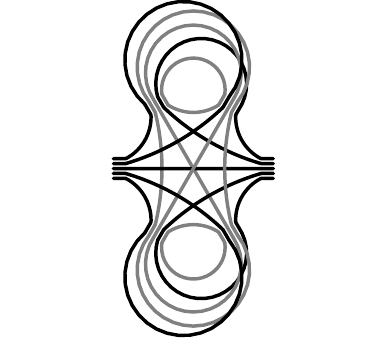}}
\hspace{1cm}
\subfloat[]{\includegraphics[scale=0.25]{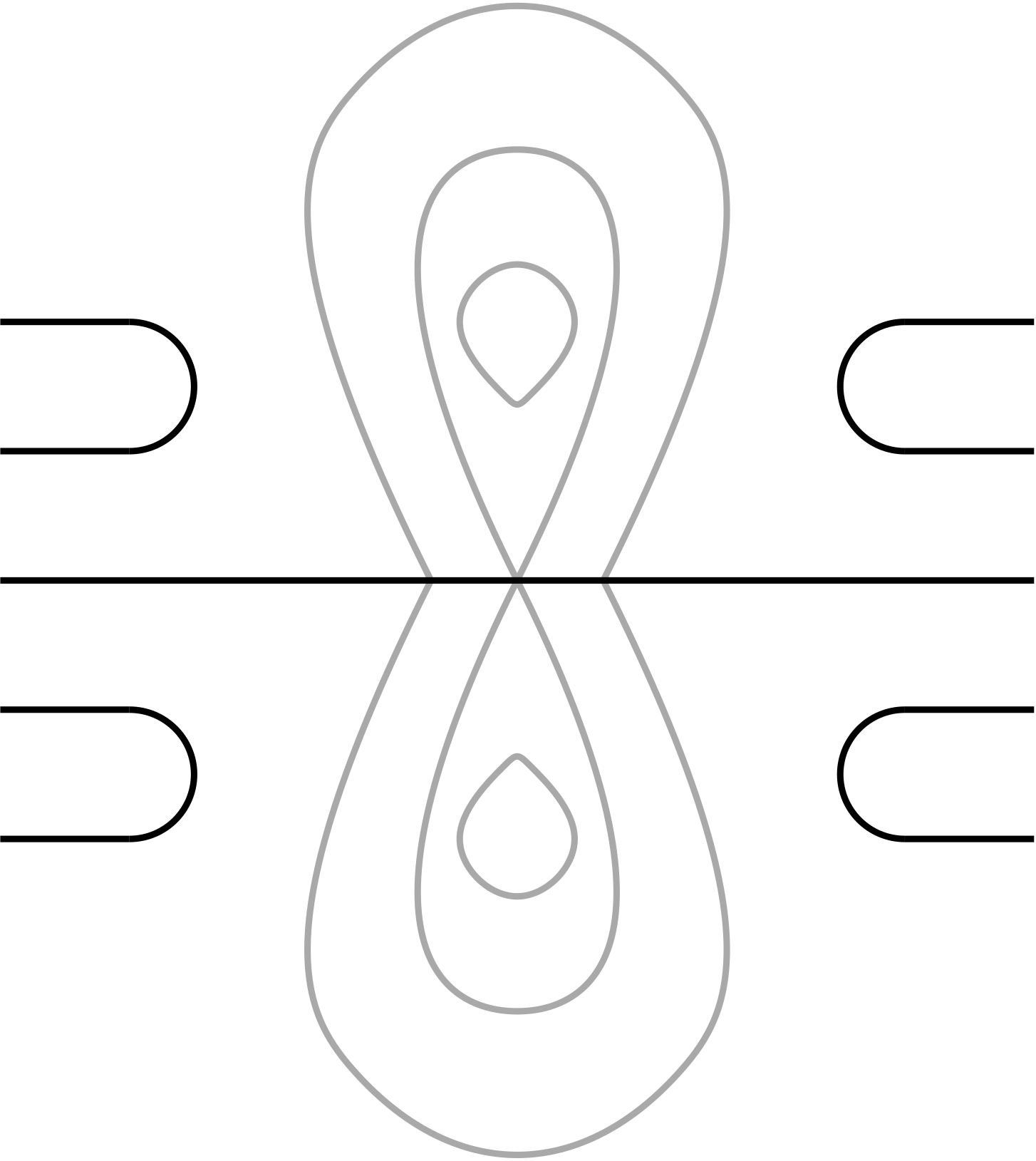}}
\hspace{1.5cm}
\subfloat[]{\includegraphics[scale=0.75]{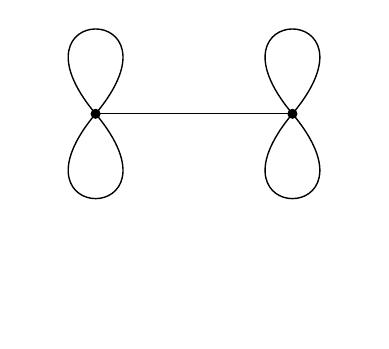}}
\caption{An example of bad double-tadpole subgraph (left panel). The equivalent representation provided in the central panel gives a clearer picture of the face structure. The associated boundary graph is represented in the rightmost panel.}
\label{fig:double_tadpole}
\end{figure}

A serious difficulty arises from the fact that we can construct chains of bad double-tadpoles, arranged in such a way that: for every double-tadpole we add to the chain, two additional faces are being closed (see Fig.~\ref{fig:chain_bad}). 
\begin{figure}[htbp]
\centering
\includegraphics[scale=0.75]{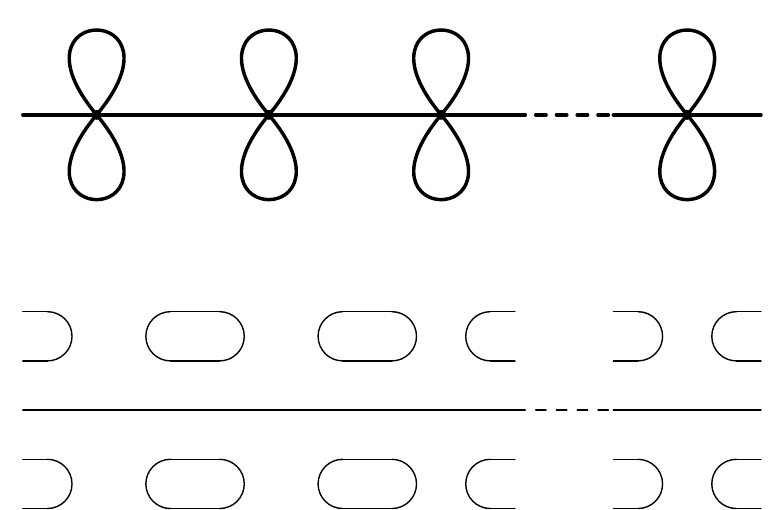}
\caption{A chain of bad double-tadpoles: two additional faces are being closed for each double-tadpole one adds to the chain.}
\label{fig:chain_bad}
\end{figure}
With $p$ double-tadpoles, the scaling of such a chain is:
\begin{equation}
\left(\frac{1}{N}\right)^pN^{2p-1}=N^{p-1}\,,
\end{equation}
which is unbounded from above. 

As a result, we observe that the degree is unbounded from below in the class of all stranded graphs. Nevertheless, we will see that the irreducible nature of the tensor representations we are working with allows to tame the contributions of such diagrams.

\subsection{Face of length two: melons, dipoles and dipole-tadpoles}

We will focus on three particular submap structures that can support faces of length two. We start with the minimal one.
\begin{definition}
A \emph{dipole} is an eight-point Feynman map or stranded graph with two vertices, two edges (which we call internal edges) and no self-loop. See Fig.~\ref{fig:dipole}.
\end{definition}
\begin{figure}[htbp]
\centering
\includegraphics[scale=1]{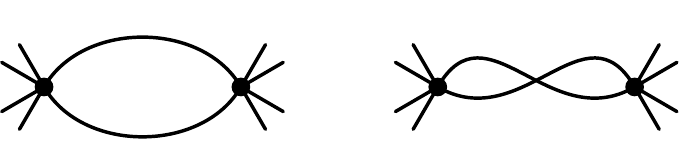}
\caption{Examples of (planar and non-planar) dipoles.}
\label{fig:dipole}
\end{figure}
\nomenclature[9]{\textbf{Dipole.}}{Eight-point Feynman map or stranded graph with two vertices, two edges (which we call internal edges) and no self-loop; see Fig.~\ref{fig:dipole}.}

As will become clear later on, we will have to pay extra attention to dipole subgraphs which appear in two other types of structures, which we now introduce. The first one is the familiar melon.  
\begin{definition}
A \emph{melon} is a two-point Feynman map or stranded graph with two vertices, five edges, and no self-loop. See Fig.~\ref{fig:melon}.
\end{definition}
\nomenclature[a]{\textbf{Melon.}}{Two-point Feynman map or stranded graph with two vertices, five edges, and no self-loop; see Fig.~\ref{fig:melon}.}

\begin{figure}[htbp]
\centering
\includegraphics[scale=1]{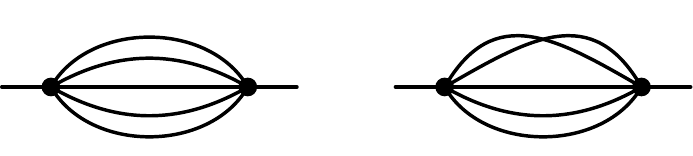}
\caption{Examples of (planar and non-planar) melon two-point maps.}
\label{fig:melon}
\end{figure}
As for tadpoles and double-tadpoles, a graph or map obtained from a melon by dressing its propagator edges with non-trivial two-point functions will be called a \emph{generalized melon}.

We finally introduce a particular subgraph containing a dipole and two tadpoles, which we call a dipole-tadpole. 
\begin{definition}\label{def:dipole-tadpole}
A \emph{dipole-tadpole} is a four-point Feynman map or stranded graph with two vertices, four edges, and exactly one self-loop on each vertex. See Fig.~\ref{fig:dipole-tadpole}.
A dipole-tadpole will be called \emph{separating} if it is adjacent to a generalized double-tadpole, as represented in the right panel of Fig.~\ref{fig:dipole-tadpole}.
\end{definition}
\begin{figure}[htbp]
\captionsetup[subfigure]{labelformat=empty}
\centering
\subfloat[]{\includegraphics[scale=1]{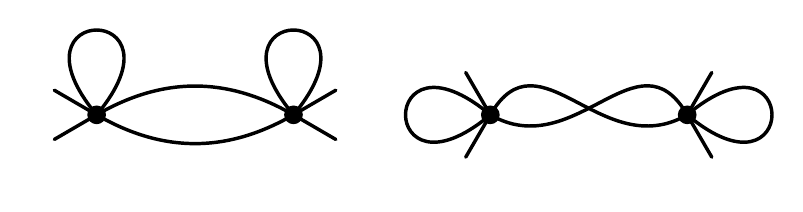}}
\hspace{1cm}
\subfloat[]{\includegraphics[scale=1]{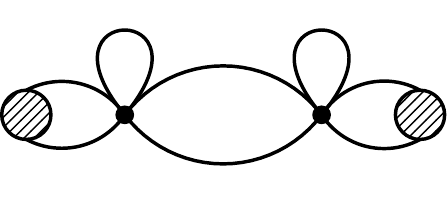}}
\caption{Examples of dipole-tadpole maps. The rightmost dipole-tadpole is separating.}
\label{fig:dipole-tadpole}
\end{figure}
We will also talk of \emph{generalized dipole-tadpole} if we allow the internal edges of a dipole-tadpole to be dressed by non-trivial two-point functions. 
\nomenclature[b]{\textbf{Dipole-tadpole.}}{ Four-point Feynman map or stranded graph with two vertices, four edges, and exactly one self-loop on each vertex; see Fig.~\ref{fig:dipole-tadpole}.}

\subsection{Type-$I$ and type-$II$ configurations}\label{sec:types}

Finally, it will later prove convenient to distinguish two types of tadpoles and dipoles: those which appear as subgraphs of generalized double-tadpoles, generalized melons or generalized dipole-tadpoles, as represented in Fig.~\ref{fig:not_easy}; and all the others. We will label the latter as \emph{type-$I$}, the former as \emph{type-$II$}. 
\begin{figure}[H]
\centering
\captionsetup[subfigure]{labelformat=empty}
\subfloat[]{\includegraphics[scale=1]{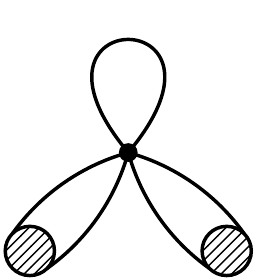}}
\hspace{1cm}
\subfloat[]{\includegraphics[scale=1]{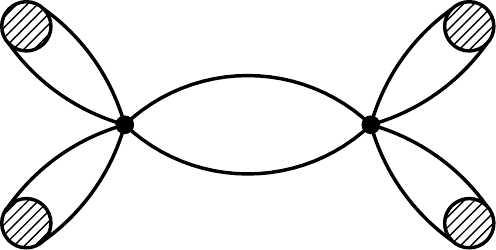}}
\hspace{1cm}
\subfloat[]{\includegraphics[scale=1]{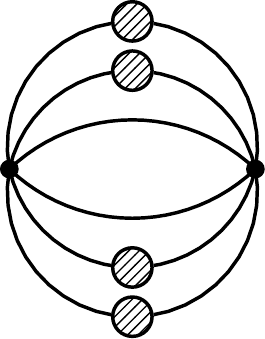}}
\caption{Type-$II$ tadpole (left) and dipoles (middle and right). By exclusion, a tadpole or dipole in any other configuration is of type $I$.}
\label{fig:not_easy}
\end{figure} 
 
\section{Subtraction of double-tadpoles and melons}
\label{sec:subtraction}

In this section, we show that the irreducible model with propagator $\pmb P$ in Eq.~\eqref{eq:model} is equivalent to a theory with renormalized covariance, in which melons and double-tadpoles have been subtracted from the Feynman expansion. While not strictly necessary to prove the existence of the large $N$ expansion \cite{Carrozza:2018ewt}, this reformulation is convenient. It cleanly separates Feynman maps that support stranded configurations with non-positive degrees, from those that do not. Only the latter can be accurately estimated by analyzing the combinatorial structure of their stranded configurations, a task we will turn to in section \ref{sec:deletions}.

Let us start by estimating the amplitude of melon and double-tadpole maps. 
\begin{lemma}\label{lemma:melon_tadpole}
Let $\mathcal{G}$ be a (non-amputated) two-point Feynman map. The associated amplitude $\mathcal{A}(\mathcal{G})_{\pmb a,\pmb b}$ can be written as:
\begin{equation}\label{eq:schur1}
\mathcal{A}(\mathcal{G})_{\pmb a,\pmb b}=\lambda^{V(\mathcal{G})}f_{\mathcal{G}}(N)\pmb P_{\pmb a,\pmb b}\,,
\end{equation}
where $f_{\mathcal{G}}$ is some (rational) function. 
Furthermore:
\begin{itemize} 
\item if $\cG$ is a double-tadpole, then $f_{\mathcal{G}}(N) = \mathcal{O}(1/N)$;
\item if $\cG$ is a melon, then $f_{\mathcal{G}}(N) = f_\cG^{(0)} + \mathcal{O}(1/N)$ where $f_\cG^{(0)} \in \mathbb{R}$. Moreover, when $\pmb P \in \{ \pmb S , \pmb A \}$, one necessarily has $f_\cG^{(0)} > 0$.
\end{itemize}
\end{lemma}
\begin{proof}
The functional form of Eq.~\eqref{eq:schur1} is a direct consequence of the irreducibility of the representation. It  follows from Schur's lemma for any two-point graph $\cG$. 

Let us first assume that $\cG$ is a double-tadpole. It is clear that any stranded configuration of $\cG$ has at most four faces (this can be formalized with the help of e.g. the bounds of Appendix~\ref{ap:bounds}). They contribute a factor of order at most $N^4$, which is compensated by the $1/N^5$ scaling of the vertex. Hence $f_{\mathcal{G}}(N) = \mathcal{O}(1/N)$.

Next, we assume $\cG$ to be a melon. Consider one of its stranded configurations $G$. As will be clear from Remark~\ref{rem:broken_unbroken} below, we can assume that $G$ contains only unbroken edges. We have $(5\times4)/2=10$ internal corners at our disposal on each vertex to build up faces, so $20$ corners in total. From the structure of the melon and the unbroken character of the edges of $G$, it is also clear that any face must have length at least two. It immediately follows that $F(G)\leq 10$, leading to a contribution to the amplitude scaling like $N^{10}$ at most. Taking the two factors of $1/N^5$ coming from the vertices into account, we infer that $f_{\mathcal{G}}(N) = \mathcal{O}(1)$. Let us finally specialize to $\pmb P  \in \{ \pmb S , \pmb A\}$. Given that any unbroken edge configuration contributes to $\pmb P$, it is straightforward to show that: 
a) this bound can be saturated; b) the configurations that do so have only unbroken edges and the same boundary graph, namely, that of an unbroken edge; c) despite having identical boundary graphs, the way in which the external strands are being paired up in any two such configurations differ by a permutation. As a result, there can be no cancellation between leading order stranded configurations, which implies that $f_\cG^{(0)} \neq 0$.\footnote{This is a crucial difference with double-tadpoles. For the latter, leading order stranded configurations are of the doubly-broken type, and as a result, necessarily cancel out once resummed into the full amplitude.} Given the symmetric structure of the melon and of its leading order contributions, it is also possible to show that $f_\cG^{(0)} > 0$, irrespectively of the choice of irreducible representation. This is direct for the symmetric traceless propagator since there are no signs involved in its unbroken stranded contributions. For the antisymmetric representation, we can infer from the structure of a leading-order melon stranded graph that the product of the signatures of the permutations labeling its unbroken edges (including the external one) is necessarily even, leading to an overall positive sign. We leave the details of the proof, which follows from footnote \ref{ft:vertex_inv}, to the interested reader.  
\end{proof}

\begin{remark}\label{rem:improved_lemma}
If one works with the colorable interaction kernel \eqref{eq:colored_5-simplex} instead of \eqref{eq:5-simplex}, it is possible to prove that $f_\cG^{(0)} \geq 0$ for any melon $\cG$, and $f_\cG^{(0)} > 0$ for at least one such $\cG$. Indeed, it can be shown that, in this particular case, the unique leading-order stranded configuration of a closed melon happens to be decorated by the same unbroken edge (that is, the same permutation $\sigma \in \mathcal{S}_5$) on all six propagators. Hence, the coefficient associated to this particular unbroken edge is raised to an even power, and the overall sign of the amplitude is always positive. Moreover, it is straightforward to see that any $\sigma$ can contribute to a leading-order melon in such a way, as there is always at least one non-zero unbroken contribution in each propagator. We then conclude that at least one melon is non-vanishing at leading order.
\end{remark}

In light of the previous Lemma, it is clear that double-tadpole submaps are well-behaved in the large $N$ limit, even though some of their stranded configurations are not. To prove the existence of the large $N$ limit, we must therefore make sure to always bound a double-tadpole Feynman map as a whole. Moreover, it is also clear from Lemma~\ref{lemma:melon_tadpole} that melon two-point functions will contribute to the leading order. By dressing double-tadpole subgraphs with such two-point functions, we can generate a family of stranded graphs with arbitrarily negative degrees, but no double-tadpoles. This indicates that the whole family of two-point functions generated by double-tadpoles and melons needs to be treated with care: the existence of the large $N$ expansion cannot be deduced from bounds on their individual stranded configurations. 

We then follow the method of \cite{Benedetti:2017qxl} and adapt it to rank $5$. We consider a modified theory with covariance $K\pmb P$ where $K$ is a real number. Let us denote by $\Sigma^{(2)}$ the contribution of melon and double-tadpole maps to the self-energy. By Lemma~\ref{lemma:melon_tadpole}, we have:
\begin{equation}
\Sigma^{(2)}_{\pmb a,\pmb b}=\left(\lambda K f_1^{\pmb P}+\lambda^2K^5f_2^{\pmb P}\right) \pmb P_{\pmb a,\pmb b}\,,
\end{equation}
where $f_1^{\pmb P}(N)$ and $f_2^{\pmb P}(N)$ are series in $1/N$ verifying
\begin{equation}
f_1^{\pmb P}(N)=\mathcal{O}(1/N) \qquad \mathrm{and} \qquad f_2^{\pmb P}(N)= m_{\pmb P} + \mathcal{O}(1/N)\,.
\end{equation}
Moreover, for $\pmb P \in \{ \pmb S , \pmb A\}$, the constant $m_{\pmb P}$ is necessarily non-vanishing and positive.\footnote{Owing to Remark~\ref{rem:improved_lemma}, we also have $m_{\pmb P} >0$ for \emph{any} $\pmb P$ with the alternative choice of vertex kernel \eqref{eq:colored_5-simplex}.} Up to symmetry factors, it essentially counts the number of leading order melon stranded graphs.

As an illustration, for $\pmb P = \pmb A$ or $\pmb S$, we have the exact formula:\footnote{Analogous formulas exist for mixed representations, but they are slightly more involved as in those cases the Feynman amplitudes may depend on the embedding information of the Feynman maps.}
\begin{align}
\Sigma^{(2)}_{\pmb a,\pmb b}&=15\frac{\lambda K^2}{N^5}\sum_c \pmb P_{a_1a_2a_3a_4a_5,c_1c_2c_3c_4c_5}\pmb P_{c_5c_6c_7c_8c_9,c_9c_4c_{10}c_{11}c_{12}}\pmb P_{c_{10}c_6c_3c_{13}c_{14},c_{14}c_{11}c_7c_2c_{15}}\pmb P_{c_{15}c_{13}c_{12}c_8c_1,b_1b_2b_3b_4b_5} \crcr
& +120 \frac{\lambda^2K^5}{N^{10}}\sum_{c,d}\pmb P_{a_1a_2a_3a_4a_5,c_1c_2c_3c_4c_5}\pmb P_{c_5c_6c_7c_8c_9,d_5d_6d_7d_8d_9}\pmb P_{c_9c_4c_{10}c_{11}c_{12},d_9d_4d_{10}d_{11}d_{12}}\crcr
& \qquad\pmb P_{c_{12}c_8c_{3}c_{13}c_{14},d_{12}d_8d_{3}d_{13}d_{14}}\pmb P_{c_{14}c_{11}c_{7}c_{2}c_{15},d_{14}d_{11}d_{7}d_{2}d_{15}} \pmb P_{c_{15}c_{13}c_{16}c_{2}c_{1},d_{15}d_{13}d_{10}d_{6}d_{1}}\pmb P_{d_1d_2d_3d_4d_5,b_1b_2b_3b_4b_5}
\end{align}
Using the explicit expression of the propagators, we can find exact expressions for $f_1^{\pmb P}$ and $f_2^{\pmb P}$. For instance, we determined (by numerical methods) that:\footnote{The computation could in principle be performed for $f_2^{\pmb P}$ as well, but it is more costly.}
\begin{align}
f_1^{\pmb A}= \frac{(N-4)^2(N^2-13N+34)}{115200 N^5} \,, \crcr
f_1^{\pmb S}= \frac{(N+8)^2(N^5+19N^4+50N^3-356N^2+8N+672)}{960N^4(N+4)^2(N+6)^2}  \, . 
\end{align}
More interestingly for the large $N$ limit itself, we can in fact evaluate $m_{\pmb P} = \underset{N\to \infty}{\lim} f_2^{\pmb P}(N)$ exactly, which we briefly sketch. Consider a melon map $\cG$. A leading order stranded configuration of $\cG$ with unbroken boundary graph can only have unbroken edges. Furthermore, once we fix the structure of the external strands, there is a unique choice of configuration of the five internal edges that makes the graph leading order. Since in both $\pmb S$ and $\pmb A$, an unbroken edge is weighted by the combinatorial factor $1/5!$ (up to a sign), we conclude that the contribution of $\cG$ to $m_{\pmb P}$ is $(1/5!)^5$. Given that there are $5!$ melon maps, this finally leads to: 
\begin{equation}\label{eq:m}
m_{\pmb S} = m_{\pmb A} = (1/5!)^4\,.
\end{equation}

Following \cite{Benedetti:2017qxl}, we denote $\Sigma^{(2)}=\lambda K f_1^{\pmb P} + \lambda^2 K^5 f_2^{\pmb P}$, $T^6$ the interaction of equation \eqref{eq:model1}, and define the subtracted interaction:
\begin{equation}
\frac{\lambda}{6N^5}:T^6:_K=\frac{\lambda}{6N^5}T^6-\frac{1}{2}\Sigma^{(2)}T\pmb P T\,.
\end{equation}
As is clear from the notation, this enforces a form of Wick ordering with respect to the covariance $K \pmb P$, which subtracts the double-tadpole and melon interactions. As a result, the model with covariance $K \pmb P$ and interaction $:T^6:_K$ can be expanded in terms of Feynman maps which have neither double-tadpoles nor melons subgraphs. 

The last step amounts to choosing $K$ in such a way that the model with covariance $K \pmb P$ and subtracted interaction is nothing but our original model of equation \eqref{eq:model}.
\begin{align}
F_{\pmb P}(\lambda)&=\frac{6}{N^{5}}\lambda\partial_{\lambda}\ln \left\lbrace \left[e^{\frac{1}{2}\partial_T \pmb P \partial_T}e^{\frac{\lambda}{6N^5}T^6}\right]_{T=0}\right\rbrace \crcr
&=\frac{6}{N^{5}}\lambda\partial_{\lambda}\ln \left\lbrace \left[e^{\frac{1}{2}\partial_T \pmb P \partial_T}e^{\frac{\lambda}{6N^5}:T^6:_K+\frac{\Sigma^{(2)}}{2}T\pmb P T}\right]_{T=0}\right\rbrace \crcr
&=\frac{6}{N^{5}}\lambda\partial_{\lambda}\ln \left\lbrace \left[e^{\frac{1}{2}\frac{1}{1-\Sigma^{(2)}}\partial_T \pmb P \partial_T}e^{\frac{\lambda}{6N^5}:T^6:_K}\right]_{T=0}\right\rbrace
\end{align}
From the last line, we need to ensure that $K = (1-\Sigma^{(2)})^{-1}$, which results in a polynomial equation for $K$:
\begin{equation}\label{eq:sde}
1-K+\lambda f_1^{\pmb P}K^2+\lambda^2f_2^{\pmb P} K^6=0 \,. 
\end{equation}
Equivalently, this equation can be deduced from the Schwinger-Dyson equation of melon and double-tadpole two-point functions, which we have illustrated in Fig.~\ref{fig:SDE}. For $N$ large and $\lambda$ small enough this equation admits a unique solution $K(\lambda,N)$ with the following properties: it is a series in both $\lambda$ and $1/N$, it is uniformly bounded in both $N$ and $\lambda$, and
\begin{equation}
\lim_{\lambda \rightarrow 0}\left[\lim_{N\rightarrow \infty}K(\lambda,N)\right]=1 \, . 
\end{equation} 
Furthermore, $\lim_{N\rightarrow \infty}K(\lambda,N)$ is a series in $\lambda^2$ which coincides with the generating function of Fuss-Catalan numbers $A_n (6,1)$.

\begin{figure}[htbp]
\centering
\includegraphics[scale=.8]{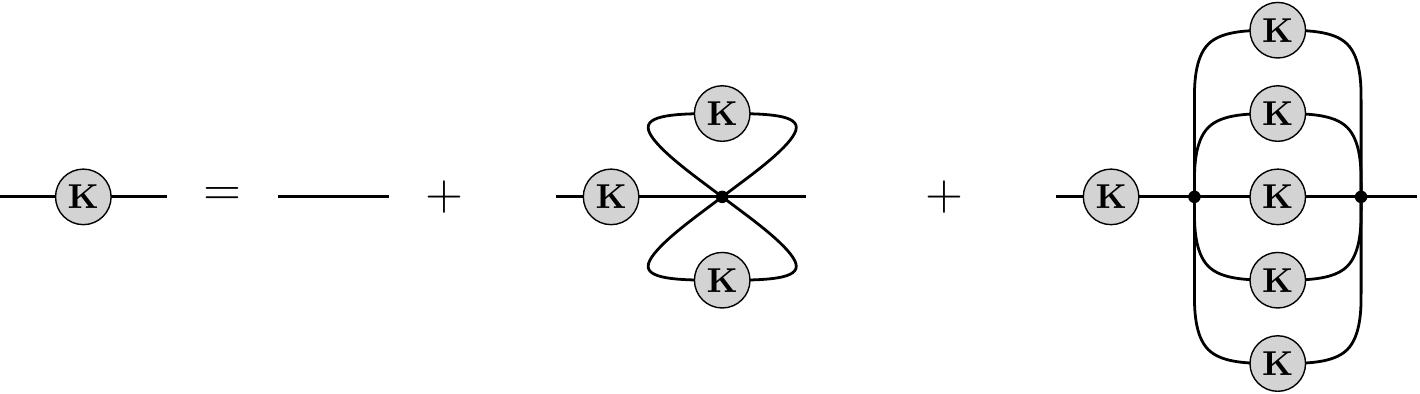}
\caption{Schematic structure of the Schwinger-Dyson equation resumming arbitrary two-point melon and double-tadpole maps (for simplicity, we are ignoring embedding information).}
\label{fig:SDE}
\end{figure}

We can then write equation \eqref{eq:model} as:
\begin{equation}
F_{\pmb P}(\lambda)=\frac{6}{N^{5}}\lambda\partial_{\lambda}\ln \left\lbrace \left[e^{\frac{K(\lambda,N)}{2}\partial_T \pmb P \partial_T}e^{\frac{\lambda}{6N^5}:T^6:_{K(\lambda,N)}}\right]_{T=0}\right\rbrace\,,
\end{equation}
and obtain the looked-for perturbative expansion in terms of Feynman maps with no double-tadpoles or melons:
\begin{equation}
F_{\pmb P}(\lambda)=\sum_{\substack{\hat{G} \text{ connected, rooted }\\  \text{with no double-tadpoles or melons}}} \lambda^{V(\hat{G})}\left[ K(\lambda,N) \right]^{U(\hat{G})+B_1(\hat{G})+B_2(\hat{G})}\mathcal{A}(\hat{G})\,.
\label{eq:model_subtracted}
\end{equation}
In this equation, $\mathcal{A}$ designates the same amplitude map as defined in Eq.~\eqref{eq:ampli}.
Given that $K(\lambda,N)$ is also a series in $1/N$, the $1/N$ expansion in Theorem \ref{theorem_princ} follows from Remark \ref{rem:broken_unbroken} and from Proposition \ref{prop:positive_degree}.

\section{Non-negativity of the degree}\label{sec:deletions}

In this section, we prove that the degree of a stranded graph with no melon and no double-tadpole is non-negative. Because the distinction between graphs and embedded graphs does not matter for this purpose, we will ignore it. In particular, our figures should now be understood as representing equivalent classes of maps which only differ through their embedding. 

\subsection{Flip distance between boundary graphs and scaling bounds}\label{sec:flip_distance}

In the following, it will be convenient to extract large $N$ scaling information by direct inspection of the boundary graph of a given stranded configuration. 

To this effect, we first introduce combinatorial moves acting on pairs of edges in a boundary graph, which we call \emph{flips}. Given two distinct edges $e_1$ and $e_2$ in a boundary graph $B$, a flip amounts to: 1) cutting $e_1$ and $e_2$ open; and 2) recombining the resulting four half-edges in one of two possible channels, to obtain a new boundary graph $\tilde{B}$. This is illustrated in Fig.~\ref{fig:flips_boundary}. It is easy to see that the set of (not necessarily connected) boundary graphs with prescribed number of vertices is stable under flips. Moreover, these moves are ergodic in this space: given two $5$-regular graphs, it is always possible to transform one into the other through a finite number of successive flips. As a result, we can introduce a notion of \emph{flip distance} between such graphs.  

\begin{figure}[htbp]
\centering
\includegraphics[scale=1]{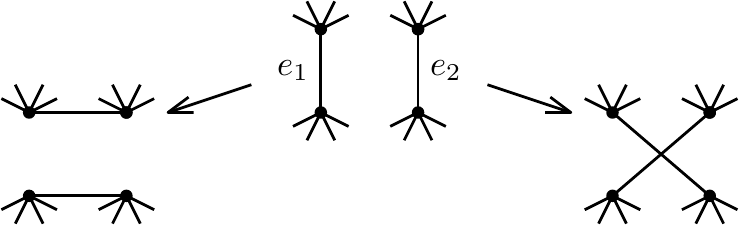}
\caption{The two possible flips of edges $e_1$ and $e_2$ in a boundary graph.}
\label{fig:flips_boundary}
\end{figure}

\begin{definition}
Let $B_1$ and $B_2$ be boundary graphs with $n\geq 2$ vertices. We define the \emph{flip distance} between $B_1$ and $B_2$, denoted by $d(B_1 , B_2)$, as the minimal number of successive flips required to map $B_1$ to $B_2$. 
\end{definition}
\nomenclature[c]{\textbf{Flip distance.}}{Minimal number of successive flips required to map two boundary graphs; see Sec.~\ref{sec:flip_distance}.}
It is elementary to check that $d$ defines a proper notion of distance on the space of boundary graphs with $n$ vertices. By convention, we will also postulate that $d(B_1 , B_2) = \infty$ whenever $B_1$ and $B_2$ do not have the same number of vertices. 
The relation between flip distance and scaling is captured by the following proposition.
\begin{proposition}\label{propo:general_deletion}
Consider a stranded graph $G$, and a strict subgraph $S \subset G$. Let $B$ be a boundary graph such that $d(B, S_\partial) < \infty$. Then, there exists a stranded graph $S'$ such that $S'_\partial = B$ and
\begin{equation}\label{propo:bound1}
\vert F(G') - F(G) \vert \leq \vert F(S) - F(S') \vert + d(B, S_\partial)\,,
\end{equation}
where $G'$ is the graph obtained by substitution of $S$ by $S'$ into $G$.

In particular, if $G$ and $G'$ contain only unbroken edges, $F(S') = 0$, and $G'$ remains connected, we find:
\begin{equation}\label{propo:bound2}
\omega(G) \geq \omega(G') + 5 \left( V(S) - V(S') \right) - F(S) - d(S_\partial , S'_\partial )\,.
\end{equation}
\end{proposition}
\begin{proof}
The idea is to perform a succession of cut-and-glue operations on the internal strands of $S$ (while leaving the rest of the graph unchanged), until we obtain a new subgraph with boundary $B$. Since each cut-and-glue operation is reflected by a flip at the level of boundary graphs, this can be done in at most $d(B,S_\partial)$ steps. We then insert or remove internal faces to change $F(S)$ into $F(S')$, and obtain the target stranded subgraph $S'$. This last step is responsible for the first term in the right-hand-side of Eq.~\ref{propo:bound1}. Furthermore, it is clear that each cut-and-glue operation changes the number of faces in the graph by $-1$, $0$ or $1$, which explains the second term.  
Finally, if we assume that $G$ and $G'$ contain only unbroken edges, then $B_1 = B_2 =0$ for both graphs, and together with $F(S')=0$, Eq.~\eqref{propo:bound2} follows from Eqs.~\eqref{propo:bound1} and \eqref{eq:degree}.
\end{proof}
Equation \eqref{propo:bound2} will be particularly relevant because it will allow us to derive inductive bounds of the form $\omega(G) \geq \omega(G')$ from the local combinatorial condition:
\begin{equation}
d(S_\partial , S'_\partial ) \leq  5 \left( V(S) - V(S') \right) - F(S) \,.
\end{equation}

As a simple illustration of equation \eqref{propo:bound1}, consider the boundary graph of a single edge $e$. If $e$ is doubly-broken, it is at flip distance one from the boundary graph of a broken edge, which is itself at flip distance one from the boundary graph of an unbroken edge; see Fig.~\ref{fig:propagator_bndy}. 

\begin{figure}[htbp]
\centering
\includegraphics[scale=1]{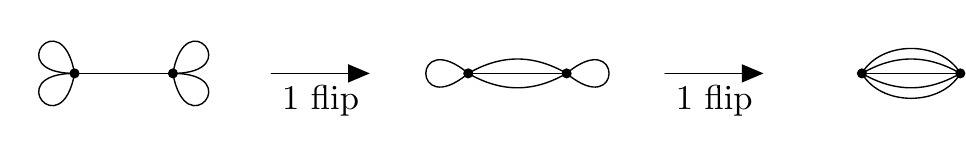}
\caption{Flip distance between the boundary graphs of the doubly-broken, broken and unbroken propagators.}
\label{fig:propagator_bndy}
\end{figure}

Therefore, we can replace a broken edge by an unbroken one in such a way that the number of faces decreases at most by one. As a result, the number of broken edges decreases by one, which implies that the degree \eqref{eq:degree} can only decrease. Likewise, we can replace a doubly-broken edge by an unbroken one in such a way that the number of faces decreases at most by two, whereas the number of doubly-broken edges decreases by one. Again, the degree can only decrease. This leads to the following observation.
\begin{remark}\label{rem:broken_unbroken}
For any stranded graph $G$, there exists a stranded graph $G'$ with $B_1(G')= B_2(G') = 0$, and such that
\begin{equation}
\omega(G)\geq \omega(G')\,.
\end{equation}
\end{remark}
Hence, for the purpose of finding lower bounds on the degree, we can restrict ourselves to graphs with only unbroken edges. This property is assumed in the remainder of the present section.

We now turn to the definition of basic combinatorial moves, which we will use in combination in the proof of subsection \ref{subsec:main_proof}. A first straightforward example concerns double-tadpoles.
\begin{lemma}\label{lemma:double_tadpole_deletion}
Consider a stranded graph $G$ with a double-tadpole subgraph $S$. It is possible to replace $S$ by an unbroken edge in such a way that the resulting graph $G'$ verifies:
\begin{equation}
\omega(G) \geq \omega(G') - 1\,.
\end{equation}
\end{lemma}
\begin{proof}
We notice that: $F(S)\leq 4$ if $S_\partial$ is of the doubly-broken type; $F(S)\leq 3$ if $S_\partial$ is of the simply-broken type; and $F(S)\leq 2$ otherwise. The result then follows from Eq.~\eqref{propo:bound2}.
\end{proof}
Less straightforward examples will be the focus of the next four subsections. 

\subsection{Single-tadpole deletions}

We first look for combinatorial moves that replace a single-tadpole subgraph with two (unbroken) propagators, and delete as few faces as possible. If we ignore for the moment the permutations labeling the two edges after the deletions, there are exactly three ways of doing so, which amount to a choice of pairing of the external legs of the subgraph: we call these \textit{deletion channels}, or simply \emph{channels}. They are the \emph{parallel} (pairing $(a,c)$ and $(b,d)$), \emph{cross} (pairing $(a,d)$ and $(b,c)$) and \emph{orthogonal} (pairing $(a,b)$ and $(c,d)$) channels, as illustrated in Fig.~\ref{fig:channel_single_tadpole}. Note that this nomenclature is purely conventional: it depends on an arbitrary labeling of the external legs of the tadpole. In the following, we will fix a canonical labeling for each possible structure of the boundary graph. 

\begin{figure}[htbp]
\centering
\includegraphics[scale=1]{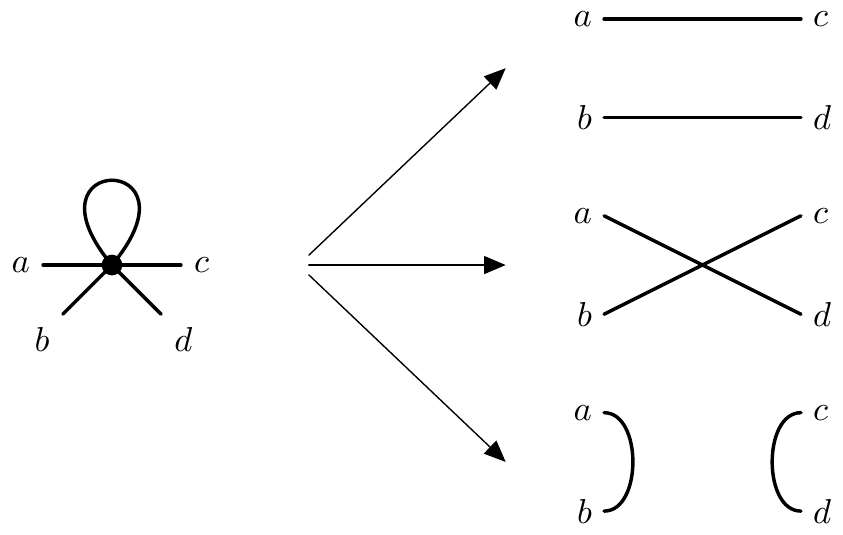}
\caption{The three deletion channels of a single-tadpole. From top to bottom: parallel, cross and orthogonal channels.}
\label{fig:channel_single_tadpole}
\end{figure}

To find a suitable deletion along the lines of Proposition \ref{propo:general_deletion}, we first need to determine the structure of the boundary graph $S_\partial$. Up to a relabeling of the vertices, we find the five possibilities represented in Fig.~\ref{fig:tadpole_config}. Indeed, first notice that the structure of the vertex imposes the presence of a $K_4$ subgraph (the complete graph on $4$ vertices). We have represented this subgraph in grey in Fig.~\ref{fig:tadpole_config}. We are left with a choice of pairing of eight remaining half-edges (two per vertex), to form the four edges that we have represented in black. 
The five configurations we end up with are distinguished by the lengths of the cycles formed by the black edges, and can be labeled by the partitions of $4$. Indeed, we have a budget of four edges, which can be split up into: four cycles of length one ($1+1+1+1$); two cycles of length one and one of length two ($1+1+2$); one cycle of length one and one of length three ($1+3$); two cycles of length two ($2+2$); or one cycle of length four ($4$). 

\begin{figure}[htbp]
\centering
\begin{tabular}{ccccc}
\subfloat[1+1+1+1]{\includegraphics[width=0.15\textwidth]{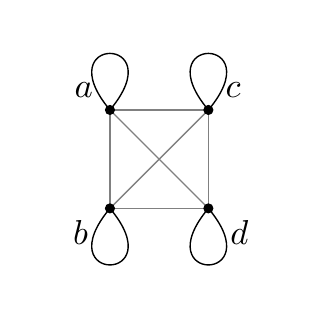}}
&
\subfloat[1+1+2]{\includegraphics[width=0.15\textwidth]{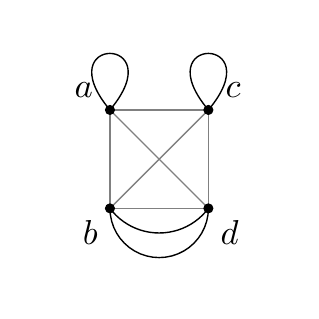}}
&
\subfloat[1+3]{\includegraphics[width=0.15\textwidth]{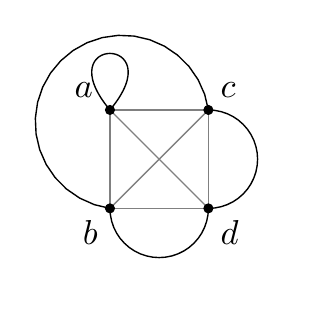}}
&
\subfloat[2+2\label{fig:tadpole_config_d}]{\includegraphics[width=0.15\textwidth]{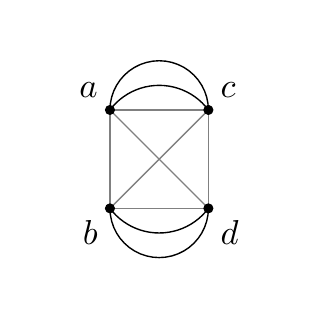}}
&
\subfloat[4\label{fig:tadpole_config_e}]{\includegraphics[width=0.15\textwidth]{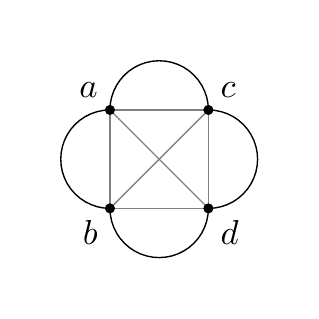}}
\end{tabular}
\caption{The five possible boundary graphs of a single-tadpole.}
\label{fig:tadpole_config}
\end{figure}

We can now write the following Lemma.
\begin{lemma}\label{lemma:single_tadpole}
Let $G$ be a stranded graph, and $S$ a strict single-tadpole subgraph of $G$. Call $G'$ the graph obtained after a deletion of $S$ in the channel $c$, and assume that $G'$ remains connected. 
\begin{enumerate}
\item If $S_\partial$ is in the configuration $1+1+2$, it is possible to choose $G'$ such that:
\begin{enumerate}
\item $\omega(G) \geq \omega(G') + 1$ 
when $c$ is the parallel channel;
\item $\omega(G) \geq \omega(G') - 1$ 
when $c$ is any other channel.
\end{enumerate}
\item If $S_\partial$ is in any other configuration, it is possible to choose $G'$ such that
$\omega(G) \geq \omega(G')$.
\end{enumerate}
\end{lemma}

\begin{proof}
The single-tadpole $S$ can support at most one internal face, and exactly one vertex is lost upon deletion. $G'$ being connected, Proposition~\ref{propo:general_deletion} guarantees that we can arrange the strands in such a way that:
\begin{equation}
\omega(G) \geq \omega(G') + 4 - d(S_\partial , B_c)\,,
\end{equation} 
where $B_c$ is the boundary graph characterizing the channel $c$. For instance, if $c$ is the parallel channel, $B_c$ is the four-vertex graph in which vertices $a$ and $c$ are connected by five edges, and likewise for vertices $b$ and $d$. It remains to bound the flip distance between $S_\partial$ and $B_c$. It helps to first determine the flip distance between the various possible configurations of $S_\partial$, which we have represented in Fig.~\ref{fig:distance_tadpole}. It is then apparent that the Lemma follows from the following sufficient conditions: 
\begin{itemize}
\item if $S_\partial$ is in the configuration $1+1+1+1$, $d(S_\partial , B_c)\leq 4$ for any $c$;
\item if $S_\partial$ is in the configuration $2+2$ and $c$ is the parallel channel, then $d(S_\partial , B_c)\leq 2$;
\item if $S_\partial$ is in the configuration $4$, $d(S_\partial , B_c)\leq 3$ if $c$ is the parallel or orthogonal channel, and $d(S_\partial , B_c)\leq 4$ otherwise;
\item if $S_\partial$ is in the configuration $3+1$ and $c$ is the cross channel, then $d(S_\partial , B_c)\leq 4$.
\end{itemize}

\begin{figure}[htbp]
\centering
\includegraphics[scale = 1]{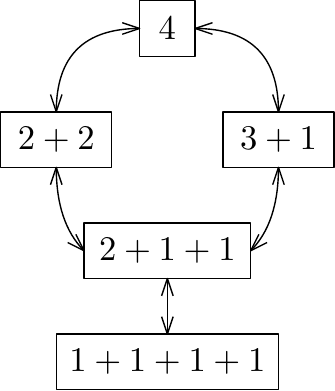}
\caption{Distance between the five single-tadpole configurations: any two partitions connected by an edge are at flip distance one from each other.}
\label{fig:distance_tadpole}
\end{figure}

If $S_\partial$ is in the configuration $1+1+1+1$, we need $2$ flips to disconnect the graph in the appropriate channel, and $2$ more flips to remove the self-loops. Hence, we have $d(S_\partial , B_c) \leq 4$
. This is illustrated in Fig.~\ref{fig:ex_deletion_tadpole}, for the parallel channel.

\begin{figure}[htbp]
\centering
\includegraphics[scale=1]{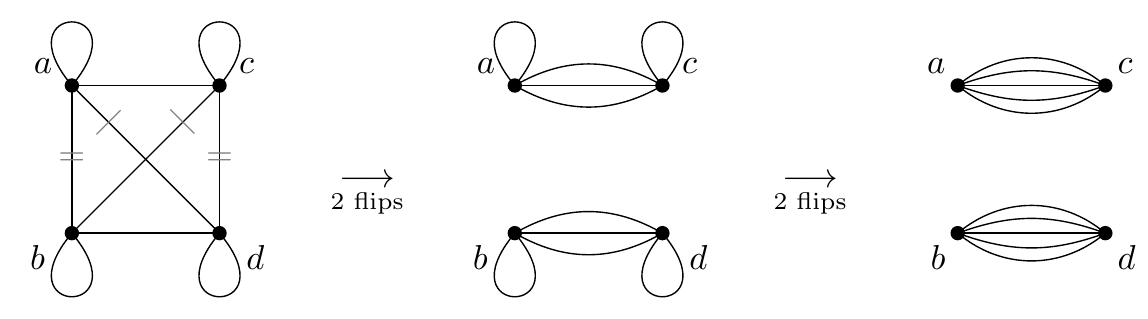}
\caption{Flip distance between a tadpole in the configuration $1+1+1+1$ and the parallel channel.}
\label{fig:ex_deletion_tadpole}
\end{figure}

If $S_\partial$ is in the configuration $2+2$ and $c$ is the parallel channel, we can infer that $d(S_\partial, B_c) \leq 2$ by first flipping the edges $(a,b)$ and $(c,d)$, then the edges $(a,d)$ and $(b,c)$.

Likewise, if $S_\partial$ is in the configuration $4$, we can show that $d(S_\partial, B_c) \leq 3$ if $c$ is the parallel or orthogonal channel. One needs an extra flip in the cross channel, because $a$ and $d$ (resp. $b$ and $c$) are initially connected by a single edge; hence, $d(S_\partial, B_c) \leq 4$  in that case.

If $S_\partial$ is in the configuration $3+1$ and $c$ is the cross channel, we need $3$ flips to disconnect. This has the effect of creating a second self-loop. We can then perform one more flip to remove the self-loops, and obtain the boundary graph $B_c$. As a result, $d(S_\partial , B_c) \leq 4$. This is illustrated in Fig.~\ref{fig:ex_deletion_tadpole2}.
\begin{figure}[htbp]
\centering
\includegraphics[scale=1]{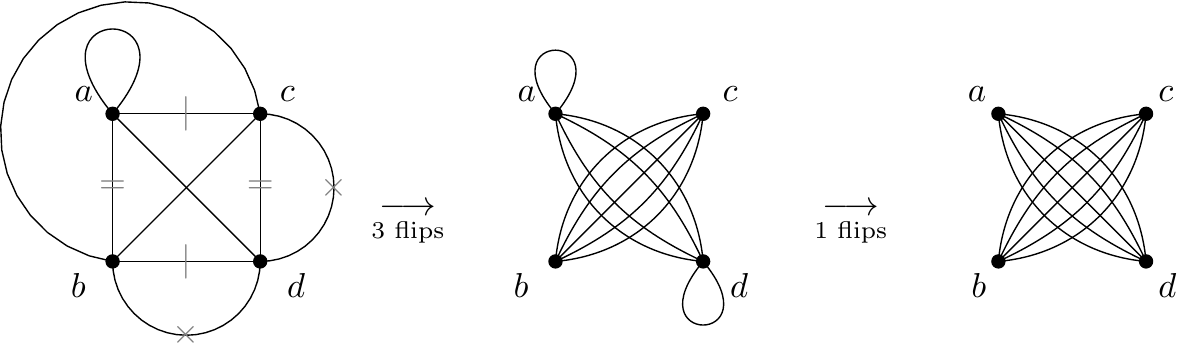}
\caption{Flip distance between a tadpole in the configuration $1+3$ and the cross channel.}
\label{fig:ex_deletion_tadpole2}
\end{figure}
\end{proof}

\subsection{Dipole deletions}

We will now look for combinatorial moves that replace a dipole subgraph with four (unbroken) propagators and delete as few faces as possible. In contrast to the single-tadpole deletions of the previous section, there are many more ways of doing so, leading to many more than three channels of deletions. However, for our purpose, it will be sufficient to consider only four of those channels. Indeed, all we need is a sufficiently rich set of deletion moves to ensure that, in all situations, at least one of them can be performed while maintaining our combinatorial constraints (connectedness, and the absence of melons or double-tadpoles). This subset of channels is presented in Fig.~\ref{fig:channel_dipole}. Note that, apart from the fact that the groups of half-edges $\{ 1, 2, 3, 4 \}$ and $\{ 5, 6, 7, 8 \}$ are attached to different vertices, the labeling is purely conventional at this stage. This will be taken advantage of and made more precise in the proof of Lemma \ref{lemma:dipole} (see also Fig.~\ref{fig:dipole_config}).

\begin{figure}[htbp]
\centering
\captionsetup[subfigure]{labelformat=empty}
\subfloat[]{\includegraphics[scale=1]{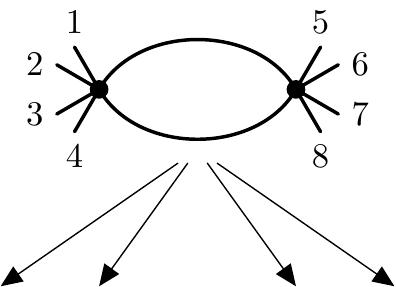}}
\\
\subfloat[(1)]{\includegraphics[scale=0.8]{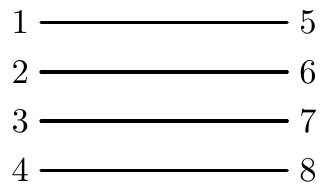}}
\hspace{1cm}
\subfloat[(2a)]{\includegraphics[scale=0.8]{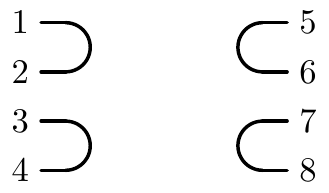}}
\hspace{1cm}
\subfloat[(2b)]{\includegraphics[scale=0.8]{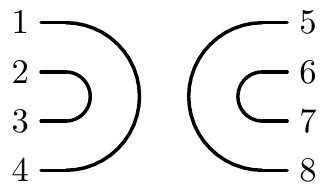}}
\hspace{1cm}
\subfloat[(2c)]{\includegraphics[scale=0.8]{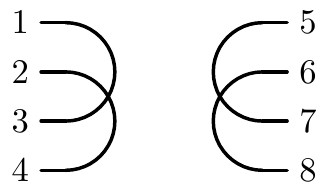}}
\caption{The four deletion channels we consider for a dipole; from left to right: channels $(1)$, $(2a)$, $(2b)$ and $(2c)$.}
\label{fig:channel_dipole}
\end{figure}

The reason for choosing these four channels is that, if we assume that the dipole is of type $I$, then at least one of them does not disconnect the graph. Indeed, suppose channel $(2a)$ disconnects. Then the graph is in either one of the configurations depicted in Fig.~\ref{fig:dipole_connected_a}, \ref{fig:dipole_connected_b} and \ref{fig:dipole_connected_c}, with subgraphs $A$ and $B$ not necessarily connected. 
\begin{itemize}
\item If it is in the first configuration, then channels $(2b)$ and $(2c)$ also disconnect but channel $(1)$ does not. Indeed, otherwise there would be generalized double-tadpoles on both vertices of the dipole, which means that the latter would be of type $II$.
\item If it is in the second configuration, channels $(2b)$ and $(2c)$ do not disconnect. Otherwise, the subgraphs $A$ and $B$ would have to be disconnected, in a way that would either generate two generalized double-tadpoles, or a generalized melon. Both cases are excluded given that the dipole is of type $I$.
\item If it is in the third configuration, suppose that channel $(2b)$ also disconnects. Then, the subgraph $B$ must be disconnected and the dipole is in the configuration of Fig.~\ref{fig:dipole_connected_d} with $A$, $B$ and $C$ subgraphs not necessarily connected. Then, channel $(2c)$ does not disconnect, otherwise $C$ would have to be disconnected and there would again be two generalized double-tadpoles. 
\end{itemize}
All in all, at least one channel does not disconnect if the dipole is of type $I$.
\begin{figure}[htbp]
\centering
\subfloat[]{\label{fig:dipole_connected_a}\includegraphics[scale=1]{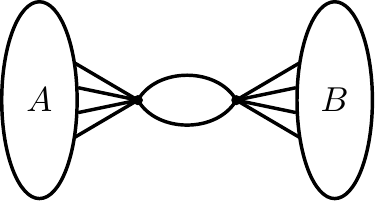}}
\hspace{1cm}
\subfloat[]{\label{fig:dipole_connected_b}\includegraphics[scale=1]{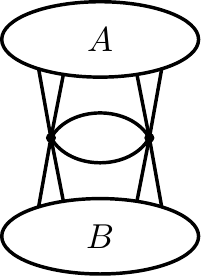}}
\hspace{1cm}
\subfloat[]{\label{fig:dipole_connected_c}\includegraphics[scale=1]{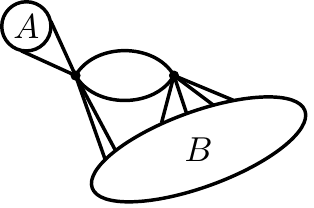}}
\hspace{1cm}
\subfloat[]{\label{fig:dipole_connected_d}\includegraphics[scale=1]{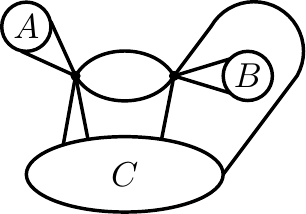}}
\caption{Possible configurations of a dipole which disconnects the graph upon deletion in channel $(2a)$.}
\label{fig:dipole_connected}
\end{figure}
  
The following lemma gives us tools to recursively remove dipoles from a stranded graph. 
\begin{lemma}\label{lemma:dipole}
Let $G$ be a stranded graph and $S$ a strict dipole subgraph of $G$. Call $G'$ the graph obtained after deletion of $S$ in the channel $c$, and assume $G'$ remains connected. There exists a conventional labeling of the external legs of $S$ (see Fig.~\ref{fig:channel_dipole}) such that: if $c$ is channel $(1),\,(2a),\,(2b)$ or $(2c)$, then it is possible to choose $G'$ such that 
$\omega(G)\geq \omega(G')$. 
\end{lemma}

\begin{proof}
We have two cases to consider: the dipole contains one internal face or none. In the latter case ($F(S)=0$), the two corners of the dipole can either be on the same external face or on two distinct ones, as represented in Fig.~\ref{fig:dipole_reduction}. In both those cases, this subset of strands can be reconfigured in such a way as to ensure that the dipole contains an internal face. Moreover, such a move does not affect the rest of the graph. We obtain in this way a graph $\tilde{G}$ containing a dipole subgraph $\tilde{S}$ such that: $F(\tilde{S})=1$ and $\omega(G) \geq \omega(\tilde{G}) + 1$. It is then clear that the Lemma will hold in general if we can prove it for configurations like $\tilde{G}$, which we now turn too.
\begin{figure}[htbp]
\centering
\includegraphics[scale=.8]{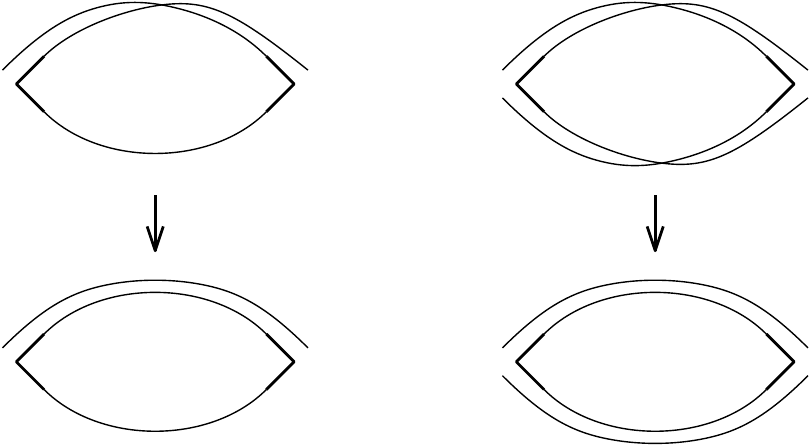}
\caption{Two types of configurations of the dipole when $F(S)=0$: the two internal corners can lie on the same external face (left), or on two distinct ones (right). In both cases, we can reconfigure this subset of strands in such a way as to ensure $F(S)=1$, without affecting the rest of the graph $G$.}
\label{fig:dipole_reduction}
\end{figure}

We can thus assume, without loss of generality, that $F(S)=1$, which makes it easier to determine the structure of the boundary graph $S_\partial$. We can proceed similarly as for single-tadpoles, and associate a $K_4$ subgraph to each of the two vertices in the dipole, which we represent in gray. We are left with a choice of pairing between eight additional half-edges, four of them attached to each $K_4$ subgraph, which we can represent in black. Given that the two internal corners of the dipole have been used to build up the internal face, any pairing of the black half-edges must connect one $K_4$ subgraph to the other. Consequently, we can again classify the allowed contractions in terms of the number and lengths of cycles with support on black edges only. The resulting boundary graphs can be labeled by partitions of 8 into even integers, yielding five possibilities: $8 = 6 + 2 = 4 + 4 = 4 + 2 + 2 = 2 + 2 + 2 + 2$. See Fig.~\ref{fig:dipole_config}.

\begin{figure}[htbp]
\centering
\begin{tabular}{ccc}
\subfloat[2+2+2+2]{\includegraphics[width=0.25\textwidth]{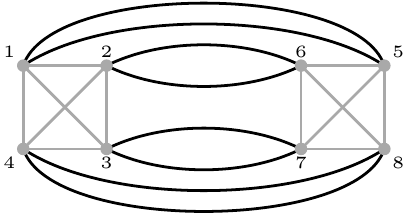}}
&
\subfloat[4+4]{\includegraphics[width=0.25\textwidth]{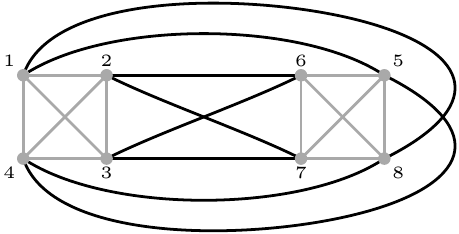}}
&
\subfloat[6+2]{\includegraphics[width=0.25\textwidth]{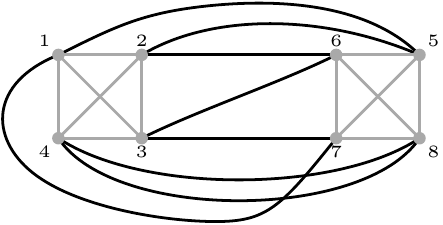}} \\
\subfloat[8]{\includegraphics[width=0.25\textwidth]{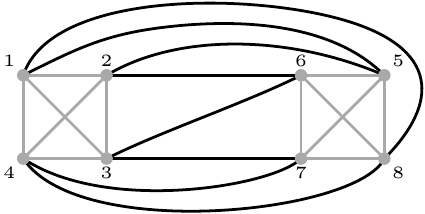}}
& &
\subfloat[4+2+2]{\includegraphics[width=0.25\textwidth]{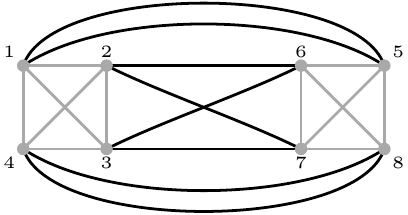}} 
\end{tabular}
\caption{The five possible boundary graphs of a dipole.}
\label{fig:dipole_config}
\end{figure}

Given that $G'$ remains connected, $F(S)=1$ and $V(S)=2$, Proposition~\ref{propo:general_deletion} guarantees that we can arrange the strands in such a way that:
\begin{equation}
\omega(G) \geq \omega(G') + 9 - d(S_\partial , B_c)\,,
\end{equation} 
where $B_c$ is the boundary graph characterizing the channel $c$. Hence, the looked-for bound will follow from $d(S_\partial , B_c) \leq 9$, which we now prove. 

We can first determine the flip distance between the five boundary graphs of Fig.~\ref{fig:dipole_config}, which is reported in Fig.~\ref{fig:distance}. As a result, it is sufficient to prove that:
\begin{itemize}
\item if $c=1$ and $S$ is in the configuration $2+2+2+2$, then $d(S_\partial , B_c )\leq 6$;
\item if $c=2a$ and $S$ is in the configuration $4+2+2$ or $4+4$, then $d(S_\partial , B_c )\leq 8$;
\item if $c=2b$ and $S$ is in the configuration $8$ or $4+2+2$, then $d(S_\partial , B_c )\leq 8$;
\item if $c=2c$ and $S$ is in the configuration $2+2+2+2$ or $4+4$, then $d(S_\partial , B_c )\leq 8$; if, on the other hand, $S$ is in configuration $6+2$, then $d(S_\partial , B_c )\leq 9$.
\end{itemize}

\begin{figure}[htbp]
\centering
\includegraphics[scale = 1]{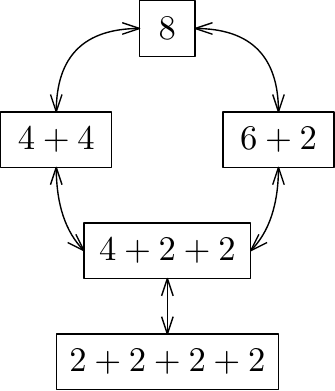}
\caption{Distance between the five dipole configurations: any two partitions connected by an edge are at flip distance one from each other.}
\label{fig:distance}
\end{figure}

\textbf{Channel $1$ (parallel channel).} To map the $2+2+2+2$ configuration to the parallel configuration, we need to cut all twelve grey edges. This can be achieved in $6$ flips. The other four dipole configurations being at distance at most $3$ from $2+2+2+2$, we always have $d(S_\partial , B_1) \leq 9$.

\textbf{Channel $(2a)$.} Let us first look at the grey edges. 
In the two configurations $4+4$ and $4+2+2$, we need to cut eight of the twelve grey strands; this requires $4$ flips. 
We then have to perform four more flips on pairs of black edges to obtain the boundary graph $B_c$. We thus have $d(S_\partial , B_c) \leq 8$ in the configurations $4+4$ and $4+2+2$. Since $8$ is at distance one from $4+4$, while $6+2$ and $2+2+2+2$ are at distance one from $4+2+2$, we conclude that $d(S_\partial , B_c) \leq 9$ in all cases.

\textbf{Channel $(2b)$.} We again need to perform $4$ flips on grey edges. In configurations $8$ and $4+2+2$, we can then obtain $B_{2b}$ after $4$ more flips on black edges. As $4+2+2$ is at distance one from $4+4$, $6+2$ and $2+2+2+2$, we always have $d(S_\partial , B_c) \leq 9$. 

\textbf{Channel $(2c)$.} As before, we have to cut eight of the twelve grey strands, which can be achieved in $4$ flips. Then, for the configurations $2+2+2+2$ and $4+4$, we can implement $4$ additional flips on black strands to obtain the boundary graph $B_{2c}$. Given that $8$ and $4+4$ are at distance one from either $2+2+2+2$ or $4+4$, we infer that $d(S_\partial , B_{2c}) \leq 9$ for these four configurations.  
We can finally check that the configuration $6+2$ can also be mapped to $B_{2c}$ in $9$ flips (4 flips on grey edges, and $5$ on black edges), as illustrated in Fig.~\ref{fig:moves_62}).
\begin{figure}[htbp]
\centering
\includegraphics[scale=1.2]{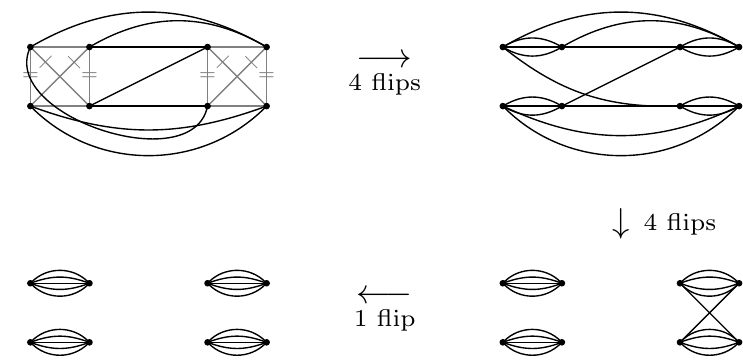}
\caption{Graphical proof that the configuration $6+2$ is at flip distance (at most) $9$ from $B_{2c}$.}
\label{fig:moves_62}
\end{figure}
\end{proof}

\subsection{Dipole-tadpole and quartic rung deletions}

We will now proceed with the deletion of the four-point \textit{dipole-tadpole} subgraphs. As for the deletion of single-tadpoles, there are three possible channels: parallel, cross and orthogonal. They are represented in Fig.~\ref{fig:dipole_tadpole_channel}.

\begin{figure}[htbp]
\centering
\includegraphics[scale=1]{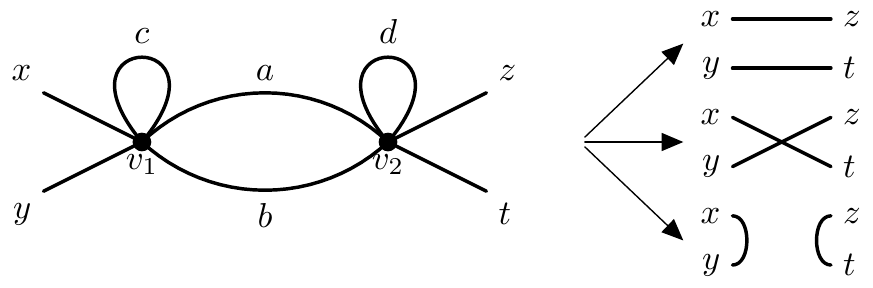}
\caption{The three channels of deletion for the dipole-tadpole four-point subgraph.}
\label{fig:dipole_tadpole_channel}
\end{figure}

In the following Lemma, we prove that it is always advantageous to delete a dipole-tadpole in the orthogonal channel, and in at least one of the parallel or cross channels. 
\begin{lemma}
Let $G$ be a stranded graph and $S$ a strict dipole-tadpole subgraph of $G$. Call $G'$ the graph obtained after deletion of $S$ in the channel $c$ and assume $G'$ remains connected. 

\begin{itemize}
\item[(a)] Suppose $c$ is the orthogonal channel. Then, it is possible to choose $G'$ such that $\omega(G)\geq \omega(G')$. 

\item[(b)] Suppose $c$ is the parallel (resp. cross) channel, and call $G''$ the graph obtained after deletion of $S$ in the cross (resp. parallel) channel. Then, if it is not possible to choose $G'$ such that $\omega(G)\geq \omega(G')$, provided that $G''$ remains connected, it is possible to choose $G''$ such that $\omega(G)\geq \omega(G'')$.

\end{itemize}
\label{lemma:dipole_tadpole}
\end{lemma}

\begin{proof}

Call $v_1$ the vertex connected to $(x,y)$ and $v_2$ the other vertex. We first observe that a number of situations can be dealt with Lemma \ref{lemma:single_tadpole}, through successive deletions of the tadpoles at $v_1$ and $v_2$. We distinguish three cases. 
\begin{itemize}
\item If neither $v_1$ nor $v_2$ are of type $1+1+2$, then we can perform the move in any channel (and in particular in the orthogonal channel): delete $v_1$ in the parallel channel,  then $v_2$ in the desired channel. This is illustrated in Fig.~\ref{fig:lemma_diptad_case1}. 

\item If one of them (say $v_1$) is of type $1+1+2$, there is a single channel $c$ in which $v_1$ can be deleted (in which case we gain a factor $1/N$). We then distinguish two subcases. If the channel $c$ pairs $(x,y)$ and $(a,b)$, we first delete $v_2$ in the parallel channel, then $v_1$ in the channel $c$. This implements the orthogonal deletion, as illustrated in Fig.~\ref{fig:lemma_diptad_case2a}. 

If, on the other hand, the channel $c$ is parallel/cross (say, it pairs $(x,a)$ and $(y,b)$), we delete $v_1$ in this channel, and then $v_2$ in any desired channel, as shown in Fig.~\ref{fig:lemma_diptad_case2b}. 

\item Assume, finally, that both $v_1$ and $v_2$ are of type $1+1+2$. If at least one of them (say $v_1$) can be deleted in the parallel or cross channel, we perform this move and then delete $v_2$ in the appropriate channel: the first deletion yields a $1/N$ suppression, while the second deletion brings at most a factor of $N$. We can therefore perform the deletion in all three channels, as illustrated in Fig.~\ref{fig:lemma_diptad_case3a}. 

On the other hand, if both $v_1$ and $v_2$ can only be deleted in the orthogonal channel, we are still able to implement the orthogonal channel: delete $v_1$ in the orthogonal channel (yields a $1/N$ suppression), then $v_2$ in the parallel channel (results in an additional factor of $N$). This is illustrated in Fig.~\ref{fig:lemma_diptad_case3b}. 

\end{itemize}

\begin{figure}[htbp]
\centering
\begin{tabular}{c}
\subfloat[Neither $v_1$ nor $v_2$ are of type $1+1+2$\label{fig:lemma_diptad_case1}]{\includegraphics[scale=.9]{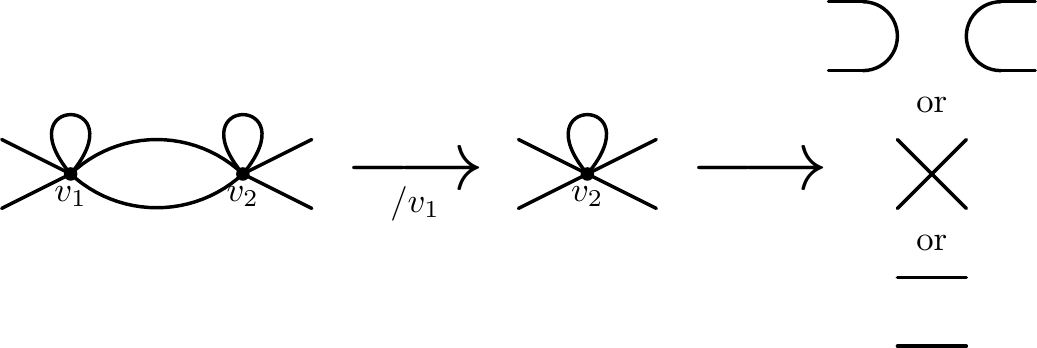}} \\
\subfloat[Only $v_1$ is of type $1+1+2$ and can be deleted in the orthogonal channel\label{fig:lemma_diptad_case2a}]{\includegraphics[scale=0.9]{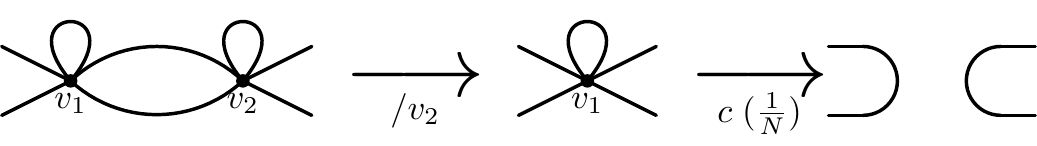}} \\
\subfloat[Only $v_1$ is of type $1+1+2$ and can be deleted in the parallel/cross channel\label{fig:lemma_diptad_case2b}]{\includegraphics[scale=0.9]{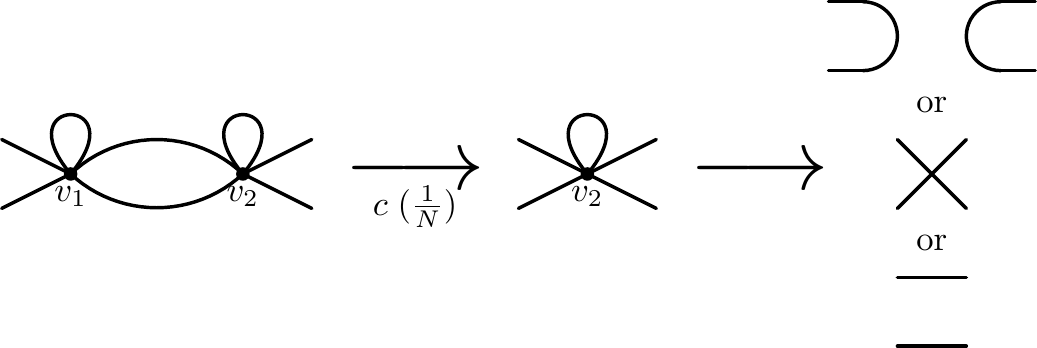}} \\
\subfloat[$v_1$ and $v_2$ are of type $1+1+2$, and $v_1$ can be deleted in the parallel/cross channel\label{fig:lemma_diptad_case3a}]{\includegraphics[scale=0.9]{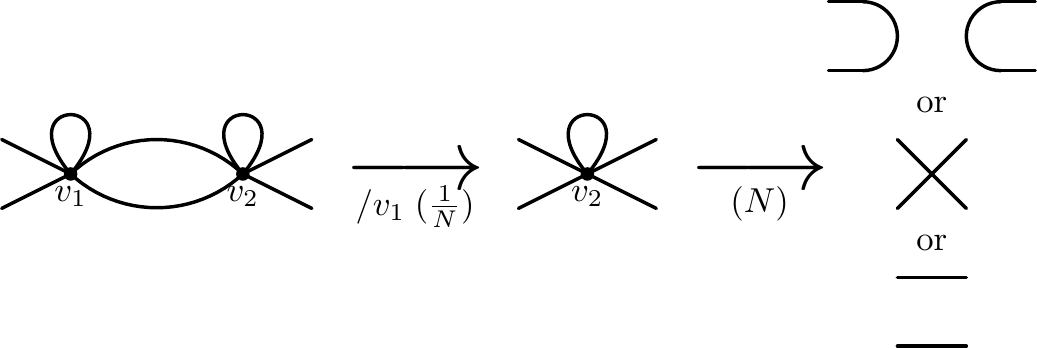}} \\
\subfloat[$v_1$ and $v_2$ are of type $1+1+2$, and both can be deleted in the cross channel\label{fig:lemma_diptad_case3b}]{\includegraphics[scale=0.9]{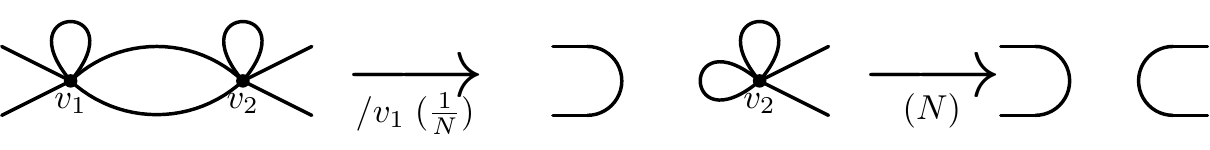}} 
\end{tabular}
\caption{Dipole-tadpole deletions from tadpole deletions (Lemma \ref{lemma:dipole_tadpole}).}
\label{fig:dip_tad_1st}
\end{figure}

All in all, the only subcases left to investigate are about the parallel/cross channels, in the following situation: $v_2$ is of type $1+1+2$ and can be deleted in the orthogonal channel; furthermore, if $v_1$ is of type $1+1+2$, its easy channel is also the orthogonal one. In particular, we can now assume that $v_2$ is in one of the two configurations shown in Fig.~\ref{fig:config_v2}.

\begin{figure}[htbp]
\centering
\captionsetup[subfigure]{labelformat=empty}
\subfloat[]{\includegraphics[scale=.6]{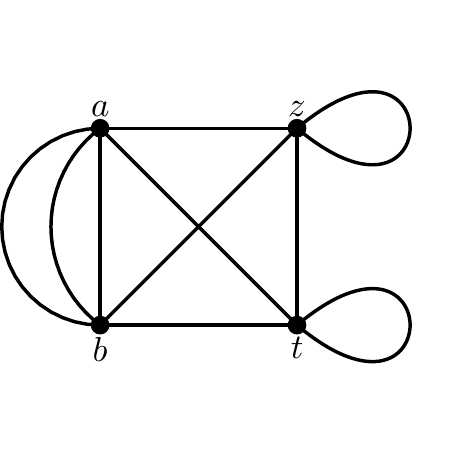}} \hspace{2cm}
\subfloat[]{\includegraphics[scale=.6]{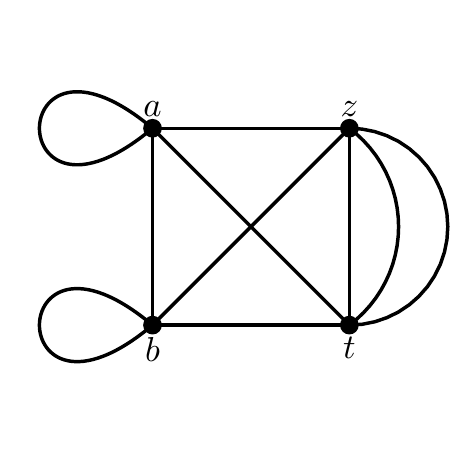}}
\caption{Special configurations of $v_2$.}\label{fig:config_v2}
\end{figure}

\medskip

We now determine the allowed boundary graphs in the remaining configurations. The Lemma will follow if we can prove that $d(S_\partial,B_c) \leq 10 - F(S)$, with $c$ either the parallel or cross channel. For this purpose, it will be sufficient to consider equivalent classes of boundary graphs under exchanges of $x$ and $y$, $z$ and $t$, as well as $(x,y)$ and $(z,t)$.
As $v_2$ can only be in the configurations of Fig.~\ref{fig:config_v2}, the full boundary graph must have one of the two structures depicted in Fig.~\ref{fig:full_bndy_v2}. We distinguish three cases.

\begin{figure}[htbp]
\centering
\includegraphics[scale=1]{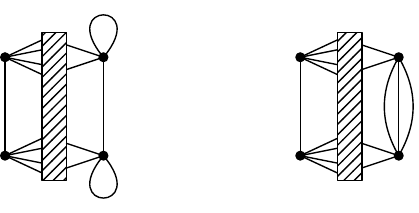}
\caption{Two special configurations for the boundary graph of a dipole-tadpole.}
\label{fig:full_bndy_v2}
\end{figure}

Assume first that $S_\partial$ has two connected components. In that case, we necessarily have $F(S)\leq 3$. Indeed, we infer from the boundary structure of tadpoles in Fig.~\ref{fig:tadpole_config} that at most three corners are available at $v_1$ (resp. $v_2$) to support faces of length two or higher. Furthermore, each such face will use at least one corner. But when $S_\partial$ is disconnected, two of those corners must already be occupied by strands that connect $z$ to $t$ (resp. $x$ to $y$). We have therefore at most one face of length two or higher. Remembering that each tadpole line can support an additional face, we obtain the claimed bound: $F(S) \leq 3$. The six inequivalent boundary graphs which can be realized under those conditions are represented in Fig.~\ref{fig:config_ortho_del_diptad}. It is straightforward to check that $d(S_\partial , B_c) \leq 7 \leq 10 - F(S)$, where $c$ is e.g. the parallel channel. 

\begin{figure}[htbp]
\centering
\includegraphics[scale=1]{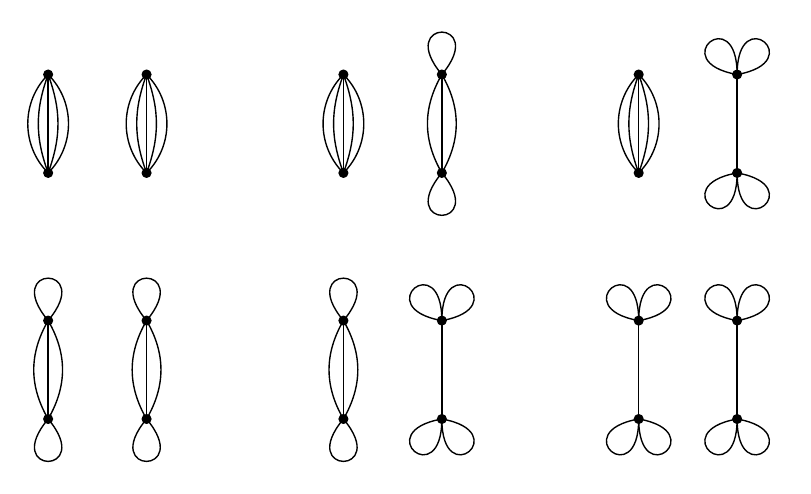}
\caption{The six inequivalent boundary graphs with two connected components.}
\label{fig:config_ortho_del_diptad}
\end{figure}

We can now assume that $S_\partial$ is connected. We note that there is at least two edges connecting the pair of vertices $(x,y)$ and $(z,t)$. Indeed, if there were only one such edge, we would have an additional seven half-edges to match from the pair $(x,y)$, which by parity is impossible. Let us first assume that there are exactly two edges connecting $(x,y)$ to $(z,t)$, and that furthermore, they are not both in the same parallel/cross channel.  For definiteness, and without loss of generality, we can suppose those edges are between $y$ and $t$, and between $y$ and $z$. There are then four possible boundary graphs, as represented in Fig.~\ref{fig:bndy_exactly_1para}. Compared to the previous paragraph, one additional corner is available at $v_1$ (resp. $v_2$) to build up faces of length two or higher. This leads to the weaker bound $F(S) \leq 4$. However, a straightforward inspection of the graphs of Fig.~\ref{fig:bndy_exactly_1para} shows that $d(S_\partial , B_c) \leq 6 \leq 10 - F(S)$ always holds (and is saturated for the last configuration), where $c$ is the parallel channel. 

\begin{figure}[htbp]
\centering
\includegraphics[scale=1]{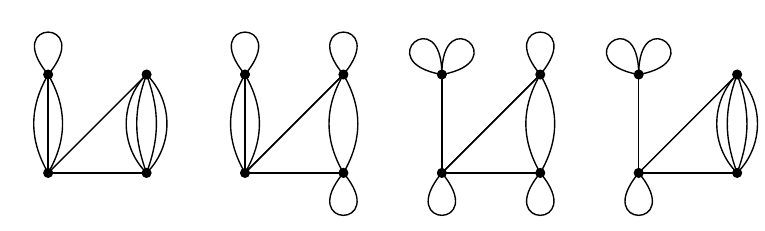}
\caption{The four inequivalent boundary graphs with exactly one edge in the cross channel, and one edge in the parallel channel.}
\label{fig:bndy_exactly_1para}
\end{figure}

Let us finally suppose that there are at least two edges in the same parallel or cross channel. Without loss of generality, we can assume it to be the parallel channel. After a straightforward (but tedious) inspection, we find another sixteen inequivalent boundary graphs, which we have depicted in Fig.~\ref{fig:bndy_16}. Any such configuration verifies $d(S_\partial , B_c) \leq 5 \leq 10 - F(S)$, where we have used that $F(S) \leq 5$ for any dipole-tadpole $S$. 

This concludes the proof.
\begin{figure}[htbp]
\centering
\includegraphics[scale=1]{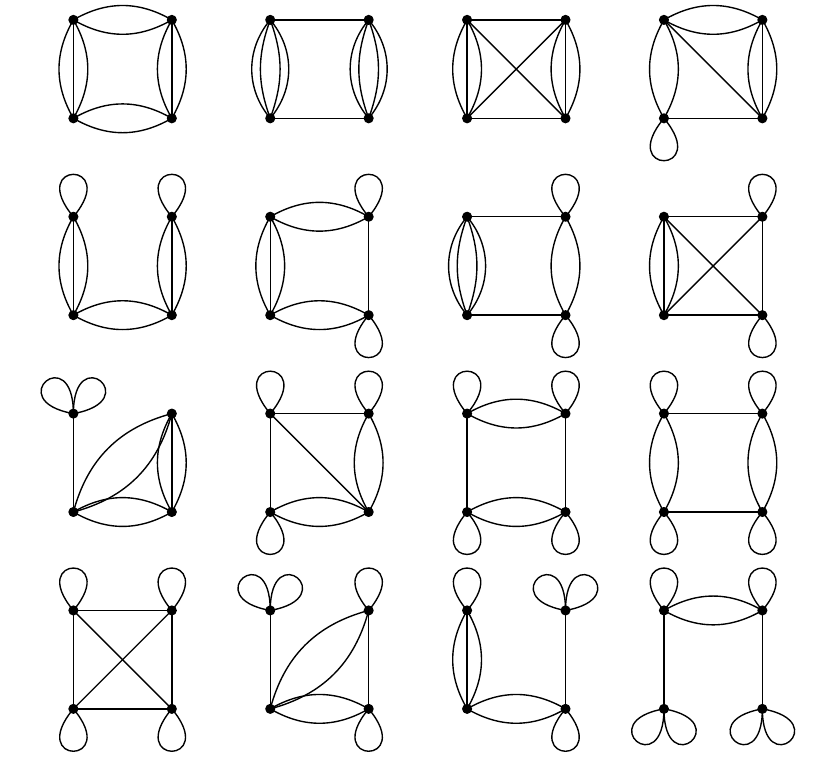}
\caption{The sixteen inequivalent boundary graphs with at least two edges in the parallel channel.}
\label{fig:bndy_16}
\end{figure}
\end{proof}

We will also need to delete a particular type of four-point subgraph represented in Fig.~\ref{fig:deletion_melon_4pt}. We call this type of graph \textit{quartic rung}. There are three different channels of deletion but we will only consider one (as represented in Fig.~\ref{fig:deletion_melon_4pt}). 

\begin{figure}[htbp]
\centering
\includegraphics[scale=1]{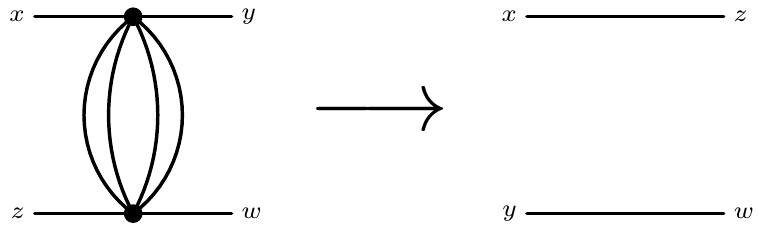}
\caption{Deletion of a quartic rung subgraph.}
\label{fig:deletion_melon_4pt}
\end{figure}

\begin{lemma}
Let $G$ be a stranded graph and $S$ a strict quartic rung subgraph of $G$. Call $G'$ the graph obtained after deletion of $S$ in the channel depicted in Fig.~\ref{fig:deletion_melon_4pt} and assume $G'$ remains connected. 

It is always possible to perform the deletion in such a way that $\omega(G)\geq \omega(G')$. 
\label{lemma:melon4pt}
\end{lemma}

\begin{proof}

We want to prove the following bound:
\begin{equation}
d(S_{\partial},B)\leq 10-F(S)
\end{equation}
with $B$ the boundary graph of the deletion channel depicted in Fig.~\ref{fig:deletion_melon_4pt} with no self-loops.

First notice that we can apply the argument of Fig.~\ref{fig:dipole_reduction} (from the proof of Lemma~\ref{lemma:dipole}) to each of the six dipole subgraphs of $S$. This allows to assume, without loss of generality, that $F(S)= 6$. 
In such a situation, $S_\partial$ is one of the five boundary graphs represented in Fig.~\ref{fig:bndy_lemma5}: there is one edge from $x$ to $y$, one edge from $z$ to $w$, and each of the eight remaining edges connects a vertex of the pair $(x,y)$ to a vertex of the pair $(z,w)$. It is then clear that eight edges need to be reconfigured to map $S_\partial$ to $B$. Given that those do not include any self-loop, we conclude that $d(S_\partial, B) \leq 4 = 10 - F(S)$. 
\begin{figure}[htbp]
\centering
\includegraphics[scale=1]{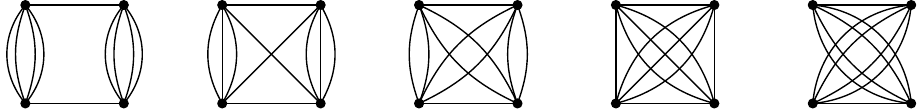}
\caption{The five possible boundary graphs of a quartic rung $S$ with $F(S)=6$.}
\label{fig:bndy_lemma5}
\end{figure}
\end{proof}

\subsection{Two-point subgraph deletions}\label{subsec:main_proof}

The following Lemma will allow us to find inductive bounds on two-particle reducible graphs.
\begin{lemma}\label{lemma:2PR}
Let $G$ be a two-particle reducible stranded graph. That is, $G$ has the structure represented in Fig.~\ref{fig:2PR}, where $S_1$ and $S_2$ are (non-empty) two-point subgraphs. Denote by $G_1$ (resp. $G_2$) the graph obtained by closing $S_1$ (resp. $S_2$) with an unbroken edge $e_1$ (resp. $e_2$). It is possible to choose $e_1$ and $e_2$ such that:
\begin{equation}
\omega(G) \geq \omega(G_1) + \omega(G_2) - 4\,.
\end{equation}
Moreover, if $S_2$ has no tadpole and no dipole, then 
\begin{equation}
\omega(G) \geq \omega(G_1)+1 \,.
\end{equation}

\begin{figure}[htbp]
\centering
\includegraphics[scale=1]{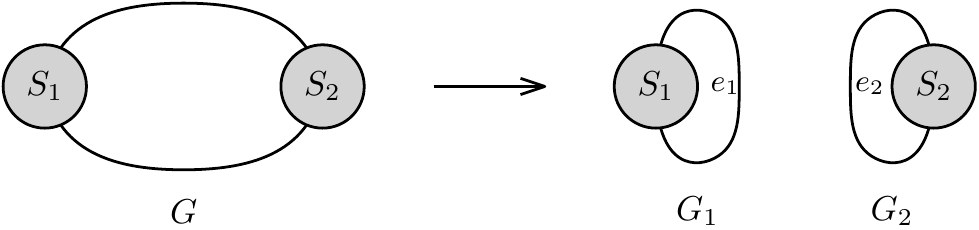}
\caption{Deletion of two-particle reducible components.}
\label{fig:2PR}
\end{figure}
\end{lemma}
\begin{proof}
The boundary graph of $S_1$ (resp. $S_2$) is one of the three configurations shown in Fig.~\ref{fig:propagator_bndy}. It is then apparent that we can choose $e_1$ and $e_2$ such that:\footnote{Note that we are constrained by the fact that $e_1$ and $e_2$ are required to be unbroken. Without this restriction, we could ensure that $F(G_1) = F(S_1) + 5$ and $F(G_2) = F(S_2) + 5$.}
\begin{equation}
F(G_1) \geq F(S_1) +3 \quad \mathrm{and} \quad F(G_2) \geq F(S_2) + 3\,.
\end{equation}
Furthermore, it is clear that $F(G) \leq F(S_1) + F(S_2) + 5$ (there are ten strands in $G\setminus(S_1 \cup S_2)$, and they all belong to faces of length at most two). Hence, $F(G_1)+ F(G_2) - F(G) \geq 1$, which is equivalent to $\omega(G_1)+\omega(G_2) - 4 \leq \omega(G)$.
Finally, if we assume that $S_2$ has no tadpole or dipole, it is clear that $G_2$ has at most one dipole or one tadpole. Hence, 
\begin{equation}
\omega(G_2)=5+\sum_{k \geq 1} \frac{k-3}{3}F_k (G_2) \geq 5-2/3
\end{equation}
which implies $\omega(G_2)\geq 5$ (since $\omega \in \mathbb{N}$). As a consequence, $\omega(G) \geq \omega(G_1) +1$.
\end{proof}

Finally, to prove the main result of this section (Proposition~\ref{prop:positive_degree}), we will also rely on a number of special two-point moves, which we gather in the next Lemma.  
\begin{lemma}
Let $G$ be a stranded graph and $S$ a strict subgraph of $G$. Suppose $S$ is one of the two-point subgraphs of Fig.~\ref{fig:particular_cases}. Call $G'$ the graph obtained from $G$ by substituting $S$ with an unbroken edge $e$. 
Then, it is always possible to choose $e$ in such a way that $\omega(G)\geq \omega(G')+1$.

\begin{figure}[H]
\captionsetup[subfigure]{labelformat=empty}
\begin{center}
\begin{tabular}{cccc}
\subfloat[$H_0$]{\includegraphics[scale=0.9]{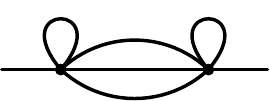}}
&
\subfloat[$H_1$]{\includegraphics[scale=0.9]{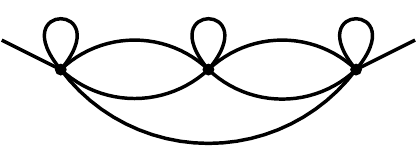}}
&
\subfloat[$H_2$]{\includegraphics[scale=0.9]{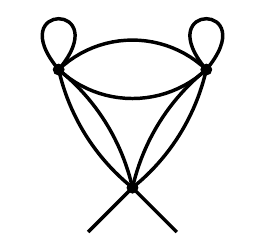}}
&
\subfloat[$H_3$]{\includegraphics[scale=0.9]{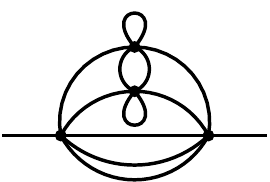}}
\\
\subfloat[$H_4$]{\includegraphics[scale=0.9]{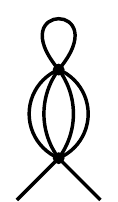}}
&
\subfloat[$H_{5}$]{\includegraphics[scale=0.9]{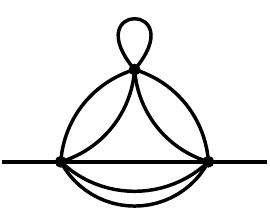}}
&
\subfloat[$H_{6}$]{\includegraphics[scale=0.9]{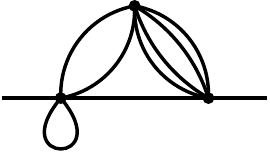}}
&
\subfloat[$H_{7}$]{\includegraphics[scale=0.9]{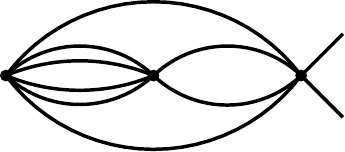}}
\\
\subfloat[$H_{8}$]{\includegraphics[scale=0.9]{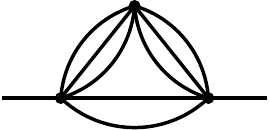}}
&
\subfloat[$H_{9}$]{\includegraphics[scale=0.9]{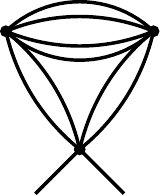}}
&
\subfloat[$H_{10}$]{\includegraphics[scale=0.9]{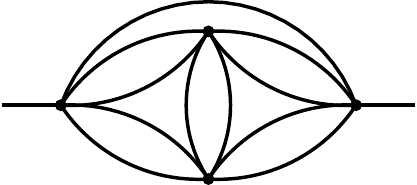}}
&
\subfloat[$H_{11}$]{\includegraphics[scale=0.9]{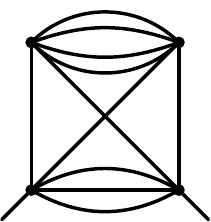}}
\\
\subfloat[$H_{12}$]{\includegraphics[scale=0.9]{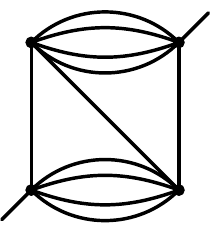}}
&
&
&
\end{tabular}
\end{center}
\caption{Particular two-point subgraphs which can always decrease the degree upon deletion.}
\label{fig:particular_cases}
\end{figure}
\label{lemma:particular_cases}
\end{lemma}
\begin{proof}
See Appendix~\ref{app:particular}.
\end{proof}

\begin{lemma}\label{lemma:two-point_deletion}
Let $G$ be a stranded graph with no double-tadpole, no melon, and no \emph{separating dipole-tadpole} (as introduced in Definition~\ref{def:dipole-tadpole} and Fig.\ref{fig:dipole-tadpole}). Suppose there exists a proper two-point subgraph $H \subset G$, which can be deleted (i.e. replaced by an unbroken edge) in a way that strictly decreases the degree. Then, there exists a graph $G'$, with no double-tadpole and no melon, such that:
\begin{equation}
V(G') < V(G) \quad \mathrm{and} \quad \omega(G) \geq \omega(G')\,.
\end{equation}  
\end{lemma}
\begin{proof}
Let us call $G_1$ the (connected) graph obtained from $G$ by deletion of $H$, and such that $\omega(G) \geq \omega(G_1)+1$.

If $G_1$ does not contain double-tadpoles or melons, we take $G' = G_1$.

If, on the other hand, $G_1$ contains a melon, then $G$ was in the configuration of Fig.~\ref{fig:config_deletion_2pt}, namely: the subgraph $H$ is adjacent to a quartic rung. 
In this case, we define $G'$ as the graph obtained by deletion of the quartic rung, following Lemma~\ref{lemma:melon4pt}. It is clear that this move cannot create melons or double-tadpoles, and by Lemma~\ref{lemma:melon4pt}, $\omega(G) \geq \omega(G')$. 
\begin{figure}[H]
\centering
\captionsetup[subfigure]{labelformat=empty}
\includegraphics[scale=1]{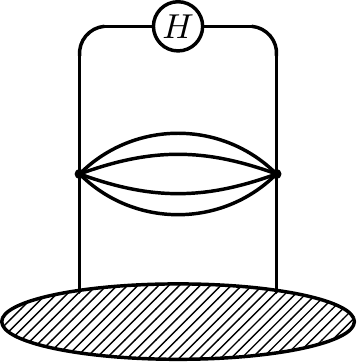}
\caption{Configuration of the graph generating a melon upon deletion of $H$.}
\label{fig:config_deletion_2pt}
\end{figure}

Finally, $G_1$ can contain a double-tadpole, in which case $G$ contains the subgraph depicted on the left side of Fig.~\ref{fig:deletion_2pt}. As illustrated in the same figure, we can subsequently remove the double-tadpole from $G_1$ to obtain a graph $G_2$ such that $\omega(G_1) \geq \omega(G_2) - 1$ (following Lemma~\ref{lemma:double_tadpole_deletion}). As a result, $\omega(G) \geq \omega(G_2)$. $G_2$ cannot contain a double-tadpole, otherwise $G$ would contain a separating dipole-tadpole (as illustrated in Fig.~\ref{fig:ind_dipole_tadpole_a}), which is excluded by assumption. If $G_2$ does not contain a melon either, we define $G' = G_2$ and conclude. If, on the other hand, $G_2$  contains a melon, then $G$ contains a subgraph with a quartic rung as represented in Fig.~\ref{fig:config_3_dbtad}. In that case, we can again invoke Lemma~\ref{lemma:melon4pt} to delete the quartic rung, and obtain a graph $G'$ with no double-tadpole or melon.
\begin{figure}[H]
\centering
\captionsetup[subfigure]{labelformat=empty}
\includegraphics[scale=1]{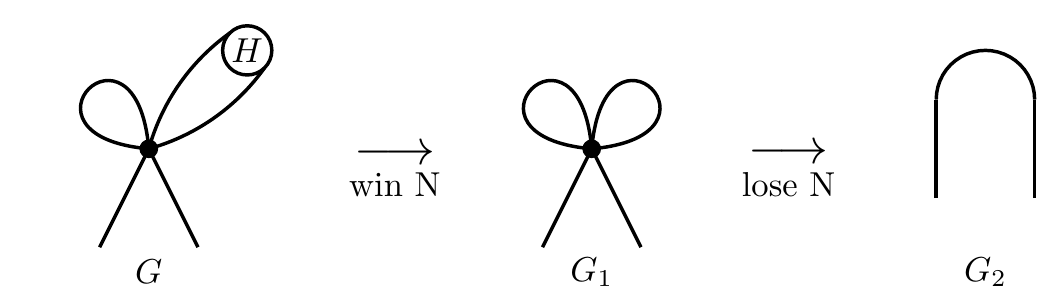}
\caption{Steps for deleting $H$ if it appears within a generalized double-tadpole. We gain a factor $N$ with the first step and lose one with the second step.}
\label{fig:deletion_2pt}
\end{figure}

\begin{figure}[H]
\centering
\captionsetup[subfigure]{labelformat=empty}
\includegraphics[scale=1]{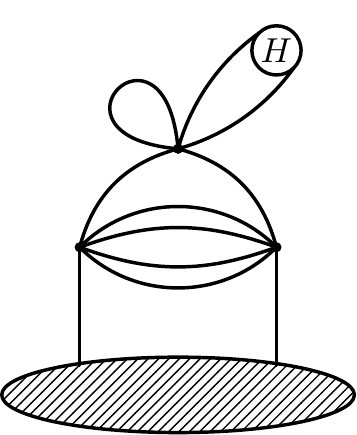}
\caption{Configuration leading to a melon when deleting the generalized double-tadpole containing $H$.}
\label{fig:config_3_dbtad}
\end{figure}
\end{proof}

\subsection{Main proposition}

\begin{proposition}
Let $G$ be a stranded graph. If $G$ has no double-tadpole and no melon, then $\omega(G) \geq 0$. 
\label{prop:positive_degree}
\end{proposition}

\begin{proof}
If $G$ has no short face or no vertex, we have already seen that $\omega(G) \geq 0$. In all other cases, we proceed by induction on the number of vertices. From now on, we assume that $V(G)\geq 1$, and that $G$ contains at least one tadpole or one dipole.  

Even if $G$ cannot have double-tadpoles or melons, generalized double-tadpoles and melons are still allowed. To avoid difficulties with such subgraphs, we will first deal with type-$I$ dipoles and tadpoles, and study the more involved case of type-$II$ configurations separately (see subsection~\ref{sec:types} and Fig.~\ref{fig:not_easy} for definitions).  

We now proceed with an exhaustive graph-theoretic distinction of cases. In each situation, we will look for a strict subgraph of $G$ which can be deleted without increasing the degree, and while preserving our combinatorial constraints (namely: connectedness, the absence of melons or double-tadpoles, as well as the absence of broken or doubly-broken edges). From the induction hypothesis, it will then follow that $\omega(G) \geq 0$.

\begin{figure}[htbp]
\centering
\subfloat[\label{fig:ind_dipole_tadpole_a}]{\includegraphics[scale=1.2]{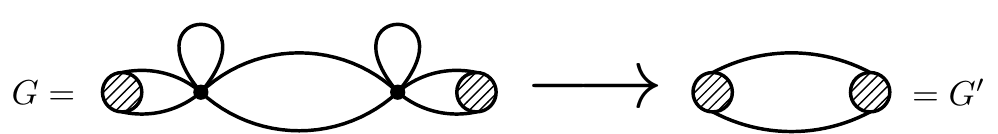}}
\\
\subfloat[\label{fig:ind_dipole_tadpole_b}]{\includegraphics[scale=1.2]{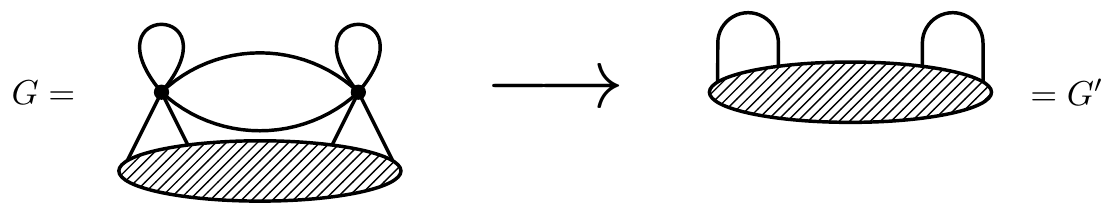}}
\caption{The two possible primary configurations of a dipole-tadpole. It can be separating (top), in which case the graph remains connected upon deletion in the parallel/cross channel; or non-separating (bottom), in which case the graph remains connected upon deletion in the orthogonal channel.}
\label{fig:ind_dipole_tadpole}
\end{figure}

\begin{itemize}
\item[\emph{Case A.}] First suppose that there exists a separating dipole-tadpole in $G$ (see Fig.~\ref{fig:ind_dipole_tadpole_a}). Performing a deletion of this dipole-tadpole in the parallel or cross channel, we obtain a graph $G'$ which cannot contain double-tadpoles or melons. Furthermore, by Lemma~\ref{lemma:dipole_tadpole} (b), we can make sure that $\omega(G)\geq \omega(G') \geq 0$ (the last inequality follows from the induction hypothesis, because $G'$ has strictly fewer vertices than $G$).
\end{itemize}

We can now assume that there are no more separating dipole-tadpoles.

\begin{itemize}
\item[\emph{Case B.}] Suppose that there exists a non-separating dipole-tadpole in $G$ (see Fig.~\ref{fig:ind_dipole_tadpole_b}). We can then use Lemma~\ref{lemma:dipole_tadpole} to delete the latter in the orthogonal channel, and obtain a connected graph $G'$ with $\omega(G) \geq \omega(G')$.
If $G'$ has no double-tadpole and no melon, then we are done. If, on the other hand, $G'$ has a double-tadpole, then $G$ is in either one the configuration represented in Fig.~\ref{fig:ind_dipole_tadpole_ortho_a} or contains the subgraph $H_2.$
If, instead, $G'$ has a melon, then $G$ either contains the two-point subgraph $H_4$ (see Fig.~\ref{fig:particular_cases}) or contains the four-point function depicted in Fig.~\ref{fig:ind_dipole_tadpole_ortho_c}.

\begin{figure}[H]
\centering
\subfloat[]{\label{fig:ind_dipole_tadpole_ortho_a}\includegraphics[scale=1]{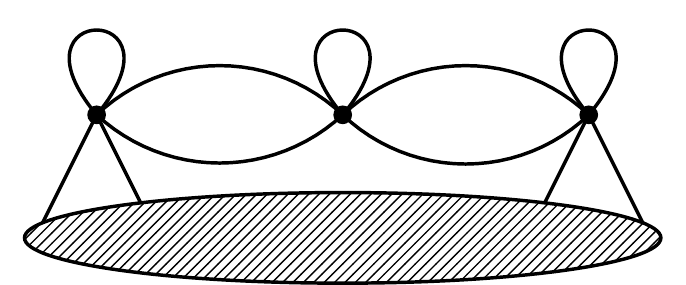}}
\hspace{1cm}
\subfloat[]{\label{fig:ind_dipole_tadpole_ortho_c}\includegraphics[scale=1]{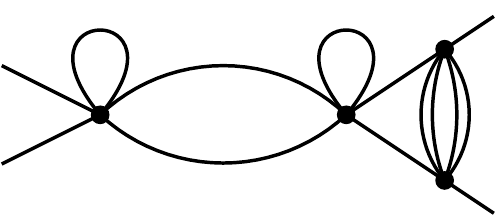}}
\caption{Two configurations of a non-separating dipole-tadpole that can create double-tadpoles or melons upon deletion in the orthogonal channel. 
}
\label{fig:ind_dipole_tadpole_ortho}
\end{figure}

In the two situations of Fig.~\ref{fig:ind_dipole_tadpole_ortho}, we can try to perform the deletion of the dipole-tadpole in the parallel channel, which cannot disconnect the graph. One may however create double-tadpoles or melons, in which case $G$ contains one of the subgraph $H_1$, $H_2$, $H_3$, or the two-point subgraph depicted in Fig.~\ref{fig:ind_dipole_tadpole_2pt}. In the latter case, we can use Lemma~\ref{lemma:melon4pt} to remove the quartic rung, and reduce the problem to the situation in which $G$ contains the subgraph $H_0$. In conclusion, we have shown that $G$ always contains a subgraph $H_i$ covered by Lemma~\ref{lemma:particular_cases}. We can therefore apply Lemma~\ref{lemma:two-point_deletion} as a last step, which outputs a suitable graph $G''$ with $\omega(G)\geq \omega(G'')$.

\begin{figure}[H]
\centering
\includegraphics[scale=1]{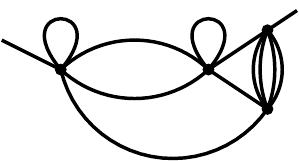}
\caption{Configuration of a dipole-tadpole subgraph which can generate a melon in the parallel or cross channel.}
\label{fig:ind_dipole_tadpole_2pt}
\end{figure}
 
\end{itemize}

We can now assume that there are no more dipole-tadpole subgraphs in $G$. 

\begin{itemize}
\item[\emph{Case C.}] Suppose that there exists a type-$I$ tadpole in $G$. Then the graph remains connected upon deletion of this tadpole in any channel. By application of Lemma~\ref{lemma:single_tadpole}, we obtain a graph $G'$ with $\omega(G)\geq \omega(G')$. 
\end{itemize}
\begin{itemize}
\item If $G'$ has a double-tadpole, then $G$ either contained a dipole-tadpole or the subgraph $H_4$ depicted in Fig.~\ref{fig:particular_cases}. The first situation has already been excluded, and we can use Lemmas~\ref{lemma:particular_cases} and \ref{lemma:two-point_deletion} in combination to deal with the second. 

\item If $G'$ has a melon, then  $G$ either contained the subgraph $H_5$ from Fig.~\ref{fig:particular_cases} or the four-point graph of Fig.~\ref{fig:ind_easy_tadpole}.
In the second case, we can use Lemma~\ref{lemma:melon4pt} to delete the quartic rung. This last step cannot create a melon. If it does not create a double-tadpole either, we conclude.  If it does, then $G$ necessarily contained the subgraph $H_6$ from Fig.~\ref{fig:particular_cases}. We can deal with this situation, as well as with the configuration $H_5$, by application of Lemmas~\ref{lemma:particular_cases} and \ref{lemma:two-point_deletion}. 
\end{itemize}

\begin{figure}[htbp]
\centering
\includegraphics[scale=1]{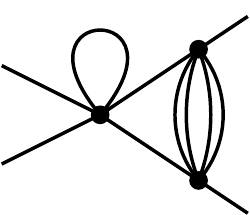}
\caption{A four-point subgraph that can lead to the creation of a melon when deleting a type-$I$ tadpole.}
\label{fig:ind_easy_tadpole}
\end{figure}

From now on, we assume that there is no type-$I$ tadpole in $G$. 

\begin{itemize}
\item[\emph{Case D.}] Suppose that $G$ contains a type-$I$ dipole. We can then attempt to delete this dipole in one of the channels covered by Lemma~\ref{lemma:dipole}. 
\end{itemize}

\begin{itemize}
\item If the channels $(2a),(2b)$ and $(2c)$ all disconnect the graph, we instead perform the deletion in the parallel channel to obtain a new graph $G'$. Thanks to our assumption that the deleted dipole was of type-$I$, $G'$ is necessarily connected (see our earlier discussion around Fig.~\ref{fig:dipole_connected}). Moreover, $G'$ cannot have double-tadpoles or melons. By application of Lemma~\ref{lemma:dipole}, we can ensure that $\omega(G) \geq \omega(G')$.

\item If one of the three channels $2$ does not disconnect the graph, for example the channel $(2a)$ (which we can assume without loss of generality), we perform the deletion in this channel. However, this can create a double-tadpole. In this case $G$ was in one of the first three configurations of Fig.~\ref{fig:ind_config_easy_dipole}.

\begin{figure}[htpb]
\centering
\captionsetup[subfigure]{labelformat=empty}
\subfloat[(i)]{\includegraphics[scale=1]{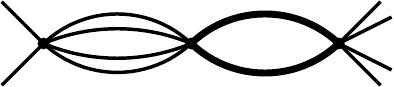}}
\hspace{1cm}
\subfloat[(ii)]{\includegraphics[scale=1]{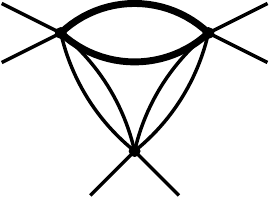}}
\hspace{1cm}
\subfloat[(iii)]{\includegraphics[scale=1]{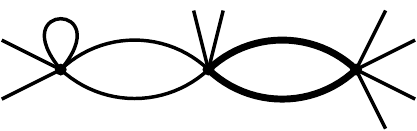}}
\\
\subfloat[(iv)]{\includegraphics[scale=1]{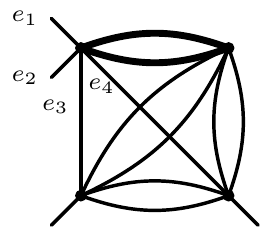}}
\hspace{1cm}
\subfloat[(v)]{\includegraphics[scale=1]{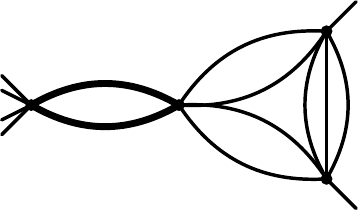}}
\hspace{1cm}
\subfloat[(vi)]{\includegraphics[scale=1]{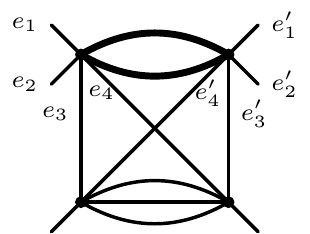}}
\hspace{1cm}
\subfloat[(vii)]{\includegraphics[scale=1]{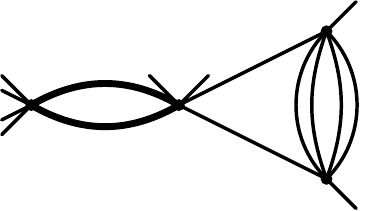}}
\caption{Seven configurations of a dipole (in bold) that lead to the creation of a melon or a double-tadpole. The first three lead to the creation of a double-tadpole while the four other lead to the creation of a melon.}
\label{fig:ind_config_easy_dipole}
\end{figure}

This can also create a melon. In this case, $G$ was either in one of the last four configurations of Fig.~\ref{fig:ind_config_easy_dipole} or contains the subgraph $H_{10}$ depicted in Fig.~\ref{fig:particular_cases}.

$H_{10}$ can be dealt with by application of Lemmas~\ref{lemma:particular_cases} and \ref{lemma:two-point_deletion}. We have to look at the other configurations in more detail.

\begin{itemize}
\item[(i)] We delete the dipole in the parallel channel instead. If this creates a double-tadpole, the graph is either in the configuration of Fig.~\ref{fig:ind_deletion_easy_dipole_a} or contains the subgraph $H_7$ of Fig.~\ref{fig:particular_cases}.

\begin{figure}[htbp]
\centering
\subfloat[]{\label{fig:ind_deletion_easy_dipole_a}\includegraphics[scale=1]{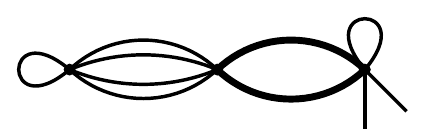}}
\hspace{0.125cm}
\subfloat[]{\label{fig:ind_deletion_easy_dipole_b}\includegraphics[scale=1]{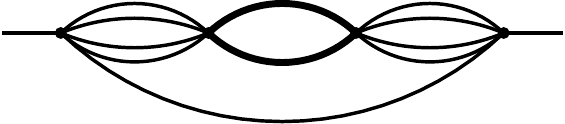}}
\hspace{0.125cm}
\subfloat[]{\label{fig:ind_deletion_easy_dipole_c}\includegraphics[scale=1]{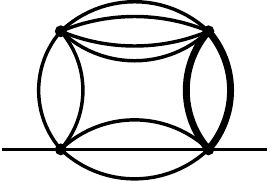}}
\\
\subfloat[]{\label{fig:ind_deletion_easy_dipole_d}\includegraphics[scale=1]{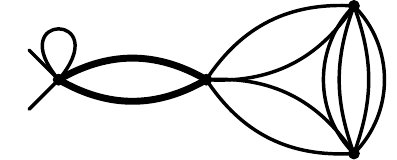}}
\hspace{1cm}
\subfloat[]{\label{fig:ind_deletion_easy_dipole_e}\includegraphics[scale=1]{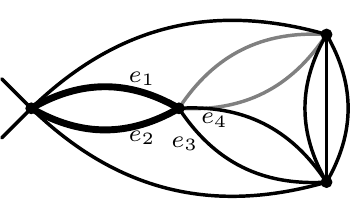}}
\\
\subfloat[]{\label{fig:ind_deletion_easy_dipole_f}\includegraphics[scale=1]{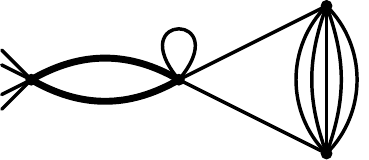}}
\hspace{1cm}
\subfloat[]{\label{fig:ind_deletion_easy_dipole_g}\includegraphics[scale=1]{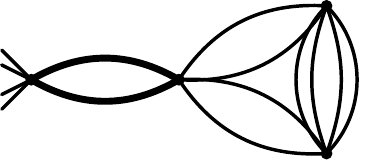}}
\caption{The interesting configurations when deleting a type-$I$ dipole. The type-$I$ dipole we want to delete is represented in bold.}
\label{fig:ind_deletion_easy_dipole}
\end{figure}

The first graph contains a type-$I$ tadpole (at its left end) so it has already been excluded.  We can deal with $H_7$ by means of Lemmas~\ref{lemma:particular_cases} and \ref{lemma:two-point_deletion}.

If this creates a melon, the graph is in one of the configurations of Fig.~\ref{fig:ind_deletion_easy_dipole_b} and \ref{fig:ind_deletion_easy_dipole_c}.

We can first use Lemma~\ref{lemma:melon4pt} to delete a quartic rung in both of these graphs. The first one reduces to $H_4$, the second one to $H_0$. In both situations, we can then conclude by invoking Lemmas~\ref{lemma:particular_cases} and \ref{lemma:two-point_deletion}.

\item[(ii)] We perform the move of Fig.~\ref{fig:ind_move_easy_dipole}, which can be justified from Lemma~\ref{lemma:dipole} as follows. We first try to delete the dipole in the (unique) channel $2$ that connects $a$ to $w$ and $b$ to $x$. If this move turns out to connect $c$ to $y$ (resp. $z$) and $d$ to $z$ (resp. $y$), we are done. If not, $c$ is mapped to $d$, which may result in a disconnected graph. In that situation, we instead perform a deletion in the (again unique) channel that connects $a$ to $x$ and $b$ to $y$. We are then guaranteed that $c$ is not mapped to $d$, and we have successfully implemented the combinatorial move of Fig.~\ref{fig:ind_move_easy_dipole}. 
\begin{figure}[H]
\centering
\includegraphics[scale=1]{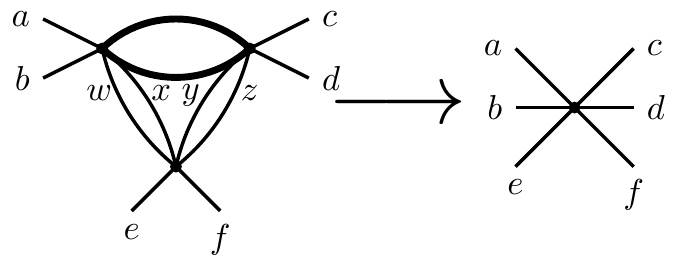}
\caption{Graphical representation of the move used to delete a type-$I$ dipole (in bold) in the configuration (ii).}
\label{fig:ind_move_easy_dipole}
\end{figure}

If this creates a double-tadpole, then $G$ must contain one of the subgraphs $H_8$ or $H_9$, as depicted in Fig.~\ref{fig:particular_cases}.

If this creates a melon, we are instead led to the two-point subgraph $H_{10}$.

Again, these special cases are dealt with Lemmas~\ref{lemma:particular_cases} and \ref{lemma:two-point_deletion}. 

\item[(iii)] The configuration with a generalized tadpole on the vertex that already has a tadpole is forbidden, otherwise $G'$ would be disconnected. Hence $G$ contains a type-$I$ tadpole, which we have already excluded. 

\item [(iv)] We delete the dipole in the unique channel $2$ that sends $e_1$ onto $e_3$ (and $e_2$ onto $e_4$). We either obtain a dipole-tadpole or a quartic rung. 

This can create a melon or a double-tadpole only if $G$ contains $H_{10}$. 

\item [(v)] We perform the deletion in the parallel channel. This cannot create a double-tadpole but can create a melon if the graph was in one of the two configurations depicted in Fig.~\ref{fig:ind_deletion_easy_dipole_d} and \ref{fig:ind_deletion_easy_dipole_e}.

The first one is incompatible with our assumption that the dipole can be deleted in one of the channels $2$ without disconnecting the graph. For the second one, we delete instead the grey dipole in an appropriate channel $2$, for instance the one that sends $e_1$ onto $e_3$ and $e_2$ onto $e_4$. This reduces to the particular case $H_4$, which is covered by Lemmas~\ref{lemma:particular_cases} and \ref{lemma:two-point_deletion}. 

\item[(vi)] We instead implement a deletion in an appropriate channel $2$. In more detail, we first try the deletion that sends $e_1$ to $e_3$ (and $e_2$ to $e_4$). The graph is guaranteed to remain connected unless this move maps $e_1'$ to $e_2'$. In that case, we can instead implement a deletion in the channel $2$ that maps $e_1$ to $e_4$, in which case $e_1'$ is mapped to $e_3'$ or $e_4'$. In both situations, the graph remains connected. However, a melon could be created, in which case the graph is the two-point subgraph $H_{11}$ depicted in Fig.~\ref{fig:particular_cases}.

We then conclude with Lemmas~\ref{lemma:particular_cases} and \ref{lemma:two-point_deletion}. 

\item[(vii)] We start by deleting the quartic rung subgraph on the right using Lemma~\ref{lemma:melon4pt}. This can create a double-tadpole if $G$ was in one of the configurations depicted in Fig.~\ref{fig:ind_deletion_easy_dipole_f} and \ref{fig:ind_deletion_easy_dipole_g}. The first one is excluded because it contains a melon. The second one is excluded because the graph disconnects in all three channels $2$ of the dipole we started from. 

This can also create a melon if the graph contains either of the subgraphs $H_{11}$ and $H_{12}$ of Fig.~\ref{fig:particular_cases}. We can deal with both of these with Lemmas~\ref{lemma:particular_cases} and \ref{lemma:two-point_deletion}. 
\end{itemize}

\end{itemize}

We can now assume that there are no type-$I$ dipoles left. 

\begin{itemize}
\item[\emph{Case E.}] We finally assume that $G$ contains a type-$II$ tadpole or dipole.

Consider a type-$II$ tadpole or dipole $S\subset G$, and assume that the root edge is not contained in $S$. By definition, we know that $S$ is included in a two-point subgraph with one of the structures depicted in Fig.~\ref{fig:not_easy_2}. Moreover, there is a unique such subgraph that \emph{does not contain the root edge}. We call it the \emph{canonical two-point subgraph} associated to $S$, and denote it by $C_S$.  
\begin{figure}[H]
\centering
\captionsetup[subfigure]{labelformat=empty}
\subfloat[]{\includegraphics[scale=1]{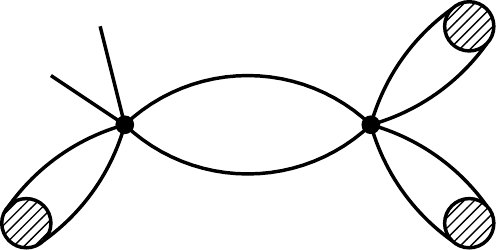}}
\hspace{1cm}
\subfloat[]{\includegraphics[scale=1]{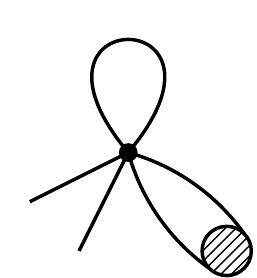}}
\hspace{1cm}
\subfloat[]{\includegraphics[scale=1]{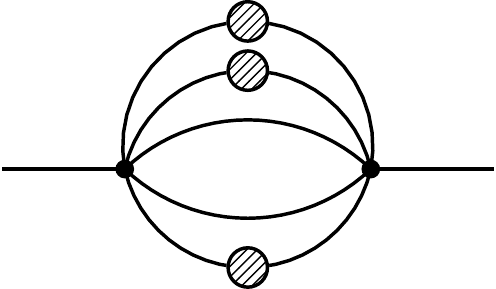}}
\caption{Canonical two-point subgraphs associated to type-$II$ dipoles and tadpoles. Given a type-$II$ tadpole or dipole $S$, $C_S$ is uniquely determined by the requirement that the root edge is not contained in it.}
\label{fig:not_easy_2}
\end{figure}
We then claim that the family of subgraphs $C_S$ forms an \emph{inclusion forest}. That is, given two type-$II$ tadpoles or dipoles $S_1$ and $S_2$, one of the following conditions holds: 
\begin{itemize}
\item $C_{S_1} \subset C_{S_2}$ or $C_{S_2} \subset C_{S_1}$;
\item $C_{S_1}$ and $C_{S_2}$ are (vertex and edge) disjoint. 
\end{itemize}
As a result, provided that this set of subgraphs is non-empty, there exists a dipole or tadpole $S_0$ such that $C_{S_0}$ is minimal for the inclusion. Moreover, $C_{S_0}$ necessarily contains a proper two-point function (i.e. one of the two-point subgraphs represented by blobs in Fig.~\ref{fig:not_easy_2}, otherwise $G$ would contain a double-tadpole or a melon), which we call $H$. Given that $H\subset C_{S_0}$ and $C_{S_0}$ is minimal for the inclusion among all $C_S$ subgraphs, it follows that $H$ cannot contain any type-$II$ tadpole or dipole. Since we have already assumed the absence of type-$I$ tadpoles or dipoles in $G$, $H$ cannot contain any short face. Consequently, we can use Lemmas~\ref{lemma:2PR} and \ref{lemma:two-point_deletion} to construct a suitable graph $G'$ such that $\omega(G)\geq \omega(G')$.  

We are left with one last case to examine: when there is no other type-$II$ tadpole or dipole than one containing the root edge. But in this case $G$ contains at most one short face, so we can immediately conclude that $\omega(G) >0$. 
\end{itemize}

This concludes the proof.
\end{proof}

\section{Melonic dominance}
\label{sec:LO}

Proposition~\ref{prop:positive_degree} and section~\ref{sec:subtraction} immediately imply the existence of the large $N$ expansion. We now set to prove that leading order graphs are melonic. We start with the following simple observation. 
\begin{lemma}\label{lemma:full-2pt}
Let $\mathcal{G}$ be a (non-amputated) two-point Feynman map. The associated amplitude $\mathcal{A}(\mathcal{G})_{\pmb a,\pmb b}$ can be written as:
\begin{equation*}
\mathcal{A}(\mathcal{G})_{\pmb a,\pmb b}=\lambda^{V(\mathcal{G})}f_{\mathcal{G}}(N)\pmb P_{\pmb a,\pmb b}
\end{equation*}
where $f_{\mathcal{G}}(N)$ is uniformly bounded.
\end{lemma}
\begin{proof}
The irreducibility of the tensor representation, together with Schur's lemma, immediately imply that the amplitude is proportional to the projector $\pmb P$. Furthermore, consistency with the existence of the large $N$ expansion requires that $f_{\mathcal{G}}(N)$ is uniformly bounded.
\end{proof}

The next two lemmas demonstrate that many of the stranded configurations which we could not exclude to be of vanishing degree in the previous section, in fact cannot contribute to the leading order. This results from the same type of cancellations we already relied on in section~\ref{sec:subtraction}. But now that the existence of the large $N$ expansion has been established, we can be more systematic.
\begin{lemma}
Let $\mathcal{G}$ be a (connected and vacuum) Feynman map. If $\mathcal{G}$ has a generalized double-tadpole then it is subleading, that is:
\begin{equation}
\vert A(\cG) \vert \leq K N^4 \,.
\end{equation}
for some constant $K>0$.
\label{lemma:gen_double_tadpole}
\end{lemma}
\begin{proof}
\begin{figure}[htbp]
\centering
\includegraphics[scale=1]{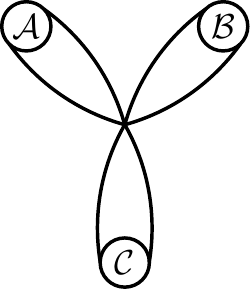}
\caption{Vacuum graph with a generalized double-tadpole}
\label{fig:gen_double_tadpole}
\end{figure}
Up to embedding information (which does not affect large $N$ scalings), $\mathcal{G}$ must have the configuration depicted in Fig.~\ref{fig:gen_double_tadpole}, where $\mathcal{A}$, $\mathcal{B}$ and $\mathcal{C}$ are two-point Feynman maps. By Lemma~\ref{lemma:full-2pt}, there exists three uniformly bounded functions $f_{\mathcal{A}}$, $f_{\mathcal{B}}$ and $f_{\mathcal{C}}$ such that:
\begin{equation}
\mathcal{A}(\mathcal{G})=\lambda^{V(\mathcal{A})+V(\mathcal{B})+V(\mathcal{C})}f_{\mathcal{A}}(N) f_{\mathcal{B}}(N) f_{\mathcal{C}}(N) \mathcal{A}(\mathcal{G}')\,,
\end{equation}
where $\cG'$ is the map obtained by replacing $\mathcal{A}$, $\mathcal{B}$ and $\mathcal{C}$ with bare propagators $\pmb P$.  $\cG'$ is nothing but a double-tadpole graph, so from the computations of section~\ref{sec:subtraction}, we also know that $\mathcal{A}(\cG')\sim f_1^{\pmb P}(N)\pmb P_{\pmb a,\pmb a} = \mathcal{O}(N^4)$.
\end{proof}

\begin{lemma}\label{lemma:CS}
Let $\cG$ be a (connected and vacuum) Feynman map. If $\cG$ contains a generalized tadpole or a type-$I$ dipole, then it is subleading.
\end{lemma}
\begin{proof}
Let us first assume that $\cG$ contains a generalized tadpole. From Lemmas~\ref{lemma:full-2pt} and \ref{lemma:gen_double_tadpole}, it is sufficient to deal with the situation of a single-tadpole, as depicted in Fig.~\ref{fig:tadpole_LO}, where $\mathcal{S}$ is a \emph{connected} four-point map. We can then apply the Cauchy-Schwarz inequality to find:
\begin{equation}
\mathcal{A}(\cG)^2 = \mathcal{A}\left(\vcenter{\hbox{\includegraphics[scale=.5]{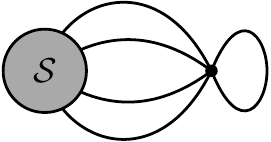}}} \right)^2 \leq \mathcal{A}\left(\vcenter{\hbox{\includegraphics[scale=.5]{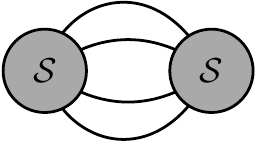}}}\right) \mathcal{A}\left(\vcenter{\hbox{\includegraphics[scale=.5]{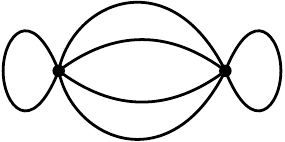}}} \right)\,.
\end{equation}
The first term on the right is the amplitude of a connected map, and is therefore in $\mathcal{O}(N^5)$. Furthermore, from Lemma~\ref{lemma:particular_cases} (subcase $H_0$), the degree of any stranded configuration contributing to the second term is at least $1$. Hence this term is in $\mathcal{O}(N^4)$. As a result, $\mathcal{A}(\cG)$ is at most in $\mathcal{O}(N^{9/2})$, which implies it is subleading. 
\begin{figure}[htbp]
\centering
\includegraphics[scale=1]{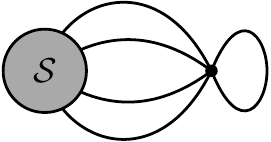}
\caption{$\cG$ contains a tadpole, and the submap $\mathcal{S}$ is assumed to be connected.}
\label{fig:tadpole_LO}
\end{figure}

Let us now assume that $\cG$ contains a type-$I$ dipole. Without loss of generality, and up to embedding, we can assume that we are in one of the situations represented in Fig.~\ref{fig:dipole_LO}, where the submaps $\mathcal{S}_i$ are all connected\footnote{In Fig.~\ref{fig:dipole_LOc}, we have used Lemma~\ref{lemma:full-2pt} to suppress one potentially non-trivial two-point subgraph.}. In the first case (Fig.~\ref{fig:dipole_LOa}), we can again invoke the Cauchy-Schwarz inequality, which implies:
\begin{equation}
\mathcal{A}(\cG)^2 = \mathcal{A}\left(\vcenter{\hbox{\includegraphics[scale=.5]{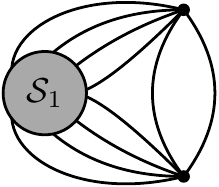}}} \right)^2 \leq \mathcal{A}\left(\vcenter{\hbox{\includegraphics[scale=.5]{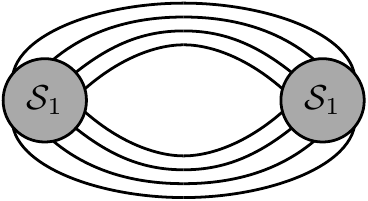}}}\right) \mathcal{A}\left(\vcenter{\hbox{\includegraphics[scale=.5]{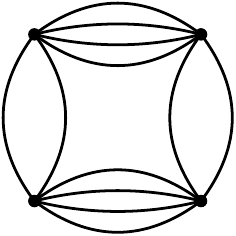}}} \right)\,.
\end{equation}
The first term on the right is the amplitude of a connected map, and the second term is subleading by Lemmas~\ref{lemma:melon4pt} and \ref{lemma:particular_cases} (subcase $H_0$). As before, we conclude that $\cG$ is subleading.

The second case (Fig.~\ref{fig:dipole_LOb}) can be dealt with by successive applications of the Cauchy-Schwarz inequality:
\begin{equation}
\mathcal{A}(\cG)^2 = \mathcal{A}\left(\vcenter{\hbox{\includegraphics[scale=.5]{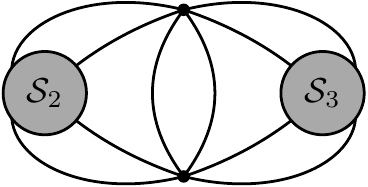}}} \right)^2 \leq \mathcal{A}\left(\vcenter{\hbox{\includegraphics[scale=.5]{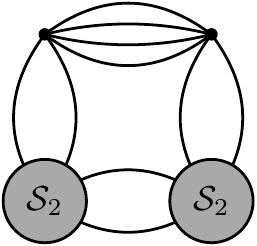}}}\right) \mathcal{A}\left(\vcenter{\hbox{\includegraphics[scale=.5]{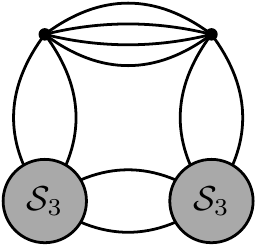}}} \right)\,,
\end{equation}
and
\begin{equation}
\mathcal{A}\left(\vcenter{\hbox{\includegraphics[scale=.5]{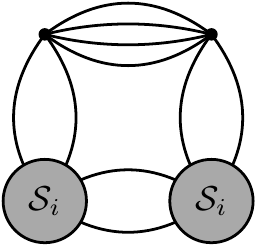}}} \right)^2 \leq \mathcal{A}\left(\vcenter{\hbox{\includegraphics[scale=.5]{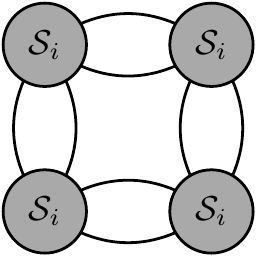}}}\right) \mathcal{A}\left(\vcenter{\hbox{\includegraphics[scale=.5]{dipole_LO_aaa.pdf}}} \right)\,.
\end{equation}
The fact that all the graphs in these relations are connected, while one of them is subleading, allows to conclude that $\cG$ is itself subleading. 

We proceed in a similar way for the last case (Fig.~\ref{fig:dipole_LOc}) and obtain:
\begin{equation}
\mathcal{A}(\cG)^2 = \mathcal{A}\left(\vcenter{\hbox{\includegraphics[scale=.5]{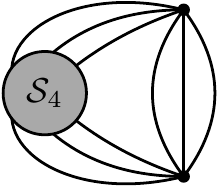}}} \right)^2 \leq \mathcal{A}\left(\vcenter{\hbox{\includegraphics[scale=.5]{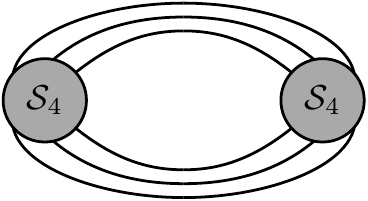}}}\right) \mathcal{A}\left(\vcenter{\hbox{\includegraphics[scale=.5]{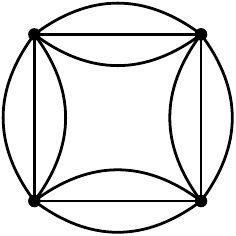}}} \right)\,.
\end{equation}
It is then sufficient to show that the second term on the right is subleading. We can in fact prove that any of the stranded configurations of this map has strictly positive degree.  Indeed, this map has no tadpole and twelve dipoles, thus $F_1=0$ and $F_2\leq 12$. Moreover, any other face has length at least four: $F_3=0$. Using the bounds of Appendix~\ref{ap:bounds} with $\mathcal{I}=15\times 4$ and $k=3$, we have $F \leq \lfloor \frac{60+24}{4}\rfloor=21$. Therefore $\omega \geq 5 + 5 \times 4 - 21 = 4$, which concludes the proof. 
\begin{figure}[htbp]
\centering
\subfloat[\label{fig:dipole_LOa}]{\includegraphics[scale=1]{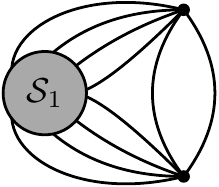}}
\hspace{1cm}
\subfloat[\label{fig:dipole_LOb}]{\includegraphics[scale=1]{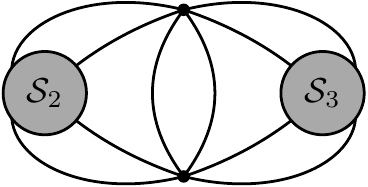}}
\hspace{1cm}
\subfloat[\label{fig:dipole_LOc}]{\includegraphics[scale=1]{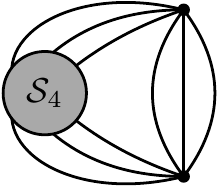}}
\caption{Three configurations of a type-$I$ dipole submap.}
\label{fig:dipole_LO}
\end{figure}
\end{proof}

\begin{proposition}\label{propo:LO}
Let $\cG$ be a (connected and vacuum) Feynman map. $\cG$ is leading order if and only if it is melonic. 
\end{proposition}
\begin{proof}
From section \ref{sec:subtraction}, we already know that melonic graphs are leading order. To prove the converse, let us consider a leading order Feynman map $\cG$. We can start by recursively removing all melon two-point functions from $\cG$, to obtain a leading-order map $\cG'$ with no melon. By definition, $\cG$ is melonic if and only if $\cG'$ is the ring map. Let us assume it is not. Then $\cG'$ must be able to support short faces, otherwise it could not be leading order. Given Lemma~\ref{lemma:CS}, the only possibility left is that $\cG'$ contains type-$II$ dipoles whose canonically associated two-point functions are generalized melons (as in the right panel of Fig.~\ref{fig:not_easy_2}). Considering a minimal such submap for the inclusion, which we call $\mathcal{S}$, leads to a contradiction. Indeed, at least one of the two-point functions in the generalized melon $\mathcal{S}$ must be non-empty, otherwise $\cG'$ would contain a melon submap. By minimality, this two-point function cannot contain any dipole or tadpole, therefore it is necessarily subleading. By Lemma~\ref{lemma:full-2pt}, $\cG'$ itself must be subleading, which yields the desired contradiction. Consequently, $\cG'$ is the ring map and $\cG$ is melonic, as claimed.    
\end{proof}

\section{Further comments and outlook}

We have established that irreducible tensor models with $5$-simplex interaction admit a melonic large $N$ expansion. Along the way, we had to estimate the large $N$ behaviour of a number of four- and eight-point functions. From this analysis, it is straightforward to include other effective interactions in our models. Any boundary graph we have explicitly investigated may lead to non-vanishing interaction terms if it does not contain self-loops. This includes, for instance, the boundary graphs represented in Figs.~\ref{fig:tadpole_config_d}, \ref{fig:tadpole_config_e}, or \ref{fig:dipole_config}. With a bit more effort, one could determine the optimal scaling of all effective $n$-point interactions which contribute at leading order, for (say) $n \leq 6$. This would be a prerequisite for potential applications of our results to large $N$ QFT, where any such interaction that is also relevant in the renormalization sense would have to be included in the bare action. Even though we chose to work in vanishing dimension for simplicity, our main theorems hold in higher dimensions with minimal changes, namely: the algebraic equation defining $F_{\pmb P}^{(0)}$ in Theorem~\ref{theorem_LO} should in general be replaced by a suitable (integro-differential) Schwinger-Dyson equation. 

In a similar spirit, it would be interesting to investigate whether our results can be generalized to fermionic tensor fields transforming under the compact symplectic group $\mathrm{Sp}(N)$, by analogy with the rank-$3$ construction of \cite{Carrozza:2018psc}. 

Beyond rank-$5$, a number of generalizations could be explored. First, it would be interesting to study irreducible rank-$4$ models with $4$-simplex interaction. The main difficulty we may expect in this case is that, just like in rank-$3$, problematic configurations requiring a detailed combinatorial analysis will also include triangle submaps.\footnote{This was actually one of the reasons why we decided to focus on rank-$5$ in the present work.} Finally, it does not seem completely unrealistic to imagine that the present proof could be generalized to arbitrary rank $r \geq 6$. This could be a worthwhile endeavour in view of potential applications of symmetric random tensors to statistics and applied mathematics\cite{Evnin:2020ddw, Gurau:2020ehg}. However, at a minimum, one would need to find a more systematic way of investigating and bounding particular two-point stranded subgraphs, such as those of Lemma~\ref{lemma:particular_cases}. Even if one could succeed in this, this would presumably lead to a very technical proof. From this point of view, it would be highly desirable to develop alternative methods that do not rely so heavily on inductive combinatorial moves.

\section*{Acknowledgements}

We would like to thank Razvan Gurau for his collaboration in the early stages of this project and for useful comments at a later stage.

SC is supported by a Radboud Excellence Fellowship from Radboud University in Nijmegen, the Netherlands. The work of SH is supported by the European Research Council (ERC) under the European Union's Horizon 2020 research and innovation program (grant agreement No818066) and by Deutsche Forschungsgemeinschaft (DFG, German Research Foundation) under Germany's Excellence Strategy EXC-2181/1 - 390900948 (the Heidelberg STRUCTURES Cluster of Excellence)

In the initial stages of this work, SC was supported by Perimeter Institute. Research at Perimeter Institute is supported, in part, by the Government of Canada through the Department of Innovation, Science and Economic Development Canada, and by the Province of Ontario through the Ministry of Colleges and Universities.

\appendix

\section{Bounds on the number of faces}
\label{ap:bounds}

We wish to bound the number of faces of a generic two-point graph. In order to do so, we will follow a method developed in Appendix C of \cite{Benedetti:2017qxl}. We label $x,y$ the external legs of the graph. We call $V$, $E$, $F$ and $F_i$, the number of vertices, edges, faces and faces of length $i$ in the subgraph. We also define $l$ as the sum of the length of the open strands of the subgraph. Then, we can write the sum of the length of the internal faces as:
\begin{equation}
\mathcal{I}=5E-l=15V-5-l
\end{equation}
because for a $6$-valent $2$-point graph, $6V=2E+2$.

Moreover, as stated in \cite{Benedetti:2017qxl}, for $F=\sum_i F_i$ and $\mathcal{I}=\sum_i iF_i$, we have the following bounds:
\begin{equation}
\forall k\geq 2~, \qquad F\leq \lfloor \frac{\mathcal{I}}{k+1} +\sum_{1\leq i\leq k} \frac{k+1-i}{k+1}F_i \rfloor 
\label{eq:comb_bounds}
\end{equation}
For $k=2$, we obtain $F\leq \lfloor(\mathcal{I}+2F_1+F_2)/3 \rfloor$. 

Likewise, we can write:
\begin{equation}
l  = \sum_{i \geq 0} i l_i \,, \quad 5p = \sum_i l_i \,,
\end{equation}
where $l_i$ is the number of external strands of length $i$, and $2p$ is the number of external legs (here, $p=2$). We then have the bounds:
\begin{equation}
\forall k\geq 1, \qquad 5p \leq \lfloor \frac{l}{k+1} +\sum_{0 \leq i\leq k} \frac{k+1-i}{k+1} l_i \rfloor 
\label{eq:comb_bounds_l}
\end{equation}
For $k=1$ and $p=2$, this yields in particular: $l \geq 20 - 2 l_0 - l_1$.

\section{Proof of lemma \ref{lemma:particular_cases}}
\label{app:particular}

\begin{proof}

We want to gain a factor $N$ by deleting the two-point graphs of Fig.~\ref{fig:particular_cases}. Therefore, the lemma will follow if we can prove $d(S_{\partial},B_u)\leq 5V-1-F(S)$ for each subgraph $S$ of Fig.~\ref{fig:particular_cases} and with $B_u$ the boundary graph of the unbroken propagator. 
In order to bound the number of internal faces, we are going to follow the method of \cite{Benedetti:2017qxl} and use the bounds of appendix \ref{ap:bounds}. 

\textbf{Graph $H_0$.} Here $V=2$ and $\mathcal{I}=25-l$. We must have $d(S_{\partial},B_u)\leq 9-F(S)$. We have two tadpoles and three dipoles thus $F_1 \leq 2$ and $F_2\leq 3$. 

\begin{figure}[H]
\centering
\includegraphics[scale=1]{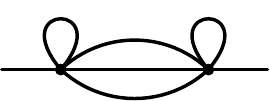}
\caption{The special case $H_0$}
\end{figure}

\begin{itemize}
\item Unbroken case: All external strands traverse and $d(S_{\partial},B_u)=0$. We can thus delete at most $9$ internal faces. We have at most three external strands of length one and two of length two. We thus have $l\geq 7$ and $\mathcal{I}\leq 18$. Then we have $F \leq \lfloor \frac{18+4+3}{3}\rfloor \leq 8$.
\item Broken case: Two external strands loop back. In this case, $d(S_{\partial},B_u)=1$ so we can delete at most $8$ internal faces. We can have all five external strands of length one: $l\geq 5$ and $\mathcal{I}\leq 10$. Here we have to consider faces of length three. We will need equation \eqref{eq:comb_bounds} for $k=3$:
\begin{equation}
F\leq \lfloor\frac{\mathcal{I}+3F_1+2F_2+F_3}{4}\rfloor
\end{equation}
We have $F_3\leq 6$. However, if $F_2=3$, $F_3=0$, if $F_2=2$, $F_3\leq 2$, if $F_2=1$, $F_3\leq 4$ and if $F_2=0$, $F_3\leq 6$. Therefore, $2F_2+F_3\leq 6$ and $F \leq \lfloor \frac{20+6+6}{4}\rfloor \leq 8$.
\item Doubly-broken case: Four external strands loop back. In this case, $d(S_{\partial},B_u)=2$ so we can delete at most $7$ internal faces. We can have three external strands of length one and two of length two: $l\geq 7$ and $\mathcal{I}\leq 18$. 
We thus have $F \leq \lfloor \frac{18+6+6}{4}\rfloor \leq 7$.
\end{itemize}

\textbf{Graph $H_1$.} We now look at the two-point subgraph $H_1$ represented in figure \ref{fig:h1_label}. It has $3$ vertices: $\mathcal{I}=40-l$. We now want $d(S_{\partial},B_u)\leq 14-F(S)$.  

\begin{figure}[htbp]
\centering
\includegraphics[scale=1]{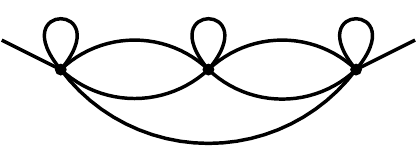}
\caption{The special case $H_1$.}
\label{fig:h1_label}
\end{figure}

\begin{itemize}
\item Unbroken case: In this case we can delete at most $14$ internal faces. The five external strands have at least length one. There can be at most one of length one and two of length two. We thus have $l\geq 11$ and $\mathcal{I} \leq 29$. There are three tadpoles and two dipoles, thus $F_1 \leq 3$ and $F_2 \leq 2$. We then have $F \leq \lfloor \frac{29+6+2}{3}\rfloor \leq 12$. 
\item Broken case:  we can delete at most $13$ internal faces. Now, there can be at most three external strands of length one and two of length two. Thus $l\geq 7$ and $\mathcal{I}\leq 33$. We still have $F_1\leq 3$ and $F_2 \leq 2$. We then have $F \leq \lfloor \frac{33+6+2}{3}\rfloor \leq 13$.
\item Doubly-broken case: we can delete at most $12$ internal faces and we still have $l\geq 7$. We need again to consider the faces of length three.
We can have at most eight faces of length three (four using the tadpoles and four using the strands between the two external points). However, if $F_2=2$ then $F_3\leq 4$, if $F_2=1$ then $F_3\leq 6$ and if $F_2=0$, $F_3\leq 8$. Thus we have $2F_2+F_3\leq 8$ and $F \leq \lfloor \frac{33+9+8}{4}\rfloor \leq 12$.
\end{itemize}

\textbf{Graph $H_2$.} Again $V=3$ and $\mathcal{I}=40-l$. There are two tadpoles and three dipoles so $F_1 \leq 2$ and $F_2 \leq 3$. We again want to prove $d(S_{\partial},B_u)\leq 14-F(S)$
\begin{figure}[htbp]
\centering
\includegraphics[scale=1]{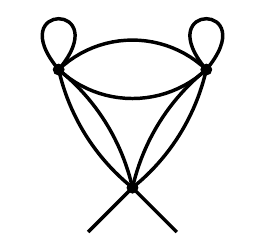}
\caption{The special case $H_2$}
\label{fig:h2_label}
\end{figure}

For this subgraph, in all cases (unbroken, broken and doubly broken), there is always an external strand of length zero between the two external points as there are connected to the same vertex. There are at most two external strands of length $2$. Thus $l \geq 10$ and $\mathcal{I} \leq 30$. We then have $F\lfloor \frac{30+4+3}{3}\rfloor \leq 12$. This gives us the right bounds for all three cases. 
\\

\textbf{Graph $H_3$.} Here $V=4$ and $\mathcal{I}=55-l$. We now need to prove the following bound: $d(S_{\partial},B_u)\leq 19-F(S)$. Moreover, we have two tadpoles and four dipoles so $F_1 \leq 2$ and $F_2 \leq 4$.

\begin{figure}[htbp]
\centering
\includegraphics[scale=1]{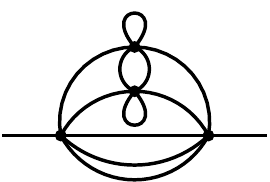}
\caption{The special case $H_3$}
\label{fig:h4_label}
\end{figure}

\begin{itemize}
\item Unbroken case: We have at most three external strands of length one and two of length two. Thus $l \geq 7$ and $\mathcal{I} \leq 48$. Therefore, we have $F \leq \lfloor \frac{48+4+4}{3}\rfloor \leq 18$.
\item Broken case: We have at most three external strands of length one. Then the two remaining external strands must loop back so they have at least length three. In this case $l\geq 9$ and $\mathcal{I}\leq 46$. Therefore, $F \leq \lfloor \frac{46+4+4}{3} \rfloor \leq 18$. 
\item Doubly-broken case: We have at most one external strand of length one as four strands must loop back. We have at most two strands of length two (we only have the two internal corners of the dipole that was not used for the external strand of length one) and two strands of length three.  So, $l \geq 11$ and $\mathcal{I} \leq 44$. Thus we have $F \leq \lfloor \frac{44+4+4}{3} \rfloor \leq 17$.
\end{itemize}

\textbf{Graph $H_4$.} Here, $V=2$ and $\mathcal{I}=25-l$. We need to prove that $d(S_{\partial},B_u)\leq 9-F(S)$. There are one tadpole and six dipoles so $F_1\leq 1$ and $F_2 \leq 6$.

\begin{figure}[htbp]
\centering
\includegraphics[scale=1]{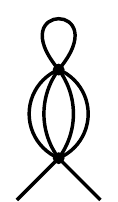}
\caption{The special case $H_4$}
\label{fig:h5_label}
\end{figure}

For all cases (unbroken, broken or doubly-broken), we always have on external strand of length zero as they are connected to the same vertex. The other have at least length two. So $l \geq 8$ and $\mathcal{I}\leq 17$.

\begin{itemize}
\item If $l=8$: we have four external strands of length two. Thus four internal corners of the dipoles are used for the external strands: there are at most two faces of length two. Thus, $F \leq \lfloor \frac{17+2+2}{3} \rfloor \leq 7$.
\item If $l=9$: there is one external strand of length three. The corners of only three dipoles are  now taken by the external strands: $F_2\leq 3$. Thus we have $F \leq \lfloor \frac{16+3+2}{3} \rfloor \leq 7$.
\item If $l\geq 10$, $F_2 \leq l-4$ thus we have $F \leq \lfloor \frac{25-l+2+l-4}{3} =\frac{23}{3}\rfloor \leq 7$.
\end{itemize}

\textbf{Graph $H_5$.} Here $V=3$ and $\mathcal{I}=40-l$. Again, we want to prove the following bound $d(S_{\partial},B_u)\leq 14-F(S)$. There are one tadpole and five dipoles so $F_1 \leq 1$ and $F_2\leq 5$. 

\begin{figure}[htbp]
\centering
\includegraphics[scale=1]{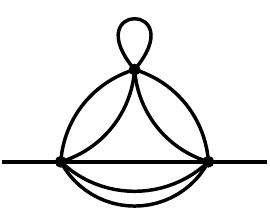}
\caption{The special case $H_5$}
\label{fig:h6_label}
\end{figure}

\begin{itemize}
\item Unbroken case: We can have at most three external strands of length one so $l \geq 7$ and $\mathcal{I}\leq 33$. Thus, $F \leq \lfloor \frac{33+2+5}{3} \rfloor \leq 13$.
\item Broken case: We still have $l\geq 7$ and thus $F \leq 13$. 
\item Doubly-broken case: As only one strand traverses, we can have at most one external strand of length one. Then, $l \geq 9$ and $\mathcal{I}\leq 31$. Thus, $F \leq \lfloor \frac{31+2+5}{3} \rfloor \leq 12$. 
\end{itemize}

\textbf{Graph $H_6$.} Here $V=3$ and $\mathcal{I}=40-l$. We still need to prove that $d(S_{\partial},B_u)\leq 14-F(S)$. There are one tadpole and seven dipoles so $F_1 \leq 1$ and $F_2 \leq 7$. 

\begin{figure}[htbp]
\centering
\includegraphics[scale=1]{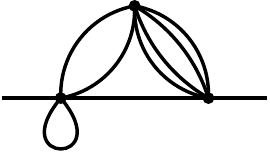}
\caption{The special case $H_6$}
\label{fig:h7_label}
\end{figure}

\begin{itemize}
\item Unbroken case: We have at most one external strand of length one and two of length two (each using one of the edges of the dipole on the left of the graph). Then,  $l\geq 11$ and $\mathcal{I} \leq 29$. Thus, we have $F \leq \lfloor \frac{29+2+7}{3} \rfloor \leq 12$.
\item Broken case: As two strands loop back, we can have one more external strand of length one and three external strand of length two. Then, $l \geq 8$ and $\mathcal{I}\leq 32$. Thus, we have $F \leq \lfloor \frac{32+2+7}{3} \rfloor \leq 13$. 
\item Doubly-broken case: we still have $l\geq 8$. Let us consider the faces of length three. There are at most $3$ faces of length three. However, if $F_2=7$ then $F_3 \leq 2$ and if $F_2\leq 6$, $F_3\leq 3$. Thus $2F_2+F_3\leq 16$ and $F \leq \lfloor \frac{32+3+16}{4} \rfloor \leq 12$.
\end{itemize}

\textbf{Graph $H_7$.} Here $V=3$ and $\mathcal{I}=40-l$. Again, we want to have $d(S_{\partial},B_u)\leq 14-F(S)$. There are no tadpoles and seven dipoles so $F_1=0$ and $F_2 \leq 7$.
For this graph, in all cases (unbroken, broken and doubly-broken), there is always an external strand of length zero and there are at most two external strands of length two. So we then have $l \geq 10$ and $\mathcal{I}\leq 30$. Thus, $F \leq \lfloor \frac{30+0+7}{3} \rfloor \leq 12$ which is okay for unbroken, broken and doubly-broken propagators. 
\\

\begin{figure}[htbp]
\centering
\includegraphics[scale=1]{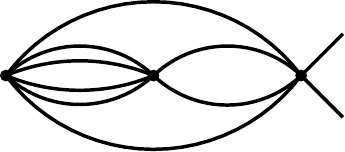}
\caption{The special case $H_7$}
\label{fig:h9_label}
\end{figure}

\textbf{Graph $H_{8}$.} Here $V=3$ and $\mathcal{I}=40-l$. We still need to have $d(S_{\partial},B_u)\leq 14-F(S)$. There are no tadpoles and seven dipoles so $F_1=0$ and $F_2\leq 7$.

\begin{figure}[htbp]
\centering
\includegraphics[scale=1]{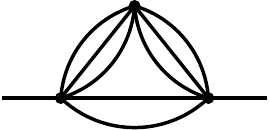}
\caption{The special case $H_{8}$}
\label{fig:h12_label}
\end{figure}

\begin{itemize}
\item Unbroken case: There can be two external strands of length one and three of length two: $l\geq 8$ and $\mathcal{I}\leq 32$. Thus, we have  $F \leq \lfloor \frac{32+0+7}{3} \rfloor \leq 13$.
\item Broken case: Two external strands must loop back but we still have $l\geq 8$ so $F\leq 13$.
\item Doubly broken case: Four strands must loop back: we can have only one external strand of length one and two of length two. Indeed, calling $x$ the external strand on the left and $a,b,c$ the three edges on the left, if we have an external strand $xabx$ we cannot have a second one using $c$ as both corners $xa$ and $xb$ are already used. Thus, $l\geq 11$ and  $F \leq \lfloor \frac{29+0+7}{3} \rfloor \leq 12$.
\end{itemize}

\textbf{Graph $H_{9}$.} Here $V=3$ and $\mathcal{I}=40-l$ and we again have to prove that $d(S_{\partial},B_u)\leq 14-F(S)$. There are no tadpoles and eight dipoles so $F_2 \leq 8$ and $F_1=0$. Here for all cases (unbroken, broken and doubly broken), we always have one external strand of length zero. We can also have at most two external strands of length two and two of length three. Then, $l\geq 10$ and $\mathcal{I}\leq 30$. This gives  $F \leq \lfloor \frac{30+0+8}{3} \rfloor \leq 12$.

\begin{figure}[htbp]
\centering
\includegraphics[scale=1]{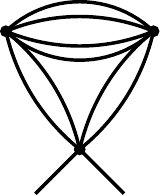}
\caption{The special case $H_{9}$}
\label{fig:h13_label}
\end{figure}

\textbf{Graph $H_{10}$.} Here $V=4$ and $\mathcal{I}=55-l$. We need to prove the following bound $d(S_{\partial},B_u)\leq 19-F(S)$. There are no tadpoles and five dipoles so $F_1=0$ and $F_2\leq 5$. In all cases, there can be at most one external strand of length one and four of length two. Then, $l\geq 9$ and $\mathcal{I}\leq 46$. This gives  $F \leq \lfloor \frac{46+0+5}{3} \rfloor \leq 17$.

\begin{figure}[htbp]
\centering
\includegraphics[scale=1]{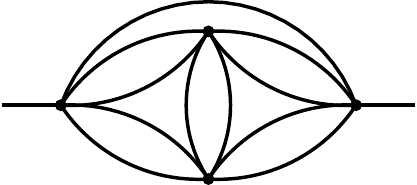}
\caption{The special case $H_{10}$}
\label{fig:h14_label}
\end{figure}

\textbf{Graph $H_{11}.$} Here $V=4$ and $\mathcal{I}=55-l$. In this case we also need to prove that $d(S_{\partial},B_u)\leq 19-F(S)$. There are no tadpoles and nine dipoles so $F_1 =0$ and $F_2 \leq 9$.

\begin{figure}[H]
\centering
\includegraphics[scale=1]{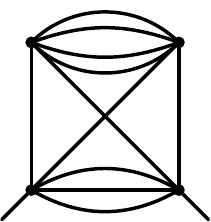}
\caption{The special case $H_{11}$}
\label{fig:hn2_label}
\end{figure}

\begin{itemize}
\item Unbroken case: We can have at most three strands of length one and two of length two. Thus, $l\geq 7$ and $\mathcal{I}\leq 48$. We then have $F \leq \lfloor \frac{48+9}{3} \rfloor \leq 19$.
\item Broken case: We can now have three external strands of length one and two of length three or one of length one and four of length two. Thus we have $l\geq 9$ and $F \leq \lfloor \frac{46+9}{3} \rfloor \leq 18$.
\item Doubly-broken case: We can now have one external strand of length one, two of length two and two of length three. Thus $l\geq 11$ and we have $F \leq \lfloor \frac{44+9}{3} \rfloor \leq 17$.
\end{itemize}

\textbf{Graph $H_{12}$.} Here $V=4$ and $\mathcal{I}=55-l$. In this case we again need to prove that $d(S_{\partial},B_u)\leq 19-F(S)$. There are no tadpoles and twelve dipoles so $F_1 =0$ and $F_2 \leq 12$.

\begin{figure}[htbp]
\centering
\includegraphics[scale=1]{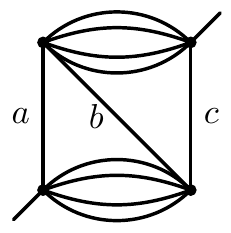}
\caption{The special case $H_{12}$}
\label{fig:hn3_label}
\end{figure}

For all cases (unbroken, broken and doubly-broken), we can have five external strands of length two. Thus, $l\geq 10$. Let us consider faces of length $3$: there can be at most two faces of length three (one using the corner $ab$ and one edge of the bottom quartic rung and one using the corner $bc$ and one edge of the top quartic rung). We thus have: $F \leq \lfloor \frac{45+24+2}{4} \rfloor \leq 17$.

\end{proof}

\bibliographystyle{JHEP}
\bibliography{Refs-TMV} 

\providecommand{\href}[2]{#2}\begingroup\raggedright\begin{thebibliography}{10}

\bibitem{Benedetti:2017qxl}
D.~Benedetti, S.~Carrozza, R.~Gurau and M.~Kolanowski, \emph{{The $1/N$
  expansion of the symmetric traceless and the antisymmetric tensor models in
  rank three}}, \href{https://doi.org/10.1007/s00220-019-03551-z}{\emph{Commun.
  Math. Phys.} {\bfseries 371} (2019) 55}
  [\href{https://arxiv.org/abs/1712.00249}{{\ttfamily 1712.00249}}].

\bibitem{Carrozza:2018ewt}
S.~Carrozza, \emph{{Large $N$ limit of irreducible tensor models: $O(N)$
  rank-$3$ tensors with mixed permutation symmetry}},
  \href{https://doi.org/10.1007/JHEP06(2018)039}{\emph{JHEP} {\bfseries 06}
  (2018) 039} [\href{https://arxiv.org/abs/1803.02496}{{\ttfamily
  1803.02496}}].

\bibitem{Bonzom:2011zz}
V.~Bonzom, R.~Gurau, A.~Riello and V.~Rivasseau, \emph{{Critical behavior of
  colored tensor models in the large N limit}},
  \href{https://doi.org/10.1016/j.nuclphysb.2011.07.022}{\emph{Nucl.\ Phys.\ B}
  {\bfseries 853} (2011) 174}
  [\href{https://arxiv.org/abs/1105.3122}{{\ttfamily 1105.3122}}].

\bibitem{RTM}
R.~Gurau, \emph{{Random Tensors}}. Oxford University Press, Oxford, 2016.

\bibitem{Klebanov:2018fzb}
I.~R. Klebanov, F.~Popov and G.~Tarnopolsky, \emph{{TASI} lectures on large
  {$N$} tensor models}, \href{https://doi.org/10.22323/1.305.0004}{\emph{PoS}
  {\bfseries TASI2017} (2018) 004}
  [\href{https://arxiv.org/abs/1808.09434}{{\ttfamily 1808.09434}}].

\bibitem{Guida:1998bx}
R.~Guida and J.~Zinn-Justin, \emph{{Critical exponents of the N vector model}},
  \href{https://doi.org/10.1088/0305-4470/31/40/006}{\emph{J. Phys. A}
  {\bfseries 31} (1998) 8103}
  [\href{https://arxiv.org/abs/cond-mat/9803240}{{\ttfamily
  cond-mat/9803240}}].

\bibitem{Moshe:2003xn}
M.~Moshe and J.~Zinn-Justin, \emph{Quantum field theory in the large {N} limit:
  A review}, \href{https://doi.org/10.1016/S0370-1573(03)00263-1}{\emph{Phys.
  Rept.} {\bfseries 385} (2003) 69}
  [\href{https://arxiv.org/abs/hep-th/0306133}{{\ttfamily hep-th/0306133}}].

\bibitem{'tHooft:1973jz}
G.~'t~Hooft, \emph{{A planar diagram theory for strong interactions}},
  \href{https://doi.org/10.1016/0550-3213(74)90154-0}{\emph{Nucl. Phys.}
  {\bfseries B72} (1974) 461}.

\bibitem{Brezin:1977sv}
E.~Brezin, C.~Itzykson, G.~Parisi and J.~B. Zuber, \emph{{Planar diagrams}},
  \href{https://doi.org/10.1007/BF01614153}{\emph{Commun. Math. Phys.}
  {\bfseries 59} (1978) 35}.

\bibitem{DiFrancesco:1993nw}
P.~Di~Francesco, P.~H. Ginsparg and J.~Zinn-Justin, \emph{{{$2-D$} {Gravity}
  and random matrices}},
  \href{https://doi.org/10.1016/0370-1573(94)00084-G}{\emph{Phys. Rept.}
  {\bfseries 254} (1995) 1}
  [\href{https://arxiv.org/abs/hep-th/9306153}{{\ttfamily hep-th/9306153}}].

\bibitem{Witten:2016iux}
E.~Witten, \emph{An {SYK}-like model without disorder},
  \href{https://doi.org/10.1088/1751-8121/ab3752}{\emph{J. Phys.} {\bfseries
  A52} (2019) 474002} [\href{https://arxiv.org/abs/1610.09758}{{\ttfamily
  1610.09758}}].

\bibitem{Gurau:2016lzk}
R.~Gurau, \emph{{The complete {$1/N$} expansion of a SYK--like tensor model}},
  \href{https://doi.org/10.1016/j.nuclphysb.2017.01.015}{\emph{Nucl. Phys.}
  {\bfseries B916} (2017) 386}
  [\href{https://arxiv.org/abs/1611.04032}{{\ttfamily 1611.04032}}].

\bibitem{Klebanov:2016xxf}
I.~R. Klebanov and G.~Tarnopolsky, \emph{Uncolored random tensors, melon
  diagrams, and the {SYK} models},
  \href{https://doi.org/10.1103/PhysRevD.95.046004}{\emph{Phys. Rev.}
  {\bfseries D95} (2017) 046004}
  [\href{https://arxiv.org/abs/1611.08915}{{\ttfamily 1611.08915}}].

\bibitem{Peng:2016mxj}
C.~Peng, M.~Spradlin and A.~Volovich, \emph{{A Supersymmetric SYK-like Tensor
  Model}}, \href{https://doi.org/10.1007/JHEP05(2017)062}{\emph{JHEP}
  {\bfseries 05} (2017) 062}
  [\href{https://arxiv.org/abs/1612.03851}{{\ttfamily 1612.03851}}].

\bibitem{Krishnan:2016bvg}
C.~Krishnan, S.~Sanyal and P.~N. Bala~Subramanian, \emph{Quantum chaos and
  holographic tensor models},
  \href{https://doi.org/10.1007/JHEP03(2017)056}{\emph{JHEP} {\bfseries 03}
  (2017) 056} [\href{https://arxiv.org/abs/1612.06330}{{\ttfamily
  1612.06330}}].

\bibitem{Krishnan:2017lra}
C.~Krishnan, K.~V. Pavan~Kumar and D.~Rosa, \emph{{Contrasting SYK-like
  Models}}, \href{https://doi.org/10.1007/JHEP01(2018)064}{\emph{JHEP}
  {\bfseries 01} (2018) 064}
  [\href{https://arxiv.org/abs/1709.06498}{{\ttfamily 1709.06498}}].

\bibitem{Bulycheva:2017ilt}
K.~Bulycheva, I.~R. Klebanov, A.~Milekhin and G.~Tarnopolsky, \emph{Spectra of
  operators in large {$N$} tensor models},
  \href{https://doi.org/10.1103/PhysRevD.97.026016}{\emph{Phys.\ Rev.\ D}
  {\bfseries 97} (2018) 026016}
  [\href{https://arxiv.org/abs/1707.09347}{{\ttfamily 1707.09347}}].

\bibitem{Choudhury:2017tax}
S.~Choudhury, A.~Dey, I.~Halder, L.~Janagal, S.~Minwalla and R.~Poojary,
  \emph{{Notes on melonic $O(N)^{q-1}$ tensor models}},
  \href{https://doi.org/10.1007/JHEP06(2018)094}{\emph{JHEP} {\bfseries 06}
  (2018) 094} [\href{https://arxiv.org/abs/1707.09352}{{\ttfamily
  1707.09352}}].

\bibitem{Halmagyi:2017leq}
N.~Halmagyi and S.~Mondal, \emph{{Tensor Models for Black Hole Probes}},
  \href{https://doi.org/10.1007/JHEP07(2018)095}{\emph{JHEP} {\bfseries 07}
  (2018) 095} [\href{https://arxiv.org/abs/1711.04385}{{\ttfamily
  1711.04385}}].

\bibitem{Klebanov:2018nfp}
I.~R. Klebanov, A.~Milekhin, F.~Popov and G.~Tarnopolsky, \emph{{Spectra of
  eigenstates in fermionic tensor quantum mechanics}},
  \href{https://doi.org/10.1103/PhysRevD.97.106023}{\emph{Phys. Rev.}
  {\bfseries D97} (2018) 106023}
  [\href{https://arxiv.org/abs/1802.10263}{{\ttfamily 1802.10263}}].

\bibitem{Carrozza:2018psc}
S.~Carrozza and V.~Pozsgay, \emph{{SYK-like tensor quantum mechanics with
  $\mathrm{Sp}(N)$ symmetry}},
  \href{https://doi.org/10.1016/j.nuclphysb.2019.02.012}{\emph{Nucl. Phys.}
  {\bfseries B941} (2019) 28}
  [\href{https://arxiv.org/abs/1809.07753}{{\ttfamily 1809.07753}}].

\bibitem{Delporte:2018iyf}
N.~Delporte and V.~Rivasseau, \emph{The tensor track {V}: Holographic tensors},
   in \emph{{Proceedings, 17th Hellenic School and Workshops on Elementary
  Particle Physics and Gravity (CORFU2017)}: {Corfu, Greece, September 2-28,
  2017}}, 4, 2018, \href{https://arxiv.org/abs/1804.11101}{{\ttfamily
  1804.11101}}.

\bibitem{Sachdev:1992fk}
S.~Sachdev and J.~Ye, \emph{{Gapless spin fluid ground state in a random,
  quantum Heisenberg magnet}},
  \href{https://doi.org/10.1103/PhysRevLett.70.3339}{\emph{Phys. Rev. Lett.}
  {\bfseries 70} (1993) 3339}
  [\href{https://arxiv.org/abs/cond-mat/9212030}{{\ttfamily
  cond-mat/9212030}}].

\bibitem{Kitaev2015}
A.~Kitaev, \emph{A simple model of quantum holography},  Talks at KITP, April
  7, 2015 and May 27, 2015.

\bibitem{Maldacena:2016hyu}
J.~Maldacena and D.~Stanford, \emph{{Remarks on the Sachdev-Ye-Kitaev model}},
  \href{https://doi.org/10.1103/PhysRevD.94.106002}{\emph{Phys. Rev.}
  {\bfseries D94} (2016) 106002}
  [\href{https://arxiv.org/abs/1604.07818}{{\ttfamily 1604.07818}}].

\bibitem{Polchinski:2016xgd}
J.~Polchinski and V.~Rosenhaus, \emph{{The Spectrum in the Sachdev-Ye-Kitaev
  Model}}, \href{https://doi.org/10.1007/JHEP04(2016)001}{\emph{JHEP}
  {\bfseries 04} (2016) 001}
  [\href{https://arxiv.org/abs/1601.06768}{{\ttfamily 1601.06768}}].

\bibitem{Gross:2016kjj}
D.~J. Gross and V.~Rosenhaus, \emph{{A Generalization of Sachdev-Ye-Kitaev}},
  \href{https://doi.org/10.1007/JHEP02(2017)093}{\emph{JHEP} {\bfseries 02}
  (2017) 093} [\href{https://arxiv.org/abs/1610.01569}{{\ttfamily
  1610.01569}}].

\bibitem{Giombi:2017dtl}
S.~Giombi, I.~R. Klebanov and G.~Tarnopolsky, \emph{{Bosonic tensor models at
  large {$N$} and small {$\epsilon$}}},
  \href{https://doi.org/10.1103/PhysRevD.96.106014}{\emph{Phys. Rev.}
  {\bfseries D96} (2017) 106014}
  [\href{https://arxiv.org/abs/1707.03866}{{\ttfamily 1707.03866}}].

\bibitem{Prakash:2017hwq}
S.~Prakash and R.~Sinha, \emph{A complex fermionic tensor model in $d$
  dimensions}, \href{https://doi.org/10.1007/JHEP02(2018)086}{\emph{JHEP}
  {\bfseries 02} (2018) 086}
  [\href{https://arxiv.org/abs/1710.09357}{{\ttfamily 1710.09357}}].

\bibitem{Benedetti:2017fmp}
D.~Benedetti, S.~Carrozza, R.~Gurau and A.~Sfondrini, \emph{{Tensorial
  Gross-Neveu models}},
  \href{https://doi.org/10.1007/JHEP01(2018)003}{\emph{JHEP} {\bfseries 01}
  (2018) 003} [\href{https://arxiv.org/abs/1710.10253}{{\ttfamily
  1710.10253}}].

\bibitem{Giombi:2018qgp}
S.~Giombi, I.~R. Klebanov, F.~Popov, S.~Prakash and G.~Tarnopolsky,
  \emph{Prismatic large {$N$} models for bosonic tensors},
  \href{https://doi.org/10.1103/PhysRevD.98.105005}{\emph{Phys. Rev.}
  {\bfseries D98} (2018) 105005}
  [\href{https://arxiv.org/abs/1808.04344}{{\ttfamily 1808.04344}}].

\bibitem{Benedetti:2018ghn}
D.~Benedetti and N.~Delporte, \emph{{Phase diagram and fixed points of
  tensorial Gross-Neveu models in three dimensions}},
  \href{https://doi.org/10.1007/JHEP01(2019)218}{\emph{JHEP} {\bfseries 01}
  (2019) 218} [\href{https://arxiv.org/abs/1810.04583}{{\ttfamily
  1810.04583}}].

\bibitem{Benedetti:2019eyl}
D.~Benedetti, R.~Gurau and S.~Harribey, \emph{{Line of fixed points in a
  bosonic tensor model}},
  \href{https://doi.org/10.1007/JHEP06(2019)053}{\emph{JHEP} {\bfseries 06}
  (2019) 053} [\href{https://arxiv.org/abs/1903.03578}{{\ttfamily
  1903.03578}}].

\bibitem{Benedetti:2019ikb}
D.~Benedetti, R.~Gurau, S.~Harribey and K.~Suzuki, \emph{{Hints of unitarity at
  large $N$ in the $O(N)^3$ tensor field theory}},
  \href{https://doi.org/10.1007/JHEP02(2020)072}{\emph{JHEP} {\bfseries 02}
  (2020) 072} [\href{https://arxiv.org/abs/1909.07767}{{\ttfamily
  1909.07767}}].

\bibitem{Benedetti:2019rja}
D.~Benedetti, N.~Delporte, S.~Harribey and R.~Sinha, \emph{{Sextic tensor field
  theories in rank $3$ and $5$}},
  \href{https://doi.org/10.1007/JHEP06(2020)065}{\emph{JHEP} {\bfseries 06}
  (2020) 065} [\href{https://arxiv.org/abs/1912.06641}{{\ttfamily
  1912.06641}}].

\bibitem{Lettera:2020uay}
D.~Lettera and A.~Vichi, \emph{{A large-$N$ tensor model with four
  supercharges}},  \href{https://arxiv.org/abs/2012.11600}{{\ttfamily
  2012.11600}}.

\bibitem{Benedetti:2020seh}
D.~Benedetti, \emph{{Melonic CFTs}},
  \href{https://doi.org/10.22323/1.376.0168}{\emph{PoS} {\bfseries CORFU2019}
  (2020) 168} [\href{https://arxiv.org/abs/2004.08616}{{\ttfamily
  2004.08616}}].

\bibitem{Gurau:2019qag}
R.~Gurau, \emph{Notes on tensor models and tensor field theories},  7, 2019.

\bibitem{Ambjorn:1990ge}
J.~Ambjorn, B.~Durhuus and T.~Jonsson, \emph{{Three-dimensional simplicial
  quantum gravity and generalized matrix models}},
  \href{https://doi.org/10.1142/S0217732391001184}{\emph{Mod. Phys. Lett.}
  {\bfseries A6} (1991) 1133}.

\bibitem{Sasakura:1990fs}
N.~Sasakura, \emph{{Tensor model for gravity and orientability of manifold}},
  \href{https://doi.org/10.1142/S0217732391003055}{\emph{Mod. Phys. Lett.}
  {\bfseries A6} (1991) 2613}.

\bibitem{Gurau:2009tw}
R.~Gurau, \emph{{Colored Group Field Theory}},
  \href{https://doi.org/10.1007/s00220-011-1226-9}{\emph{Commun.\ Math.\ Phys.}
  {\bfseries 304} (2011) 69} [\href{https://arxiv.org/abs/0907.2582}{{\ttfamily
  0907.2582}}].

\bibitem{Gurau:2010ba}
R.~Gurau, \emph{{The 1/N expansion of colored tensor models}},
  \href{https://doi.org/10.1007/s00023-011-0101-8}{\emph{Annales Henri
  Poincare} {\bfseries 12} (2011) 829}
  [\href{https://arxiv.org/abs/1011.2726}{{\ttfamily 1011.2726}}].

\bibitem{Gurau:2011xq}
R.~Gurau, \emph{{The complete 1/N expansion of colored tensor models in
  arbitrary dimension}},
  \href{https://doi.org/10.1007/s00023-011-0118-z}{\emph{Annales Henri
  Poincare} {\bfseries 13} (2012) 399}
  [\href{https://arxiv.org/abs/1102.5759}{{\ttfamily 1102.5759}}].

\bibitem{Lionni:2017yvi}
L.~Lionni, \emph{{Colored discrete spaces: higher dimensional combinatorial
  maps and quantum gravity}}, Ph.D. thesis, Saclay, 2017.
\newblock \href{https://arxiv.org/abs/1710.03663}{{\ttfamily 1710.03663}}.
\newblock 10.1007/978-3-319-96023-4.

\bibitem{Bonzom:2018btd}
V.~Bonzom, \emph{{Maximizing the number of edges in three-dimensional colored
  triangulations whose building blocks are balls}},
  \href{https://arxiv.org/abs/1802.06419}{{\ttfamily 1802.06419}}.

\bibitem{Bonzom:2012hw}
V.~Bonzom, R.~Gurau and V.~Rivasseau, \emph{{Random tensor models in the large
  N limit: Uncoloring the colored tensor models}},
  \href{https://doi.org/10.1103/PhysRevD.85.084037}{\emph{Phys.\ Rev.\ D}
  {\bfseries 85} (2012) 084037}
  [\href{https://arxiv.org/abs/1202.3637}{{\ttfamily 1202.3637}}].

\bibitem{Rivasseau:2013uca}
V.~Rivasseau, \emph{{The Tensor Track, {III}}},
  \href{https://doi.org/10.1002/prop.201300032}{\emph{Fortsch. Phys.}
  {\bfseries 62} (2014) 81} [\href{https://arxiv.org/abs/1311.1461}{{\ttfamily
  1311.1461}}].

\bibitem{Rivasseau:2014ima}
V.~Rivasseau, \emph{{The Tensor Theory Space}},
  \href{https://doi.org/10.1002/prop.201400057}{\emph{Fortsch. Phys.}
  {\bfseries 62} (2014) 835} [\href{https://arxiv.org/abs/1407.0284}{{\ttfamily
  1407.0284}}].

\bibitem{Eichhorn:2018phj}
A.~Eichhorn, T.~Koslowski and A.~D. Pereira, \emph{{Status of
  background-independent coarse-graining in tensor models for quantum
  gravity}}, \href{https://doi.org/10.3390/universe5020053}{\emph{Universe}
  {\bfseries 5} (2019) 53} [\href{https://arxiv.org/abs/1811.12909}{{\ttfamily
  1811.12909}}].

\bibitem{BenGeloun:2011rc}
J.~Ben~Geloun and V.~Rivasseau, \emph{{A renormalizable 4-Dimensional tensor
  field theory}},
  \href{https://doi.org/10.1007/s00220-012-1549-1}{\emph{Commun. Math. Phys.}
  {\bfseries 318} (2013) 69} [\href{https://arxiv.org/abs/1111.4997}{{\ttfamily
  1111.4997}}].

\bibitem{Samary:2012bw}
D.~O. Samary and F.~Vignes-Tourneret, \emph{{Just Renormalizable {TGFT}'s on
  {$U(1)^{d}$} with Gauge Invariance}},
  \href{https://doi.org/10.1007/s00220-014-1930-3}{\emph{Commun. Math. Phys.}
  {\bfseries 329} (2014) 545}
  [\href{https://arxiv.org/abs/1211.2618}{{\ttfamily 1211.2618}}].

\bibitem{Carrozza:2012uv}
S.~Carrozza, D.~Oriti and V.~Rivasseau, \emph{{Renormalization of Tensorial
  Group Field Theories: {Abelian} {$U(1)$} Models in Four Dimensions}},
  \href{https://doi.org/10.1007/s00220-014-1954-8}{\emph{Commun. Math. Phys.}
  {\bfseries 327} (2014) 603}
  [\href{https://arxiv.org/abs/1207.6734}{{\ttfamily 1207.6734}}].

\bibitem{Carrozza:2013wda}
S.~Carrozza, D.~Oriti and V.~Rivasseau, \emph{{Renormalization of a SU(2)
  Tensorial Group Field Theory in Three Dimensions}},
  \href{https://doi.org/10.1007/s00220-014-1928-x}{\emph{Commun. Math. Phys.}
  {\bfseries 330} (2014) 581}
  [\href{https://arxiv.org/abs/1303.6772}{{\ttfamily 1303.6772}}].

\bibitem{Krajewski:2016svb}
T.~Krajewski and R.~Toriumi, \emph{{Exact Renormalisation Group Equations and
  Loop Equations for Tensor Models}},
  \href{https://doi.org/10.3842/SIGMA.2016.068}{\emph{SIGMA} {\bfseries 12}
  (2016) 068} [\href{https://arxiv.org/abs/1603.00172}{{\ttfamily
  1603.00172}}].

\bibitem{Rivasseau:2017xbk}
V.~Rivasseau and F.~Vignes-Tourneret, \emph{{Constructive Tensor Field Theory:
  The ${T_{4}^{4}}$ Model}},
  \href{https://doi.org/10.1007/s00220-019-03369-9}{\emph{Commun. Math. Phys.}
  {\bfseries 366} (2019) 567}
  [\href{https://arxiv.org/abs/1703.06510}{{\ttfamily 1703.06510}}].

\bibitem{Dartois:2013he}
S.~Dartois, V.~Rivasseau and A.~Tanasa, \emph{{The $1/N$ expansion of
  multi-orientable random tensor models}},
  \href{https://doi.org/10.1007/s00023-013-0262-8}{\emph{Annales Henri
  Poincare} {\bfseries 15} (2014) 965}
  [\href{https://arxiv.org/abs/1301.1535}{{\ttfamily 1301.1535}}].

\bibitem{Carrozza:2015adg}
S.~Carrozza and A.~Tanasa, \emph{{$O(N)$} random tensor models},
  \href{https://doi.org/10.1007/s11005-016-0879-x}{\emph{Lett. Math. Phys.}
  {\bfseries 106} (2016) 1531}
  [\href{https://arxiv.org/abs/1512.06718}{{\ttfamily 1512.06718}}].

\bibitem{Ferrari:2017ryl}
F.~Ferrari, \emph{The large {D} limit of planar diagrams},
  \href{https://doi.org/10.4171/AIHPD/76}{\emph{Ann. Inst. Henri Poincar\'e
  Comb. Phys. Interact.} {\bfseries 6} (2019) 427}
  [\href{https://arxiv.org/abs/1701.01171}{{\ttfamily 1701.01171}}].

\bibitem{Ferrari:2017jgw}
F.~Ferrari, V.~Rivasseau and G.~Valette, \emph{{A New Large $N$ Expansion for
  General Matrix\textendash{}Tensor Models}},
  \href{https://doi.org/10.1007/s00220-019-03511-7}{\emph{Commun. Math. Phys.}
  {\bfseries 370} (2019) 403}
  [\href{https://arxiv.org/abs/1709.07366}{{\ttfamily 1709.07366}}].

\bibitem{Benedetti:2020iyz}
D.~Benedetti, S.~Carrozza, R.~Toriumi and G.~Valette, \emph{{Multiple scaling
  limits of $\mathrm{U}(N)^2 \times \mathrm{O}(D)$ multi-matrix models}},
  \href{https://arxiv.org/abs/2003.02100}{{\ttfamily 2003.02100}}.

\bibitem{Prakash:2019zia}
S.~Prakash and R.~Sinha, \emph{{Melonic Dominance in Subchromatic Sextic Tensor
  Models}}, \href{https://doi.org/10.1103/PhysRevD.101.126001}{\emph{Phys. Rev.
  D} {\bfseries 101} (2020) 126001}
  [\href{https://arxiv.org/abs/1908.07178}{{\ttfamily 1908.07178}}].

\bibitem{Klebanov:2017nlk}
I.~R. Klebanov and G.~Tarnopolsky, \emph{On large {$N$} limit of symmetric
  traceless tensor models},
  \href{https://doi.org/10.1007/JHEP10(2017)037}{\emph{JHEP} {\bfseries 10}
  (2017) 037} [\href{https://arxiv.org/abs/1706.00839}{{\ttfamily
  1706.00839}}].

\bibitem{Gurau:2017qya}
R.~Gurau, \emph{{The $1/N$ expansion of tensor models with two symmetric
  tensors}}, \href{https://doi.org/10.1007/s00220-017-3055-y}{\emph{Commun.
  Math. Phys.} {\bfseries 360} (2018) 985}
  [\href{https://arxiv.org/abs/1706.05328}{{\ttfamily 1706.05328}}].

\bibitem{Carrozza:2020eaz}
S.~Carrozza, F.~Ferrari, A.~Tanasa and G.~Valette, \emph{{On the large $D$
  expansion of Hermitian multi-matrix models}},
  \href{https://doi.org/10.1063/5.0008349}{\emph{J. Math. Phys.} {\bfseries 61}
  (2020) 073501} [\href{https://arxiv.org/abs/2003.04152}{{\ttfamily
  2003.04152}}].

\bibitem{Gurau:2009tz}
R.~Gurau, \emph{{Topological Graph Polynomials in Colored Group Field Theory}},
  \href{https://doi.org/10.1007/s00023-010-0035-6}{\emph{Annales Henri
  Poincare} {\bfseries 11} (2010) 565}
  [\href{https://arxiv.org/abs/0911.1945}{{\ttfamily 0911.1945}}].

\bibitem{Evnin:2020ddw}
O.~Evnin, \emph{{Melonic dominance and the largest eigenvalue of a large random
  tensor}},  \href{https://arxiv.org/abs/2003.11220}{{\ttfamily 2003.11220}}.

\bibitem{Gurau:2020ehg}
R.~Gurau, \emph{{On the generalization of the Wigner semicircle law to real
  symmetric tensors}},  \href{https://arxiv.org/abs/2004.02660}{{\ttfamily
  2004.02660}}.

\end{thebibliography}\endgroup

\addcontentsline{toc}{section}{References}


\end{document}